\pdfoutput=1
\documentclass{article}

\usepackage{hyperref}
\usepackage{xurl}
\hypersetup{
  colorlinks=true,
  linkcolor=blue!60!black,
  citecolor=blue!60!black,
  urlcolor=blue!60!black
}
\usepackage{amsmath}
\usepackage{mathtools}
\usepackage{amssymb}
\usepackage{amsthm}
\usepackage{dsfont}
\usepackage{stmaryrd}
\usepackage{xcolor}
\usepackage{listings}
\usepackage{booktabs}
\usepackage{longtable}
\usepackage{array}
\usepackage{placeins}
\usepackage{microtype}
\usepackage{fullpage}
\usepackage{authblk}
\usepackage{aliascnt}
\usepackage{tikz}
\usepackage{needspace}

\newtheorem{theorem}{Theorem}[section]

\newaliascnt{lemma}{theorem}
\newtheorem{lemma}[lemma]{Lemma}
\aliascntresetthe{lemma}

\newaliascnt{corollary}{theorem}

\aliascntresetthe{corollary}

\theoremstyle{definition}
\newaliascnt{definition}{theorem}
\newtheorem{definition}[definition]{Definition}
\aliascntresetthe{definition}

\newaliascnt{example}{theorem}
\newtheorem{example}[example]{Example}
\aliascntresetthe{example}

\newaliascnt{remark}{theorem}
\newtheorem{remark}[remark]{Remark}
\aliascntresetthe{remark}

\usepackage[nameinlink]{cleveref}
\crefname{theorem}{Theorem}{Theorems}
\Crefname{theorem}{Theorem}{Theorems}
\crefname{lemma}{Lemma}{Lemmas}
\Crefname{lemma}{Lemma}{Lemmas}
\crefname{corollary}{Corollary}{Corollaries}
\Crefname{corollary}{Corollary}{Corollaries}
\crefname{definition}{Definition}{Definitions}
\Crefname{definition}{Definition}{Definitions}
\crefname{example}{Example}{Examples}
\Crefname{example}{Example}{Examples}
\crefname{remark}{Remark}{Remarks}
\Crefname{remark}{Remark}{Remarks}

\newcommand{\R}{\mathbb{R}}

\newcommand{\FlagAlg}[1]{\mathcal{A}^{#1}}
\newcommand{\poshom}[1]{\mathrm{Hom}^+(\mathcal{A}^{#1}, \mathbb{R})}
\newcommand{\den}[2]{p(#1;\,#2)}
\newcommand{\tdensity}[2]{\pi(#1;\,#2)}
\DeclareRobustCommand{\lean}[1]{{\def\_{\textunderscore\allowbreak}\texttt{#1}}}
\newcommand{\emptyt}{\emptyset_{\mathrm{t}}}
\newcommand{\leansmul}{\mathbin{\vcenter{\hbox{$\scriptstyle\bullet$}}}}
\newcommand{\Ext}{\operatorname{Ext}}

\tikzset{
  graph vertex/.style={circle, draw, fill=white, inner sep=0pt, minimum size=4.5pt},
  graph root/.style={circle, draw, fill=black, inner sep=0pt, minimum size=4.5pt},
  graph edge/.style={line width=0.45pt},
  graph nonedge/.style={line width=0.45pt, dashed},
  inline/.style={scale=0.20, baseline=-0.5ex},
  fv/.style={circle, fill=white, draw, inner sep=0pt, minimum size=3pt},
  fr/.style={circle, fill=black, inner sep=0pt, minimum size=3pt}
}

\newcommand{\fKtwobull}{%
  \begin{tikzpicture}[inline]
    \draw (-1,0) -- (1,0);
    \node[fr] at (-1,0) {};
    \node[fv] at (1,0) {};
  \end{tikzpicture}%
}
\newcommand{\fbarKtwobull}{%
  \begin{tikzpicture}[inline]
    \node[fr] at (-1,0) {};
    \node[fv] at (1,0) {};
  \end{tikzpicture}%
}
\newcommand{\fbarKthree}{%
  \begin{tikzpicture}[inline]
    \node[fv] at (90:1) {};
    \node[fv] at (210:1) {};
    \node[fv] at (330:1) {};
  \end{tikzpicture}%
}
\newcommand{\fbarPthree}{%
  \begin{tikzpicture}[inline]
    \draw (210:1) -- (330:1);
    \node[fv] at (90:1) {};
    \node[fv] at (210:1) {};
    \node[fv] at (330:1) {};
  \end{tikzpicture}%
}
\newcommand{\fPthree}{%
  \begin{tikzpicture}[inline]
    \draw (210:1) -- (90:1) -- (330:1);
    \node[fv] at (90:1) {};
    \node[fv] at (210:1) {};
    \node[fv] at (330:1) {};
  \end{tikzpicture}%
}

\newcommand{\fbarKthreebull}{%
  \begin{tikzpicture}[inline]
    \node[fr] at (90:1) {};
    \node[fv] at (210:1) {};
    \node[fv] at (330:1) {};
  \end{tikzpicture}%
}
\newcommand{\fEthreebull}{%
  \begin{tikzpicture}[inline]
    \draw (210:1) -- (330:1);
    \node[fr] at (90:1) {};
    \node[fv] at (210:1) {};
    \node[fv] at (330:1) {};
  \end{tikzpicture}%
}
\newcommand{\fbarPthreebull}{%
  \begin{tikzpicture}[inline]
    \draw (90:1) -- (330:1);
    \node[fr] at (90:1) {};
    \node[fv] at (210:1) {};
    \node[fv] at (330:1) {};
  \end{tikzpicture}%
}
\newcommand{\fPthreebull}{%
  \begin{tikzpicture}[inline]
    \draw (90:1) -- (330:1) -- (210:1);
    \node[fr] at (90:1) {};
    \node[fv] at (210:1) {};
    \node[fv] at (330:1) {};
  \end{tikzpicture}%
}
\newcommand{\fKonetwobull}{%
  \begin{tikzpicture}[inline]
    \draw (90:1) -- (210:1);
    \draw (90:1) -- (330:1);
    \node[fr] at (90:1) {};
    \node[fv] at (210:1) {};
    \node[fv] at (330:1) {};
  \end{tikzpicture}%
}

\newcommand{\fCfivebull}{%
  \begin{tikzpicture}[inline]
    \draw (90:1) -- (18:1) -- (-54:1) -- (-126:1) -- (162:1) -- cycle;
    \node[fr] at (90:1) {};
    \node[fv] at (18:1) {};
    \node[fv] at (-54:1) {};
    \node[fv] at (-126:1) {};
    \node[fv] at (162:1) {};
  \end{tikzpicture}%
}

\lstdefinelanguage{Lean4}{
  keywords={def,theorem,lemma,instance,structure,class,import,open,namespace,
             end,where,fun,let,have,show,by,match,with,if,then,else,
             return,do,for,in,noncomputable,abbrev,variable,section,
             native_decide,decide,norm_num,simp,ring,linarith,omega,
             exact,apply,rw,calc,intro,constructor,use,refine,obtain,
             rcases,cases,induction,assumption,contradiction,trivial,rfl,
             macro,elab,notation,example,
             infix,infixl,infixr,prefix,postfix},
  keywordstyle=\color{blue!70!black}\bfseries,
  comment=[l]{--},
  morecomment=[l]{/--},
  commentstyle=\color{gray}\itshape,
  stringstyle=\color{orange!80!black},
  morestring=[b]",
  showstringspaces=false,
  sensitive=true,
  basicstyle=\ttfamily\small,
  breaklines=true,
  columns=fullflexible,
  escapeinside={(*@}{@*)},
  literate=
    {->}{{$\to$}}2
    {<-}{{$\leftarrow$}}2
    {<=}{{$\leq$}}1
    {>=}{{$\geq$}}1
    {alpha}{{$\alpha$}}1
    {forall}{{$\forall$}}1
    {exists}{{$\exists$}}1
}
\title{Formalizing Flag Algebras in Lean}
\date{}

\author[1,4]{Gyeongwon Jeong\thanks{These authors are joint first authors of this work.}}
\author[1,4]{Seonghun Park\protect\footnotemark[1]}
\author[1]{Jihoon Hyun}
\author[2,3]{Sang-il Oum}
\author[1,4]{Hongseok~Yang}
\affil[1]{\small School of Computing, KAIST, Daejeon, Korea}
\affil[2]{\small Discrete Mathematics Group, Institute for Basic Science (IBS), Daejeon, Korea}
\affil[3]{\small Department of Mathematical Sciences, KAIST, Daejeon, Korea}
\affil[4]{\small School of Computational Sciences, Korea Institute for Advanced Study (KIAS), Seoul,~Korea}
\affil[ ]{E-mail addresses: \texttt{jgyw0910@kaist.ac.kr},
\texttt{hun57@kaist.ac.kr},
\texttt{qawbecrdtey@kaist.ac.kr},
\texttt{sangil@ibs.re.kr},
\texttt{hongseokyang@kias.re.kr}}

\begin{document}
\maketitle

\begin{abstract}
Razborov's flag algebra method is one of the most successful tools for proving
asymptotic density inequalities in extremal graph theory.  It reduces an
inequality about all large graphs avoiding a forbidden pattern to finitely many
algebraic conditions, and a certificate satisfying them can be found by
semidefinite programming.  We formalize the method for finite simple graphs
in Lean~4, building on Mathlib.  The formalization covers flags and their
densities, the flag algebra as a quotient vector space, positive
homomorphisms as limits of graph sequences, Razborov's random extensions, and
the downward operator that averages out labels.  On top of this theory, we build
verified computation of flag densities and a tactic that turns certificates
found by semidefinite programming into Lean proofs.  The certificate serves
only as candidate data; Lean checks every fact the proof uses.

With this development we verify seven Tur\'an-type upper bounds.  Adding
lower-bound constructions proved in Lean yields the asymptotic form of Mantel's
theorem and the asymptotic Erd\H{o}s pentagon theorem
$\tdensity{C_5}{K_3}=24/625$; to our knowledge, the latter has not previously
been formalized in a proof assistant.  We also prove two inequalities of
Goodman.  All density theorems reported here depend only on the standard
axioms of Lean's logic.

Formalization also led us to treat forbidden subgraphs with a different
semantics.  Instead of building a separate flag algebra for each forbidden
family, we keep a single algebra of all graphs and impose the constraint only
when interpreting an inequality, through randomly labeled limits of graphs
satisfying it.  We prove that this ensemble semantics yields the intended density
bounds, and that a condition we call root-plantability characterizes when it
agrees with the standard semantics.
\end{abstract}

\section{Introduction}
\label{sec:intro}

Many central questions in extremal graph theory ask how often one finite
pattern can occur in a large graph when another pattern is forbidden.  Let $F$
be a finite simple graph and $\mathcal{H}$ a collection of finite simple
graphs.  We write
\[
  \mathrm{ex}(n,\,F;\,\mathcal{H})
  :=
  \max\bigl\{
    \text{number of induced copies of }F\text{ in }G
    \mid |V(G)|=n,\; G\text{ is }\mathcal{H}\text{-free}
  \bigr\},
\]
where an \emph{induced copy} of $F$ in $G$ is a vertex subset
$S\subseteq V(G)$ with $|S|=|V(F)|$ whose induced subgraph $G[S]$ is
isomorphic to $F$, and \emph{\(\mathcal{H}\)-free} means that $G$ contains no
copy of any member of $\mathcal{H}$ as a (not necessarily induced) subgraph.
Normalizing by the number $\binom{n}{|V(F)|}$ of $|V(F)|$-vertex subsets and
letting $n$ grow gives the \emph{Tur\'an density} of $F$ with respect to
$\mathcal{H}$:
\[
  \tdensity{F}{\mathcal{H}}
  :=
  \lim_{n\to\infty}
  \frac{\mathrm{ex}(n,\,F;\,\mathcal{H})}{\binom{n}{|V(F)|}},
\]
where the limit exists because a standard averaging argument shows that the
sequence $\mathrm{ex}(n,\,F;\,\mathcal{H})/\binom{n}{|V(F)|}$ is
non-increasing for $n\geq|V(F)|$.
Whenever $\mathcal{H}=\{H\}$ is a singleton, we write $H$ in place of
$\mathcal{H}$.

The oldest example concerns graphs without triangles: what is the largest
fraction of vertex pairs in such a graph that can be edges?  Since the induced
copies of $K_2$ in a graph are exactly its edges, the answer is the Tur\'an density
$\tdensity{K_2}{K_3}$, and Mantel's theorem~\cite{mantel} states that
\[
  \tdensity{K_2}{K_3}
  = \lim_{n\to\infty}
  \max_{\substack{|V(G)|=n\\G\text{ triangle-free}}}
  \frac{|E(G)|}{\binom{n}{2}}
  = \frac{1}{2}.
\]

The two directions of such an equality have very different proof burdens.  For
the lower bound, one explicit construction suffices: for Mantel's theorem,
balanced complete bipartite graphs contain no triangles and have edge densities
approaching $1/2$.  The upper bound is the
universal part: it must show that no sequence of increasingly large
\(\mathcal{H}\)-free graphs can asymptotically exceed the claimed value.  No
finite collection of examples establishes such a claim, and the extremal
graphs may change unpredictably as the number of vertices grows.

Razborov's flag algebra method~\cite{razborov2007flag} provides a proof
technique for such universal claims: it reduces the task of proving an
asymptotic upper bound for an infinite class of graphs to checking finitely many
algebraic conditions on small, partially labeled graphs, called
\emph{flags}.  A semidefinite programming (SDP) solver, run through software
such as Flagmatic~\cite{vaughan2020flagmatic}, proposes a finite package of
data intended to satisfy these conditions, consisting of flags, rational
matrices, and coefficients; we call this package a \emph{certificate}.  The
flag algebra framework explains why a certificate that passes the checks establishes the desired bound for the whole class.

This combination of abstraction and computer search has been unusually
effective.  It proved the Erd\H{o}s pentagon theorem on the maximum asymptotic
density of five-cycles in triangle-free graphs~\cite{grzesik2012,hatami2012},
determined the minimum asymptotic triangle density of graphs with a given edge
density~\cite{razborov2008triangles} and the maximum induced $C_5$
density~\cite{balogh2016induced5}, and settled the rainbow-triangle
problem~\cite{balogh2017rainbow}.  It has also given upper bounds for
Tur\'an's tetrahedron
problem~\cite{razborov2010tetrahedron,jeong2026tetrahedron} and determined
several Tur\'an densities of 3-uniform hypergraphs~\cite{baber2012turan}.

The certificates behind such results are finite but not small.  A nontrivial
certificate combines many flags and relies on large tables of density values
for small graphs; for the Erd\H{o}s pentagon problem, these tables contain
thousands of entries.  Such computer-generated certificates are impractical to
check by hand, which makes flag-algebra results natural candidates for formal
verification.

\paragraph{A Lean formalization of flag algebras.}
We formalize the flag algebra method for finite simple graphs in
Lean~4~\cite{moura2021lean}, building on Mathlib~\cite{mathlib2020}, and use
it to turn flag-algebra certificates into machine-checked proofs of asymptotic
density bounds.  The development has three parts.

The \emph{specification layer} (\Cref{sec:abstract}) formalizes the theory as
it appears on paper: flags as isomorphism classes of partially labeled graphs,
their densities, the flag algebra as a quotient vector space,
positive homomorphisms together with Razborov's theorem identifying them with
limits of convergent graph sequences, random extensions, and the downward
operator.  Flag-algebra arguments can therefore be stated in Lean much as they
are written by hand.

The \emph{reflection layer} (\Cref{sec:reflection}) makes the finite parts of
the theory computable.  It supplies concrete graph representations on which
densities and downward coefficients can be evaluated, together with adequacy
theorems proving that these evaluations agree with the abstract definitions.
Each of the thousands of finite identities that a certificate needs is then
proved by evaluation rather than by hand.

On top of these two layers, a tactic (\Cref{sec:flagmatic}) turns a Flagmatic
certificate into a Lean proof of the density bound it certifies.  The
certificate serves only as candidate data: Lean checks every fact the proof
uses, so an erroneous certificate makes the proof fail instead of yielding a
false theorem.

\paragraph{Forbidden subgraphs and the ensemble semantics.}
Every bound verified in this paper concerns graphs that avoid a specified
subgraph $H$.  Razborov's standard presentation builds this constraint into
the flag algebra, constructing a separate algebra for each forbidden family.
To reuse a single development across problems, we instead keep one algebra of
all finite graphs and impose $H$-freeness only when interpreting an
inequality, through randomly labeled limits of $H$-free graphs.  We call this
the \emph{ensemble semantics} and prove that it yields the intended density
bounds.  \Cref{sec:metatheory} shows that every inequality valid in the
standard semantics is valid in the ensemble semantics, and characterizes,
through a condition we call \emph{root-plantability}, when the converse holds.
A companion paper, in preparation, will develop this comparison in full.

\paragraph{Contributions.}
We make the following contributions.
\begin{itemize}
\item A formalization of the theory of flag algebras for finite simple graphs,
  and a tactic built on it that uses verified computation to turn flag-algebra
  certificates into machine-checked proofs.
\item Seven Tur\'an-type upper bounds verified from Flagmatic certificates.
  For Mantel's theorem and the Erd\H{o}s pentagon theorem, we also prove
  matching lower-bound constructions in Lean, which yields the asymptotic
  forms $\tdensity{K_2}{K_3}=1/2$ and $\tdensity{C_5}{K_3}=24/625$; to our
  knowledge, the latter is formalized here for the first time.
\item The ensemble semantics for forbidden subgraphs, and the root-plantability
  criterion comparing it with the standard semantics, whose Lean development
  was autoformalized.
\end{itemize}
The complete Lean development is publicly available as release
\texttt{v1.0} of the following repository:
\begin{center}
  \url{https://github.com/taeyool/lean-flag-algebras-release}
\end{center}
\Cref{sec:artifact} gives build instructions and a correspondence between
the paper's mathematical statements and Lean declarations.

\paragraph{Organization.}
\Cref{sec:background} reviews flag algebras and fixes notation.
\Cref{sec:abstract,sec:reflection} present the specification and reflection
layers, respectively.  \Cref{sec:flagmatic} presents the certificate-checking
tactic and the seven verified upper bounds.
\Cref{sec:lowerbounds} proves matching lower bounds for Mantel's and the
Erd\H{o}s pentagon theorems, and Goodman's inequalities.
\Cref{sec:obstacles} discusses engineering obstacles, and
\Cref{sec:metatheory} compares the ensemble and standard semantics.
\Cref{sec:unsure} discusses open design questions,
\Cref{sec:related} surveys related work, and \Cref{sec:conclusion} concludes.
\Cref{sec:compile-times,sec:artifact} report compilation times and describe the
artifact, respectively.

\section{Background: Flag Algebras}
\label{sec:background}

This section recalls the mathematical definitions underlying flag algebras,
following Razborov~\cite{razborov2007flag}, and establishes the notation used
later in the paper.  No background in flag algebras or extremal graph theory
is assumed.

\paragraph{Notation.}
For a non-negative integer $n$, we write $[n]$ for the set of positive integers
at most $n$.  We consider only undirected, finite, and simple graphs; for a
graph $G$, its vertex and edge sets are denoted by $V(G)$ and $E(G)$.  We use
the standard graph names: $K_n$ is the complete graph on $n$ vertices, $P_n$
the path and $C_n$ the cycle on $n$ vertices, $\overline{G}$ the complement of
$G$, and $K_{a,b}$ the complete bipartite graph with parts of sizes $a$
and~$b$.  We fix a (possibly empty) set $\mathcal{H}$ of forbidden
graphs; a graph is \emph{\(\mathcal{H}\)-free} if it contains no subgraph
isomorphic to any member of $\mathcal{H}$.  All flag types, flags, and flag
algebras defined below are relative to this fixed $\mathcal{H}$.

\subsection{Flag Types and Flags}
\label{sec:background-flags}

\paragraph{Flag Types.}
A flag type specifies the labeled part shared by all flags of that flag type.
Formally, a \emph{flag type} of size $k$ is an \(\mathcal{H}\)-free graph
\(\sigma\) with vertex set $[k]$.  The \emph{empty flag type} $\emptyset$, of
size zero, is the unique graph on the empty vertex set.
We use the term \emph{flag type} rather than Razborov's original \emph{type}
to avoid confusion with Lean's type system in later sections.

\paragraph{Flags.}
For a graph $G$ and a set $S\subseteq V(G)$, let $G[S]$ denote the
\emph{induced subgraph} of $G$ on $S$: its vertex set is $S$, and its edges are
exactly the edges of $G$ with both endpoints in $S$.
Fix a flag type $\sigma$ of size $k$.
A \emph{$\sigma$-flag} is a pair $F = (M, \theta_F)$ consisting of an
\(\mathcal{H}\)-free finite graph $M$ and an injective map
$\theta_F : [k] \hookrightarrow V(M)$ that induces a graph isomorphism from
$\sigma$ to $M[\operatorname{Im}(\theta_F)]$.  The \emph{size} of $F$ is
$|V(M)|$, and we write $V(F)$ for $V(M)$.
For each $i\in[k]$, the vertex $\theta_F(i)$ carries label $i$, so the vertices
in~$\operatorname{Im}(\theta_F)$ are the \emph{labeled vertices} of $F$, also
called its \emph{roots}; all remaining vertices are \emph{unlabeled}.
For $S \subseteq V(F)$ with $\operatorname{Im}(\theta_F) \subseteq S$,
we write $F[S]$ for the $\sigma$-flag $(M[S], \theta_F)$, the
\emph{$\sigma$-flag induced by $F$ on $S$}.
Flags over the empty flag type $\emptyset$ are ordinary
\(\mathcal{H}\)-free graphs with no labeled vertices.
\Cref{fig:flag-examples} illustrates several flags of different flag types.

Two $\sigma$-flags $(M_1,\theta_1)$ and $(M_2,\theta_2)$ are
\emph{isomorphic}, written $(M_1,\theta_1) \cong_\sigma (M_2,\theta_2)$, if
some graph isomorphism $\alpha: V(M_1)\to V(M_2)$ satisfies
$\alpha \circ \theta_1 = \theta_2$, an equation saying that $\alpha$
preserves every label.  We write $\mathcal{F}^\sigma_n$ for the set of
isomorphism classes of $\sigma$-flags of size $n$, and
$\mathcal{F}^\sigma = \bigcup_{n \geq k} \mathcal{F}^\sigma_n$.

\begin{figure}[t]
  \centering
  \setlength{\tabcolsep}{1.2em}
  \renewcommand{\arraystretch}{1.3}
  \begin{tabular}{cccccc}
    \begin{tikzpicture}[scale=0.8, baseline=-0.5ex]
      \node[graph vertex] (a) at (0,0) {};
      \node[graph vertex] (b) at (0.8,0) {};
      \draw[graph edge] (a) -- (b);
    \end{tikzpicture}
    &
    \begin{tikzpicture}[scale=0.8, baseline=-0.5ex]
      \node[graph vertex] (a) at (90:0.65) {};
      \node[graph vertex] (b) at (210:0.65) {};
      \node[graph vertex] (c) at (330:0.65) {};
    \end{tikzpicture}
    &
    \begin{tikzpicture}[scale=0.8, baseline=-0.5ex]
      \node[graph vertex] (a) at (90:0.65) {};
      \node[graph vertex] (b) at (210:0.65) {};
      \node[graph vertex] (c) at (330:0.65) {};
      \draw[graph edge] (b) -- (c);
    \end{tikzpicture}
    &
    \begin{tikzpicture}[scale=0.8, baseline=-0.5ex]
      \node[graph vertex] (a) at (90:0.65) {};
      \node[graph vertex] (b) at (210:0.65) {};
      \node[graph vertex] (c) at (330:0.65) {};
      \draw[graph edge] (a) -- (b);
      \draw[graph edge] (a) -- (c);
    \end{tikzpicture}
    &
    \begin{tikzpicture}[scale=0.8, baseline=-0.5ex]
      \node[graph vertex] (a) at (90:0.65) {};
      \node[graph vertex] (b) at (210:0.65) {};
      \node[graph vertex] (c) at (330:0.65) {};
      \draw[graph edge] (a) -- (b);
      \draw[graph edge] (a) -- (c);
      \draw[graph edge] (b) -- (c);
    \end{tikzpicture}
    &
    \begin{tikzpicture}[scale=0.7, baseline=-0.5ex]
      \node[graph vertex] (v1) at (90:0.65) {};
      \node[graph vertex] (v2) at (18:0.65) {};
      \node[graph vertex] (v3) at (-54:0.65) {};
      \node[graph vertex] (v4) at (-126:0.65) {};
      \node[graph vertex] (v5) at (162:0.65) {};
      \draw[graph edge] (v1) -- (v2) -- (v3) -- (v4) -- (v5) -- (v1);
    \end{tikzpicture}
    \\
    $K_2$ & $\overline{K_3}$ & $\overline{P_3}$ & $P_3$ & $K_3$ & $C_5$ \\[0.8em]
    \begin{tikzpicture}[scale=0.8, baseline=-0.5ex]
      \node[graph root] (a) at (0,0) {};
      \node[graph vertex] (b) at (0.8,0) {};
      \draw[graph edge] (a) -- (b);
    \end{tikzpicture}
    &
    \begin{tikzpicture}[scale=0.8, baseline=-0.5ex]
      \node[graph root] (a) at (0,0) {};
      \node[graph vertex] (b) at (0.8,0) {};
    \end{tikzpicture}
    &
    \begin{tikzpicture}[scale=0.8, baseline=-0.5ex]
      \node[graph vertex] (a) at (90:0.65) {};
      \node[graph root] (b) at (210:0.65) {};
      \node[graph vertex] (c) at (330:0.65) {};
      \draw[graph edge] (b) -- (c);
    \end{tikzpicture}
    &
    \begin{tikzpicture}[scale=0.8, baseline=-0.5ex]
      \node[graph vertex] (a) at (90:0.65) {};
      \node[graph root] (b) at (210:0.65) {};
      \node[graph vertex] (c) at (330:0.65) {};
      \draw[graph edge] (b) -- (a);
      \draw[graph edge] (a) -- (c);
    \end{tikzpicture}
    &
    \begin{tikzpicture}[scale=0.8, baseline=-0.5ex]
      \node[graph vertex] (a) at (90:0.65) {};
      \node[graph root] (b) at (210:0.65) {};
      \node[graph vertex] (c) at (330:0.65) {};
      \draw[graph edge] (b) -- (a);
      \draw[graph edge] (b) -- (c);
    \end{tikzpicture}
    &
    \begin{tikzpicture}[scale=0.7, baseline=-0.5ex]
      \node[graph root] (v1) at (90:0.65) {};
      \node[graph vertex] (v2) at (18:0.65) {};
      \node[graph vertex] (v3) at (-54:0.65) {};
      \node[graph vertex] (v4) at (-126:0.65) {};
      \node[graph vertex] (v5) at (162:0.65) {};
      \draw[graph edge] (v1) -- (v2) -- (v3) -- (v4) -- (v5) -- (v1);
    \end{tikzpicture}
    \\
    $K_2^\bullet$ & $\overline{K_2}^\bullet$ & $\overline{P_3}^\bullet$ & $P_3^\bullet$
    & ${P_3'}^\bullet$ & $C_5^\bullet$ \\
  \end{tabular}
  \caption{Examples of flags when no graphs are forbidden
  ($\mathcal{H}=\varnothing$).  \emph{Top row}: $\emptyset$-flags, equivalently
  ordinary graphs, of sizes 2, 3, and 5.  \emph{Bottom row}: $\sigma$-flags of the
  unique one-vertex flag type $\sigma$.  The filled vertex carries label~1,
  and the open vertices are unlabeled.  The $\sigma$-flags $P_3^\bullet$ and
  ${P_3'}^\bullet$ have the same underlying path; $P_3^\bullet$ labels an
  endpoint, whereas ${P_3'}^\bullet$ labels the middle vertex.}
  \label{fig:flag-examples}
\end{figure}
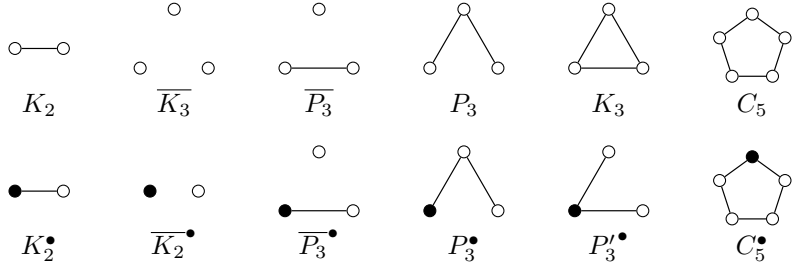

\subsection{Subflag Densities}
\label{sec:background-densities}

In a density calculation, flags play two roles.  The \emph{pattern flag} is the
small local configuration whose occurrence is being measured, while the
\emph{host flag} is the larger flag in which we sample vertices.  These are
roles in the calculation, not different kinds of flags.  We use \(F\), or
\(F_i\), for pattern flags, \(G\) for host flags, and \(G'\) for an intermediate
host flag when a third flag variable is needed.

Let $\sigma$ be a flag type of size $k$.

\paragraph{Single-flag density.}
For a $\sigma$-flag $F$ of size $m$
(the pattern $\sigma$-flag) and a $\sigma$-flag $G = (M, \theta_G)$ of size $n \geq m$
(the host $\sigma$-flag), the \emph{density of $F$ in $G$}, written $\den{F}{G}$, is
the probability that a uniformly random $(m-k)$-element subset of
$V(G) \setminus \operatorname{Im}(\theta_G)$, together with the $k$ labeled
vertices of $G$, induces a $\sigma$-flag \(F'\) with \(F' \cong_\sigma F\).
Explicitly:
\[
  \den{F}{G}
  \;:=\;
  \frac{1}{\binom{n-k}{m-k}}
    \cdot {\Bigl|\Big\{S \subseteq V(G)\setminus \operatorname{Im}(\theta_G)
                  ~\Big|~ |S|=m-k,\;
                  G[S \cup \operatorname{Im}(\theta_G)]
                  \cong_\sigma F \Big\}\Bigr|}.
\]
This value lies in $[0,1]$ and is invariant under
isomorphism of both $F$ and $G$.
When $n < m$, we set $\den{F}{G} := 0$ by convention.

\begin{example}[Subflag densities in $C_5^\bullet$]
  Among the four unlabeled vertices of
  $\fCfivebull$, exactly two are adjacent to the labeled one.  Hence
  \[
    \den{\fKtwobull}{\fCfivebull} = \frac{2}{4} = \frac12,
    \qquad
    \den{\fbarKtwobull}{\fCfivebull} = \frac{2}{4} = \frac12.
  \]
\end{example}

\paragraph{Chain rule.}
The density $\den{F}{G}$ can be computed in two equivalent ways: by sampling
$m-k$ unlabeled vertices directly from $G$, or by first sampling an
intermediate $\sigma$-flag $G'$ of some size $n'$ with $m \leq n' \leq n$ inside $G$ and
then sampling~$F$ inside $G'$.
These procedures describe the same sampling experiment, yielding the
single-flag case of Razborov's chain rule~\cite[Lemma~2.2]{razborov2007flag}.

\begin{lemma}[Chain rule]\label{lem:chain-rule}
  Let $\sigma$ be a flag type of size $k$, let $m,n,n' \in \mathbb{N}$ satisfy
  $k \leq m \leq n' \leq n$, and let $F \in \mathcal{F}^\sigma_m$ and
  $G \in \mathcal{F}^\sigma_n$.  Then
  \[
    \den{F}{G}
    \;=\;
    \sum_{G' \in \mathcal{F}^\sigma_{n'}} \den{F}{G'}\,\den{G'}{G}.
  \]
\end{lemma}

\paragraph{Multi-flag density.}
For pattern $\sigma$-flags $F_1, \ldots, F_t$ of sizes $m_1, \ldots, m_t$ and
a host $\sigma$-flag $G = (M, \theta_G)$ of size $n$, we say $F_1, \ldots, F_t$
\emph{fit} in $G$ if
\[
  n - k \;\geq\; (m_1 - k) + \cdots + (m_t - k).
\]
When $F_1, \ldots, F_t$ fit in $G$, let $(S_1, \ldots, S_t)$ be a uniformly
random tuple of pairwise-disjoint subsets $S_i \subseteq V(G) \setminus
\operatorname{Im}(\theta_G)$ with $|S_i| = m_i - k$.
The \emph{density of $F_1, \ldots, F_t$ in $G$}, written $\den{F_1,\ldots,F_t}{G}$, is the
probability that $G[S_i \cup \operatorname{Im}(\theta_G)] \cong_\sigma F_i$
for every $i \in [t]$.
When $F_1, \ldots, F_t$ do not fit in $G$, we set
$\den{F_1,\ldots,F_t}{G} := 0$ by convention.

\subsection{The Flag Algebra}

Fix a flag type $\sigma$ of size $k$.
Informally, the flag algebra $\mathcal{A}^\sigma$ is an algebra\footnote{Here
``algebra'' means a vector space whose elements can also be multiplied.  In
this paper, the scalar coefficients are real numbers, multiplication is
bilinear, associative, and commutative, and there is a multiplicative
unit.} of density expressions generated
by $\sigma$-flags.  Each $\sigma$-flag is a symbolic building block: it stands for the
density of that local labeled pattern in a large host graph.  The quotient
below records the identity of Lemma~\ref{lem:chain-rule} among these building
blocks.
Formally, let
$\mathbb{R}[\mathcal{F}^\sigma]$ be the free real vector space on
$\mathcal{F}^\sigma$, that is, the vector space of finite formal real linear
combinations of $\sigma$-flags.  For a $\sigma$-flag $F$, write $\mathbf e_F$
for its corresponding basis vector.
Define the \emph{zero space} $\mathcal{Z}^\sigma$ as the subspace generated
by all elements of the form
\begin{equation}
  \label{eq:zero-element}
  \mathbf e_F - \sum_{G \in \mathcal{F}^\sigma_n} \den{F}{G} \cdot \mathbf e_G,
  \qquad \text{for}\ F \in \mathcal{F}^\sigma_m \ \text{and}\ n \geq m.
\end{equation}
The \emph{flag algebra} is the quotient vector space
\[
  \mathcal{A}^\sigma \;:=\; \mathbb{R}[\mathcal{F}^\sigma] \,/\, \mathcal{Z}^\sigma.
\]
We write $[F]$ for the class of $\mathbf e_F$ in $\mathcal{A}^\sigma$.
Throughout the paper, capital letters $F,G$ denote finite flags, while
lowercase letters $f,g,h$ denote arbitrary flag-algebra elements.
For familiar named or pictorial flags, such as $K_2$ or
$\fKtwobull$, we retain the standard shorthand of omitting the brackets in
algebraic formulas; there the surrounding algebraic context makes the
canonical image unambiguous.

We call $\sum_{G \in \mathcal{F}^\sigma_n} \den{F}{G} \cdot \mathbf e_G$ the
\emph{expansion of $F$ at size $n$}.
By Lemma~\ref{lem:chain-rule}, a $\sigma$-flag $F \in \mathcal{F}^\sigma_m$ and
its expansion at any size $n \geq m$ have the same density in every host
$\sigma$-flag of size $\geq n$.
The quotient by $\mathcal{Z}^\sigma$ records this expansion identity as the
equality
\begin{equation}
  \label{eq:expansion-identity}
  [F]
  =
  \sum_{G \in \mathcal{F}^\sigma_n} \den{F}{G}\,[G]
  \qquad\text{in }\mathcal{A}^\sigma .
\end{equation}
We may therefore freely replace $[F]$ by its expansion at any size $n \geq m$.

\begin{example}[$K_2$ expanded at size $3$ in the triangle-free algebra]
  Working in the triangle-free flag algebra (i.e., $\mathcal{H} = \{K_3\}$),
  the $\emptyset$-flags of size $3$ are the three triangle-free graphs on three vertices:
  $\fbarKthree$, $\fbarPthree$, and $\fPthree$.
  The density of an edge in each of them is
  \[
    \den{K_2}{\fbarKthree}=0,\quad
    \den{K_2}{\fbarPthree}=\frac13,\quad
    \den{K_2}{\fPthree}=\frac23,
  \]
  so the zero-space quotient identifies
  \begin{equation}
    K_2
    \;=\;
    \frac13\,\fbarPthree
    + \frac23\,\fPthree
    \qquad\text{in } \mathcal{A}^{\emptyset}.
    \label{eq:k2-expansion}
  \end{equation}
\end{example}

\paragraph{Multiplication.}
For $m_1,m_2 \in \mathbb{N}$ and $\sigma$-flags
$F_1 \in \mathcal{F}^\sigma_{m_1}$ and $F_2 \in \mathcal{F}^\sigma_{m_2}$,
choose an auxiliary size $\ell \in \mathbb{N}$ with
$\ell \geq m_1 + m_2 - k$ and define
\begin{equation}
  \label{eq:mul}
  [F_1] \cdot [F_2]
  \;:=\;
  \sum_{G \in \mathcal{F}^\sigma_\ell} \den{F_1,F_2}{G} \cdot [G]
  \qquad\text{in } \mathcal{A}^\sigma.
\end{equation}
At first glance, this definition appears to depend on both the choice of $\ell$
and the choice of representatives for~$[F_1]$ and $[F_2]$; the zero-space
quotient ensures that neither matters.  Extending bilinearly to all of
$\mathcal{A}^\sigma$ gives a commutative algebra over the real numbers.  Its
unit is the basis element represented by the $\sigma$-flag
$(\sigma, \operatorname{id}_{[k]})$ of size~$k$, whose underlying graph
is $\sigma$ itself with every vertex labeled; we write it as
$\mathbf{1}_\sigma$, and simply as $\mathbf{1}$ when
$\sigma=\emptyset$.

\begin{example}[Multiplication of $K_2^\bullet$ and $\overline{K_2}^\bullet$]
\label{ex:typed-edge-nonedge-product}
  Let $\sigma$ be the one-vertex flag type.  Taking $\ell = 3$, the density
  $\den{\fKtwobull,\,\fbarKtwobull}{G}$ equals $\frac{1}{2}$ for each size-3
  $\sigma$-flag $G$ whose labeled vertex has degree exactly $1$:
  \[
    \den{\fKtwobull,\,\fbarKtwobull}{\fbarPthreebull}
    \;=\;
    \den{\fKtwobull,\,\fbarKtwobull}{\fPthreebull}
    \;=\; \frac{1}{2},
  \]
  and vanishes on all other size-3 $\sigma$-flags:
  \[
    \den{\fKtwobull,\,\fbarKtwobull}{G} = 0
    \qquad\text{for all } G \in \mathcal{F}_3^\sigma \setminus \{\fbarPthreebull,\,\fPthreebull\}.
  \]
  Hence
  \[
    \fKtwobull \cdot \fbarKtwobull
    \;=\;
    \frac{1}{2}\,\fbarPthreebull + \frac{1}{2}\,\fPthreebull
    \qquad\text{in } \mathcal{A}^\sigma.
  \]
\end{example}

\subsection{Positive Homomorphisms and Semantic Order}
\label{sec:background-semantic-order}

Fix a flag type $\sigma$.
The flag algebra $\mathcal{A}^\sigma$ provides algebraic syntax for expressing
densities of pattern $\sigma$-flags.  Flag algebra proofs, however, target
\emph{asymptotic} statements: rather than a single finite host $\sigma$-flag, we
consider increasing sequences~$\{G_n\}_{n\in\mathbb{N}}$ of host
$\sigma$-flags with $|V(G_n)|\to\infty$ and ask what happens to $\sigma$-flag
densities in the limit.  Positive homomorphisms can be thought of as the
limiting density profiles of convergent sequences of host $\sigma$-flags.
The semantic non-negativity cone
$\mathcal{C}^\sigma$ then consists of the flag-algebra elements
whose value under every positive homomorphism is non-negative, and establishing
membership in $\mathcal{C}^\sigma$ is precisely what flag algebra proofs aim
to do.

\paragraph{Convergent sequences.}
Call a sequence $\{G_n\}_{n\in\mathbb{N}}$ of host $\sigma$-flags \emph{increasing} if
$|V(G_1)| < |V(G_2)| < \cdots$.  Each host $\sigma$-flag $G$ defines a point
$p^G \in [0,1]^{\mathcal{F}^\sigma}$ by $p^G(F) := \den{F}{G}$.
An increasing sequence is \emph{convergent} if the sequence of points $p^{G_n}$
converges in $[0,1]^{\mathcal{F}^\sigma}$ endowed with the product topology;
equivalently, if $\lim_{n\to\infty}\den{F}{G_n}$ exists for every $\sigma$-flag $F$.

\paragraph{Positive homomorphisms.}
A \emph{positive homomorphism} on
$\mathcal{A}^\sigma$ is a map
$\phi:\mathcal{A}^\sigma \to \mathbb{R}$ that preserves addition,
multiplication, multiplication by real scalars, and the unit.  It must also
satisfy $\phi([F]) \geq 0$ for every $\sigma$-flag $F$.
Following standard flag-algebra terminology, we call such a homomorphism
\emph{positive}, although its defining inequalities are non-strict.
We write $\poshom{\sigma}$ for the set of all positive homomorphisms on
$\mathcal{A}^\sigma$.

In $\mathcal{A}^\sigma$, the following identity holds for every
$\ell \in \mathbb{N}$ with $\ell \geq |V(\sigma)|$:
\begin{equation}
  \sum_{F \in \mathcal{F}^\sigma_\ell} [F] \;=\; \mathbf{1}_\sigma.
  \label{eq:sum-to-one}
\end{equation}
Applying any $\phi \in \poshom{\sigma}$ to both sides gives
$\sum_{F \in \mathcal{F}^\sigma_\ell} \phi([F]) = 1$, and non-negativity then
forces $\phi([F]) \in [0,1]$ for every $\sigma$-flag $F$.  Each
$\phi \in \poshom{\sigma}$ therefore determines the point
$F \mapsto \phi([F])$ of $[0,1]^{\mathcal{F}^\sigma}$, identifying
$\poshom{\sigma}$ with a subset of that space.

The following theorem makes precise the sense in which positive homomorphisms
are exactly the limiting density profiles of convergent sequences of
$\sigma$-flags~\cite{razborov2007flag}.

\begin{theorem}[Lov\'asz--Szegedy~\cite{lovasz2006limits}; see~{\cite[Theorem~3.3]{razborov2007flag}}]\label{thm:convergent-hom}
  Let $\sigma$ be a flag type.
  \begin{enumerate}
  \item[\textup{(a)}] If $\{G_n\}_{n\in\mathbb{N}}$ is a convergent sequence of
    $\sigma$-flags, then
    \[
      \lim_{n\to\infty} p^{G_n} \in \poshom{\sigma}.
    \]
  \item[\textup{(b)}] Conversely, every $\phi \in \poshom{\sigma}$ satisfies
    $\phi = \lim_{n\to\infty} p^{G_n}$ for some convergent sequence
    $\{G_n\}_{n\in\mathbb{N}}$ of $\sigma$-flags.
  \end{enumerate}
\end{theorem}

\paragraph{Semantic non-negativity cone and order.}
For a fixed flag type $\sigma$, define the \emph{semantic non-negativity cone}
\[
  \mathcal{C}^\sigma
  :=
  \{f \in \mathcal{A}^\sigma
    \mid \phi(f) \geq 0
    \text{ for every } \phi \in \poshom{\sigma}\}.
\]
We write
\[
  f \leq_{\mathcal{H},\sigma} g
  \qquad\text{when}\qquad
  g-f \in \mathcal{C}^\sigma,
\]
where the subscript $\mathcal{H}$ records that the comparison is made in the
\(\mathcal{H}\)-free world: both $\mathcal{A}^\sigma$ and $\mathcal{C}^\sigma$ are
built from \(\mathcal{H}\)-free graphs.
Equivalently, $f\leq_{\mathcal{H},\sigma} g$ means that
$\phi(f)\leq\phi(g)$ for every $\phi \in \poshom{\sigma}$.  When $\sigma=\emptyset$, we write $f \leq_{\mathcal{H}} g$ in place of $f \leq_{\mathcal{H},\emptyset} g$, and $\mathcal{C}$ in place of $\mathcal{C}^\emptyset$.
Note that $f \in \mathcal{C}^\sigma$ is the same as saying $0 \leq_{\mathcal{H},\sigma} f$.

Three facts about $\mathcal{C}^\sigma$ are immediate from the definition.
Every $\sigma$-flag $F$ gives an element
$[F]\in\mathcal{C}^\sigma$, because positive homomorphisms send every flag
basis element to a non-negative real by definition.
For any $f \in \mathcal{A}^\sigma$, the square $f^2$ satisfies $0 \leq_{\mathcal{H},\sigma} f^2$,
because $\phi(f^2) = \phi(f)^2 \geq 0$ for every~$\phi \in \poshom{\sigma}$.
And since every $\phi \in \poshom{\sigma}$ is linear, $\mathcal{C}^\sigma$ is
closed under addition and under multiplication by non-negative reals.

\paragraph{From semantic non-negativity to Tur\'an density.}
The problem of bounding the Tur\'an density $\tdensity{F}{\mathcal{H}}$ from above
reduces to establishing membership in $\mathcal{C}^\emptyset$, as the following
lemma shows.

\begin{lemma}\label{lem:semantic-bound}
  Let $F \in \mathcal{F}^\emptyset$ and $c \in \mathbb{R}$ with $c\geq0$.  If
  $[F] \leq_{\mathcal{H}} c \cdot \mathbf{1}$, then
  $\tdensity{F}{\mathcal{H}} \leq c$.
\end{lemma}
\begin{proof}
  Suppose for contradiction that $\tdensity{F}{\mathcal{H}} > c$.  Then some
  $\varepsilon > 0$ and some increasing sequence of \(\mathcal{H}\)-free graphs
  $\{G_k\}_{k\in\mathbb{N}}$ satisfy $\den{F}{G_k} > c + \varepsilon$ for all
  $k$.  The points $p^{G_k}$ lie in the product space
  $[0,1]^{\mathcal{F}^\emptyset}$, which is compact and metrizable because
  $\mathcal{F}^\emptyset$ is countable.  Hence some subsequence
  $\{p^{G_{k_j}}\}$ converges, and by \Cref{thm:convergent-hom}(a) its limit
  is some $\phi \in \poshom{\emptyset}$.
  Since $\phi$ is the pointwise limit of $\{p^{G_{k_j}}\}$, it satisfies
  $\phi([F]) = \lim_j \den{F}{G_{k_j}} \geq c + \varepsilon > c$.  But
  $[F] \leq_{\mathcal{H}} c \cdot \mathbf{1}$ forces $\phi([F]) \leq c$ for
  every~$\phi \in \poshom{\emptyset}$, a contradiction.
\end{proof}

A matching lower bound typically comes from an explicit \(\mathcal{H}\)-free construction
whose $F$-density approaches~$c$, pinning down $\tdensity{F}{\mathcal{H}}$ exactly.
Two classical examples illustrate this.

\begin{example}[Mantel's theorem~\cite{mantel}]\label{ex:mantel}
  Mantel's theorem states $\tdensity{K_2}{K_3} = \tfrac{1}{2}$, meaning
  asymptotically at most half of all vertex pairs in a triangle-free graph
  are edges.  The upper bound translates into the flag algebra inequality
  $[K_2] \leq_{K_3} \tfrac{1}{2}\cdot\mathbf{1}$; the balanced complete bipartite
  graphs $K_{\lfloor n/2\rfloor,\lceil n/2\rceil}$ witness the lower bound.
\end{example}

\begin{example}[Erd\H{o}s Pentagon Theorem~{\cite{grzesik2012,hatami2012}}]\label{ex:pentagon}
  The Erd\H{o}s pentagon problem asks for the maximum density of $C_5$ in
  triangle-free graphs; the asymptotic answer is $\tdensity{C_5}{K_3} = \tfrac{24}{625}$.  The upper bound corresponds to the flag algebra inequality
  $[C_5] \leq_{K_3} \tfrac{24}{625}\cdot\mathbf{1}$; an explicit \(K_3\)-free construction
  (the balanced blow-up of $C_5$) witnesses the lower bound.
\end{example}

\subsection{The Downward Operator}
\label{sec:background-downward}
Labeled flag algebras (i.e., $\mathcal{A}^\sigma$ with $\sigma \neq \emptyset$) are often
richer than the unlabeled algebra $\mathcal{A}^\emptyset$: establishing membership
in $\mathcal{C}^\sigma$ can be easier than working directly in $\mathcal{C}$.
The final extremal statement, however, is always about unlabeled graphs.
The downward operator bridges this gap by mapping $\mathcal{A}^\sigma$ to
$\mathcal{A}^\emptyset$ in a way that preserves semantic non-negativity
(\Cref{thm:downward-nonneg}), so valid inequalities proved in a labeled algebra
can be transferred to the unlabeled one.

For a $\sigma$-flag $F = (M, \theta_F)$, let $F|_\emptyset$ denote the
$\emptyset$-flag obtained from $F$ by forgetting the embedding $\theta_F$,
i.e., viewing $M$ as a $\emptyset$-flag.
For a $\sigma$-flag $F$ of size $m$ with $|V(\sigma)|=k$, let
\[
  q_\sigma(F)
  \;:=\;
  \frac{
    \bigl|\{\theta' : [k]\hookrightarrow V(M)
       \mid \theta' \text{ embeds }\sigma\text{ into }M
       \text{ and } (M,\theta')\cong_\sigma F\}\bigr|}
       {m(m-1)\cdots(m-k+1)}
\]
be the probability that a uniformly random injection $[k]\hookrightarrow V(M)$
recovers the same $\sigma$-flag as $F$.
The \emph{downward operator}
$\llbracket\cdot\rrbracket_\sigma:\mathcal{A}^\sigma\to\mathcal{A}^\emptyset$
is then defined on flag basis elements by
\[
  \llbracket [F] \rrbracket_\sigma
  \;:=\;
  q_\sigma(F)\cdot [F|_\emptyset],
\]
extended linearly to all of $\mathcal{A}^\sigma$.
For $f\in\mathcal{A}^\sigma$, we call
$\llbracket f\rrbracket_\sigma$ the \emph{downward image} of $f$.

We write $\langle\sigma\rangle_0\in\mathcal{A}^\emptyset$ for the empty-type
flag-algebra basis element represented by the underlying graph of the flag type
$\sigma$, with its vertex labels forgotten.  Applying the downward operator to
the typed unit therefore gives
\[
  \llbracket \mathbf{1}_\sigma \rrbracket_\sigma
  =
  q_\sigma(\mathbf{1}_\sigma)\,\langle\sigma\rangle_0.
\]

\paragraph{Random homomorphisms.}
A positive homomorphism $\phi_0 \in \poshom{\emptyset}$ records the limiting
densities of unlabeled graphs along a convergent sequence of finite graphs.
For a nonempty flag type $\sigma$, its random extension can be understood
through any such sequence realizing $\phi_0$: in each finite graph, choose
uniformly at random an embedding of $\sigma$ and use the embedded vertices as
labels.  Razborov's random-extension theorem describes the limiting
distribution of the resulting $\sigma$-flag densities.
Whenever $\phi_0(\langle\sigma\rangle_0) > 0$, meaning that the underlying
unlabeled graph of $\sigma$ has positive limiting density, the following
holds~\cite[Theorem~3.5]{razborov2007flag}:
there exists a unique probability measure on $\poshom{\sigma}$,
which we denote $\Ext_\sigma(\phi_0)$,
such that for every~$f \in \mathcal{A}^\sigma$, the random positive
homomorphism $\phi^\sigma \sim \Ext_\sigma(\phi_0)$ satisfies
\begin{equation}\label{eq:ensemble}
  \mathbb{E}\bigl[\phi^\sigma(f)\bigr]
  \;=\;
  \frac{\phi_0\bigl(\llbracket f \rrbracket_\sigma\bigr)}
       {\phi_0\bigl(\llbracket \mathbf{1}_\sigma \rrbracket_\sigma\bigr)}.
\end{equation}
Note that the denominator is positive by the assumption
$\phi_0(\langle\sigma\rangle_0)>0$, since
$\phi_0\bigl(\llbracket \mathbf{1}_\sigma \rrbracket_\sigma\bigr)
=q_\sigma(\mathbf{1}_\sigma)\cdot
\phi_0(\langle\sigma\rangle_0)$ and $q_\sigma(\mathbf{1}_\sigma)>0$ always.

Equation~\eqref{eq:ensemble} implies that the downward operator $\llbracket \cdot \rrbracket_\sigma$ maps $\mathcal{C}^\sigma$ to $\mathcal{C}$:

\begin{theorem}[Razborov~{\cite[Theorem~3.1(a)]{razborov2007flag}}]\label{thm:downward-nonneg}
  For any flag type $\sigma$ and any $f \in \mathcal{A}^\sigma$ with $0 \leq_{\mathcal{H},\sigma} f$,
  \[
    0 \;\leq_{\mathcal{H}}\; \llbracket f \rrbracket_\sigma.
  \]
\end{theorem}
\begin{proof}[Proof sketch]
  Let $\phi_0 \in \poshom{\emptyset}$; we show $\phi_0(\llbracket f \rrbracket_\sigma) \geq 0$.
  If $\phi_0(\langle\sigma\rangle_0) = 0$, then $\phi_0$ assigns density zero
  to the underlying unlabeled graph of $\sigma$.  The underlying graph of
  every $\sigma$-flag contains an induced copy of $\sigma$, so the downward
  image of every $\sigma$-flag basis element also evaluates to zero under
  $\phi_0$.  By linearity, $\phi_0(\llbracket f \rrbracket_\sigma) = 0$.
  Otherwise, let $\phi^\sigma \sim \Ext_\sigma(\phi_0)$.
  Since $0 \leq_{\mathcal{H},\sigma} f$, every positive homomorphism on
  $\mathcal{A}^\sigma$ evaluates $f$ non-negatively, so
  $\mathbb{E}[\phi^\sigma(f)] \geq 0$.
  Rearranging Equation~\eqref{eq:ensemble}:
  \[
    \phi_0\bigl(\llbracket f \rrbracket_\sigma\bigr)
    \;=\;
    \phi_0\bigl(\llbracket \mathbf{1}_\sigma \rrbracket_\sigma\bigr)
    \cdot
    \mathbb{E}\bigl[\phi^\sigma(f)\bigr]
    \;\geq\; 0.
  \]
  Since $\phi_0$ was arbitrary, we have $0 \leq_{\mathcal{H}} \llbracket f \rrbracket_\sigma$.
\end{proof}

A flag algebra proof of $[F] \leq_{\mathcal{H}} c\cdot\mathbf{1}$ typically proceeds as follows:
first, one exhibits $f \in \mathcal{A}^\sigma$ with $0 \leq_{\mathcal{H},\sigma} f$ for some
nonempty flag type $\sigma$ (e.g., a square $f = g^2$) and applies
\Cref{thm:downward-nonneg} to obtain $0 \leq_{\mathcal{H}} \llbracket f \rrbracket_\sigma$;
then, one adds $\llbracket f \rrbracket_\sigma$ to $[F]$ or an expansion of $[F]$
at an appropriate size, and uses non-negativity of flag basis elements and
Equation~\eqref{eq:sum-to-one} to conclude $[F] \leq_{\mathcal{H}} c\cdot\mathbf{1}$.
The following example carries out this pattern for Mantel's theorem.

\begin{example}[Mantel's theorem via flag algebras]\label{ex:mantel-fa}
We prove $K_2 \leq_{K_3} \tfrac{1}{2}\cdot\mathbf{1}$ in the triangle-free flag
algebra.  Let $\sigma$ be the one-vertex flag type.  Squares are
non-negative, so $0 \leq_{K_3,\sigma} ((\fKtwobull) - (\fbarKtwobull))^2$, and
\Cref{thm:downward-nonneg} carries this inequality to the unlabeled algebra.
Expanding the square, computing each product at size~$3$ (the cross term is
\Cref{ex:typed-edge-nonedge-product}), and then applying the downward operator
gives the following inequality in $\mathcal{A}^{\emptyset}$:
\begin{align}
  0 \leq_{K_3} \llbracket((\fKtwobull) - (\fbarKtwobull))^2\rrbracket_\sigma
  &= \llbracket(\fKtwobull)\cdot(\fKtwobull)
     - 2\cdot(\fKtwobull)\cdot(\fbarKtwobull)
     + (\fbarKtwobull)\cdot(\fbarKtwobull)\rrbracket_\sigma\notag\\
  &= \left\llbracket\fbarKthreebull + \fEthreebull
    - \fbarPthreebull - \fPthreebull
    + \fKonetwobull\right\rrbracket_\sigma\notag\\
  &= \fbarKthree - \frac{1}{3}\,\fbarPthree - \frac{1}{3}\,\fPthree.
  \label{eq:mantel-key}
\end{align}
In the last step, the downward operator multiplies each flag by its
normalizing factor $1$, $\tfrac13$, or $\tfrac23$, according as all three, one,
or two of the choices of the labeled vertex recover that $\sigma$-flag.
Starting from the expansion~\eqref{eq:k2-expansion} of $K_2$, we now add two
non-negative terms and simplify using Equation~\eqref{eq:sum-to-one}:
\begin{align*}
  K_2
  &= \frac{1}{3}\,\fbarPthree + \frac{2}{3}\,\fPthree
  &&\text{by Equation~\eqref{eq:k2-expansion}}\\
  &\leq_{K_3} \frac{1}{3}\,\fbarPthree + \frac{2}{3}\,\fPthree +
    \frac{1}{2}\left(\fbarKthree - \frac{1}{3}\,\fbarPthree - \frac{1}{3}\,\fPthree\right) +
    \frac{1}{3}\,\fbarPthree
  &&\text{by Equation~\eqref{eq:mantel-key} and $0 \leq_{K_3} \fbarPthree$}\\
  &= \frac{1}{2}\cdot\left(\fbarKthree + \fbarPthree + \fPthree\right)
  \;=\; \frac{1}{2}\cdot\mathbf{1}
  &&\text{by Equation~\eqref{eq:sum-to-one}.}
\end{align*}
\end{example}

\section{Formalizing Flag Algebras in Lean~4}
\label{sec:abstract}

Informal flag algebra proofs are compact partly because their notation is
deliberately overloaded.  The same symbol may denote a concrete labeled graph,
its isomorphism class, a basis vector in a free vector space, an element of the
quotient algebra, or the real number obtained by applying a positive
homomorphism to that element.  On paper, these changes of viewpoint are usually
harmless and are left implicit.  In Lean, however, they create an immediate
formalization challenge: each role must be represented by a separate type, and
each passage between roles must be justified.  In particular, constructions
on concrete graphs must be invariant under isomorphism, expansion identities
must hold in the quotient, multiplication must be well defined on the quotient
algebra, and positive homomorphisms must preserve the algebraic operations.
Discharging these obligations, and proving in Lean the flag-algebra theorems
of \Cref{sec:background} that rest on them, is the work of the
\emph{specification layer}, presented in the following subsections.

One design choice pervades the development: how forbidden-subgraph assumptions
are represented.  Where
\Cref{sec:background} fixes \(\mathcal{H}\) first and builds flags and flag
algebras inside the class of \(\mathcal{H}\)-free graphs, our formalization
constructs the ambient flag algebras with no forbidden-family parameter and
imposes \(\mathcal{H}\)-freeness only where the semantic order is tested, by
restricting the positive homomorphisms it quantifies over.  Carrying a
forbidden family through the definitions would burden even flags, densities,
and multiplication with an extra parameter and side conditions; postponing it
lets us formalize the \(\mathcal{H}\)-independent theory once and reuse it for
every forbidden graph.  \Cref{sec:forbidden} defines the resulting ensemble
semantic order and its transfer theorem, and \Cref{sec:metatheory} compares it
with the order obtained by building the constraint in from the outset.

The Lean snippets below are lightly adapted for exposition.  We spell out
parameters that Lean would normally infer, so that a definition's types and
dependencies are visible without the surrounding source context; we gather
those shared by a group of definitions into a \lean{variable} block, which
applies to the definitions after it in the same subsection, leaving
each definition to show only its own arguments.  Declaration names are those
of the development, except for three expository definitions in
\Cref{sec:reflection-concrete}, which are marked as such.
Following the development itself, we write the empty flag type as
\(\emptyt\), rather than the \(\emptyset\) of \Cref{sec:background-flags}, to
distinguish it from ordinary empty sets.

\subsection{Flag Types and Flags}
\label{sec:flagtypes}

Before giving the Lean definitions of flags, we recall the small amount of
Mathlib notation used below.  In \Cref{sec:background}, we described a finite
graph as having a vertex set $V(G)$.  In Lean, the vertex set is represented
by a type: a term of type \lean{SimpleGraph V} is a simple graph whose
vertices are the elements of the type \lean{V}.  For example,
\lean{SimpleGraph (Fin n)} is a graph on the type \lean{Fin n}, whose elements
are the natural numbers \(0,\ldots,n-1\).  Finiteness of the vertex type is
not part of \lean{SimpleGraph} itself; it is supplied separately, by
assumptions such as \lean{Fintype V}.

Mathlib provides ${G\hookrightarrow_g H}$ for the type of graph embeddings
from $G$ to $H$.  Such an embedding is an injective vertex map that preserves
adjacency and non-adjacency, so its image is an induced copy of $G$ in $H$.  Mathlib also provides $G\simeq_g H$ for graph isomorphisms.
With this notation in place, the concrete pair $F=(M,\theta_F)$ from
\Cref{sec:background-flags} is represented by \lean{LabeledGraph}, while its
isomorphism class is represented by \lean{Flag}.

\begin{lstlisting}
abbrev FlagType (T : Type) := SimpleGraph T

structure LabeledGraph {T : Type} [Fintype T] ((*@$\sigma$@*) : FlagType T) (V : Type) where
  graph      : SimpleGraph V
  type_embed : (*@$\sigma$@*) (*@$\hookrightarrow$@*)g graph
\end{lstlisting}
Here \lean{FlagType T} is simply a name for \lean{SimpleGraph T}: a flag type
is represented by a simple graph on the label type \lean{T}.  The structure
\lean{LabeledGraph} packages the two pieces of a concrete $\sigma$-flag: the
field \lean{graph} stores the underlying graph $M$, and the field
\lean{type\_embed} stores the embedding
$\theta_F:\sigma\hookrightarrow_g M$.  In the signature, braces mark implicit
parameters that Lean normally infers from context, whereas square brackets
mark typeclass assumptions such as \lean{[Fintype T]}, which records that the
label type \lean{T} is finite.

We define \lean{LabeledGraphIso} as a structure representing isomorphisms
between two \lean{LabeledGraph} terms, and define \lean{Flag} by quotienting
\lean{LabeledGraph} by this isomorphism relation:
\begin{lstlisting}
structure LabeledGraphIso {T V V' : Type} [Fintype T]
    {(*@$\sigma$@*) : FlagType T} (G : LabeledGraph (*@$\sigma$@*) V) (G' : LabeledGraph (*@$\sigma$@*) V') where
  graph_iso     : G.graph (*@$\simeq$@*)g G'.graph
  type_preserve : graph_iso (*@$\circ$@*) G.type_embed = G'.type_embed

infixl:50 " (*@\textcolor{orange!80!black}{$\simeq$}@*)f " => LabeledGraphIso

instance labeledGraphSetoid {T : Type} [Fintype T] ((*@$\sigma$@*) : FlagType T) (V : Type)
    : Setoid (LabeledGraph (*@$\sigma$@*) V) where
  r     := fun (G G' : LabeledGraph (*@$\sigma$@*) V) => Nonempty (G (*@$\simeq$@*)f G')
  iseqv := ... -- proof that r defines an equivalence relation

def Flag {T : Type} [Fintype T] ((*@$\sigma$@*) : FlagType T) (V : Type) : Type :=
  Quotient (labeledGraphSetoid (*@$\sigma$@*) V)
\end{lstlisting}
The structure \lean{LabeledGraphIso} represents an isomorphism of concrete
$\sigma$-flags.
It consists of a graph isomorphism \lean{graph\_iso} between the underlying
graphs, together with a compatibility proof \lean{type\_preserve} that this
isomorphism sends the labeled copy of $\sigma$ in $G$ to the labeled copy in
$G'$, that is, $\alpha\circ\theta_G=\theta_{G'}$.  The notation
$G\simeq_f G'$ denotes this type of isomorphism data.

Lean's quotient type does not accept an arbitrary relation directly; it
requires a \emph{setoid}, which packages a relation with proofs that it is an
equivalence relation.  The instance \lean{labeledGraphSetoid} provides this
package for \lean{LabeledGraph}~$\sigma$~\lean{V}: its field \lean{r}
identifies two terms when a labeled-graph isomorphism between them exists, and
\lean{iseqv} holds the corresponding proofs.  Then
\lean{Flag}~$\sigma$~\lean{V} is the quotient of
\lean{LabeledGraph}~$\sigma$~\lean{V} by that setoid, so its elements are
precisely isomorphism classes of concrete $\sigma$-flags on the vertex type
\lean{V}.  Lean writes such classes with double brackets:
\(\llbracket G\rrbracket\) is the $\sigma$-flag represented by the concrete
labeled graph \(G\).

We can now represent $\mathcal{F}^\sigma_n$ in Lean by
\lean{Flag}~$\sigma$~\lean{(Fin}~$n$\lean{)}.  The type
\lean{FinFlag}~$\sigma$ collects $\sigma$-flags of all finite sizes and
represents $\mathcal{F}^\sigma = \bigcup_n \mathcal{F}^\sigma_n$.  It is
defined as a dependent sum type: each element is a dependent pair
$\hat F = \langle n, F\rangle$, where $n$ records the size and $F$ is the
$\sigma$-flag of that size.  In Lean, the two components are accessed as
$\hat F.1$ and $\hat F.2$.
\begin{lstlisting}
def FinFlag {k : (*@$\mathbb{N}$@*)} ((*@$\sigma$@*) : FlagType (Fin k)) : Type :=
  (*@$\Sigma$@*) (n : (*@$\mathbb{N}$@*)), Flag (*@$\sigma$@*) (Fin n)
\end{lstlisting}

Unlike the definitions in \Cref{sec:background-flags}, none of the definitions above
carries the \(\mathcal{H}\)-free condition: the underlying graphs are
arbitrary.  This is a deliberate interface choice.  The core quotient and
density machinery is developed once for all finite graphs, while
\(\mathcal{H}\)-freeness is imposed later by restricting the positive
homomorphisms used in the ensemble semantic order
(\Cref{sec:forbidden}).

\subsection{Subflag Densities}
\label{sec:formal-subflag-densities}

Formalizing the density function is not straightforward.  The density
$\den{F_1,\ldots,F_t}{G}$ from \Cref{sec:background-densities} should be a
function taking $t$ pattern $\sigma$-flags and a host $\sigma$-flag to a value in $[0,1]$, but
\lean{Flag} is a quotient type of \lean{LabeledGraph}.  The standard way to
define a function from a quotient type in Lean is to first define it on the
underlying type, prove that it respects the equivalence relation, and lift it
via \lean{Quotient.lift}.  That is, one must first define the density on
labeled graphs, and then prove that replacing any input by an isomorphic one
gives the same result.

We first introduce \lean{LabeledGraphList} as an abbreviation for a
(\lean{Fin t})-indexed tuple of labeled graphs.  Given such a tuple of
patterns and a host, \lean{labeledGraphListDensity} computes the density
directly as a ratio: the number of ways to choose pairwise-disjoint sets of
unlabeled host vertices, one per pattern, that induce the given patterns, over
the number of ways to choose such sets at all.  This is easier to work with in
Lean than the sampling formulation of \Cref{sec:background-densities}.
\begin{lstlisting}
abbrev LabeledGraphList {T : Type} [Fintype T]
    ((*@$\sigma$@*) : FlagType T) (t : (*@$\mathbb{N}$@*)) (Vl : Fin t (*@$\to$@*) Type) : Type :=
  (i : Fin t) (*@$\to$@*) LabeledGraph (*@$\sigma$@*) (Vl i)

noncomputable def labeledGraphListDensity
    {T W : Type} [Fintype T] [Fintype W] [DecidableEq W]
    {t : (*@$\mathbb{N}$@*)} {Vl : Fin t (*@$\to$@*) Type} [FintypeList Vl] [DecidableEqList Vl]
    {(*@$\sigma$@*) : FlagType T} (Hl : LabeledGraphList (*@$\sigma$@*) t Vl) (G : LabeledGraph (*@$\sigma$@*) W) : (*@$\mathbb{Q}$@*) :=
  let r_list (i : Fin t) := (Hl i).size - (*@$\sigma$@*).size
  labeledGraphListCount Hl G / multinomialCoefficient r_list (G.size - (*@$\sigma$@*).size)
\end{lstlisting}
Here a \lean{LabeledGraphList} is a function assigning to each index
\lean{i : Fin t} a labeled graph on its own vertex type \lean{Vl i}, so the
$t$ patterns need not share one ambient vertex type.  In the second
definition, \lean{Hl} is the tuple of patterns, \lean{G} the host, and the
local function \lean{r\_list} records how many unlabeled vertices each pattern
contributes: its value at index $i$ is \(|V(F_i)| - |V(\sigma)|\).  In the
ratio above, the numerator \lean{labeledGraphListCount} counts the tuples of
pairwise-disjoint unlabeled vertex sets of \lean{G} that induce \lean{Hl}, and
the denominator, which \lean{multinomialCoefficient} computes from
\lean{r\_list} and \(|V(G)| - |V(\sigma)|\), counts all such tuples.

The square-bracketed parameters record what finite data is available: for a
type \lean{X}, \lean{[Fintype X]} exhibits a finite list of all elements of
\lean{X}, and \lean{[DecidableEq X]} a procedure deciding equality on
\lean{X}.  Since the pattern tuple may use a different vertex type
\lean{Vl i} at each index, \lean{[FintypeList Vl]} and
\lean{[DecidableEqList Vl]} package the same data for every \lean{Vl i}.
Lean normally supplies all of these by typeclass inference.

The \lean{noncomputable} keyword signals that Lean cannot synthesize an
algorithm to evaluate this definition directly.  This poses no obstacle to
formalizing flag algebra itself, but becomes relevant when applying the algebra
to concrete extremal problems, a point we return to in \Cref{sec:reflection}.

To make \lean{labeledGraphListDensity} a function of $\sigma$-flags, we lift
its two inputs separately.  A \lean{Quotient.lift} on the host argument gives
\(\lean{LabeledGraphList}\to\lean{Flag}\to\mathbb{Q}\); a second one on the
pattern tuple gives
\(\lean{QuotLabeledGraphList}\to\lean{Flag}\to\mathbb{Q}\), where
\lean{QuotLabeledGraphList} is the quotient of concrete tuples by entrywise
isomorphism; and repackaging that input as a tuple of $\sigma$-flags gives the
final \(\lean{FlagList}\to\lean{Flag}\to\mathbb{Q}\).  The congruence proofs
for the two lifts, omitted below, show that the density is unchanged when the
host $\sigma$-flag or any pattern $\sigma$-flag is replaced by an isomorphic
representative.
\begin{lstlisting}
variable {T W : Type} [Fintype T] [Fintype W] [DecidableEq W]
         {t : (*@$\mathbb{N}$@*)} {Vl : Fin t (*@$\to$@*) Type} [FintypeList Vl] [DecidableEqList Vl]
         {(*@$\sigma$@*) : FlagType T}

def flagListEqv (Gl Gl' : LabeledGraphList (*@$\sigma$@*) t Vl) : Prop :=
  (*@$\forall$@*) (i : Fin t), Nonempty (Gl i (*@$\simeq$@*)f Gl' i)

def QuotLabeledGraphList ((*@$\sigma$@*) : FlagType T) (t : (*@$\mathbb{N}$@*)) (Vl : Fin t (*@$\to$@*) Type) : Type :=
  Quotient (labeledGraphListSetoid (*@$\sigma$@*) t Vl) -- the setoid of flagListEqv

abbrev FlagList ((*@$\sigma$@*) : FlagType T) (t : (*@$\mathbb{N}$@*)) (Vl : Fin t (*@$\to$@*) Type) : Type :=
  (*@$\forall$@*) (i : Fin t), Flag (*@$\sigma$@*) (Vl i)

noncomputable def labeledGraphListDensityLifted (Hl : LabeledGraphList (*@$\sigma$@*) t Vl)
    : Flag (*@$\sigma$@*) W (*@$\to$@*) (*@$\mathbb{Q}$@*) :=
  Quotient.lift (fun G => labeledGraphListDensity Hl G)
    (...) -- proof that host-graph equivalence preserves density

noncomputable def quotLabeledGraphListDensity
    : QuotLabeledGraphList (*@$\sigma$@*) t Vl (*@$\to$@*) Flag (*@$\sigma$@*) W (*@$\to$@*) (*@$\mathbb{Q}$@*) :=
  Quotient.lift labeledGraphListDensityLifted
    (...) -- proof that pattern-list equivalence preserves density

noncomputable def FlagList.coe (Fl : FlagList (*@$\sigma$@*) t Vl)
    : QuotLabeledGraphList (*@$\sigma$@*) t Vl :=
  (*@$\llbracket$@*)fun i => (Fl i).out(*@$\rrbracket$@*)

noncomputable def flagListDensity
    : FlagList (*@$\sigma$@*) t Vl (*@$\to$@*) Flag (*@$\sigma$@*) W (*@$\to$@*) (*@$\mathbb{Q}$@*) :=
  fun Fl => quotLabeledGraphListDensity (FlagList.coe Fl)
\end{lstlisting}
Here \lean{flagListEqv} is the entrywise isomorphism relation whose
setoid, formed as in \Cref{sec:flagtypes}, defines
\lean{QuotLabeledGraphList}.  The three density definitions carry out the
three steps above, the last one through \lean{FlagList.coe}, which chooses a
representative \lean{(Fl i).out} for each entry of \lean{Fl} and forms the
quotient class of the resulting concrete tuple.

The following \lean{flagDensity}$_1$ and \lean{flagDensity}$_2$ are
specializations of \lean{flagListDensity} to singleton and pair lists:
\Needspace{11\baselineskip}
\begin{lstlisting}
variable {U U(*@$_1$@*) U(*@$_2$@*) : Type} [Fintype U] [DecidableEq U]
         [Fintype U(*@$_1$@*)] [DecidableEq U(*@$_1$@*)] [Fintype U(*@$_2$@*)] [DecidableEq U(*@$_2$@*)]

noncomputable def flagDensity(*@$_1$@*) (F : Flag (*@$\sigma$@*) U) (G : Flag (*@$\sigma$@*) W) : (*@$\mathbb{Q}$@*) :=
  flagListDensity [F](*@$^\mathtt{f}$@*) G

noncomputable def flagDensity(*@$_2$@*) (F(*@$_1$@*) : Flag (*@$\sigma$@*) U(*@$_1$@*)) (F(*@$_2$@*) : Flag (*@$\sigma$@*) U(*@$_2$@*)) (G : Flag (*@$\sigma$@*) W) : (*@$\mathbb{Q}$@*) :=
  flagListDensity [F(*@$_1$@*), F(*@$_2$@*)](*@$^\mathtt{f}$@*) G
\end{lstlisting}
Here the superscript \(\mathtt{f}\) marks a tuple constructor for
\lean{FlagList}, distinguishing it from ordinary list notation and from
quotient brackets.

Building on the definitions above, we prove the following chain rule
(\Cref{lem:chain-rule}) and several of its variants.
The statement is an elementary counting identity, but its formal proof is not:
the counts are defined on concrete labeled graphs, and the identity between
them comes from an explicit bijection between the choices the two sampling
procedures make (\Cref{sec:proof-style}).
\begin{lstlisting}
theorem flagDensity_eq_sum_density_prods
    {(*@$\ell_0$@*) (*@$\ell_1$@*) (*@$\ell$@*) (*@$\ell'$@*) : (*@$\mathbb{N}$@*)} {(*@$\sigma$@*) : FlagType (Fin (*@$\ell_0$@*))} (F(*@$_1$@*) : Flag (*@$\sigma$@*) (Fin (*@$\ell_1$@*))) (G : Flag (*@$\sigma$@*) (Fin (*@$\ell$@*)))
    (h(*@$\ell_1$@*) : (*@$\ell_0$@*) (*@$\leq$@*) (*@$\ell_1$@*)) (h(*@$\ell'$@*) : (*@$\ell_1$@*) (*@$\leq$@*) (*@$\ell'$@*)) (h(*@$\ell$@*) : (*@$\ell'$@*) (*@$\leq$@*) (*@$\ell$@*))
    : flagDensity(*@$_1$@*) F(*@$_1$@*) G = (*@$\sum$@*) G' : Flag (*@$\sigma$@*) (Fin (*@$\ell'$@*)), flagDensity(*@$_1$@*) F(*@$_1$@*) G' * flagDensity(*@$_1$@*) G' G := ...
\end{lstlisting}

\subsection{The Flag Algebra}
\label{sec:formal-flag-algebra}

To define the flag algebra, we first construct the free real vector space on
$\mathcal{F}^\sigma$.  Each element of this vector space is represented as a
finitely supported function from $\mathcal{F}^\sigma$ to $\mathbb{R}$, using
Mathlib's \lean{Finsupp} type.  We call the resulting type
\lean{FlagVector}~$\sigma$; the constructor \lean{basisVector} produces the
basis vector for each $\sigma$-flag:
\begin{lstlisting}
variable {k : (*@$\mathbb{N}$@*)} {(*@$\sigma$@*) : FlagType (Fin k)}

abbrev FlagVector ((*@$\sigma$@*) : FlagType (Fin k)) : Type :=
  FinFlag (*@$\sigma$@*) (*@$\to_0$@*) (*@$\mathbb{R}$@*)

def basisVector (F : FinFlag (*@$\sigma$@*)) : FlagVector (*@$\sigma$@*) :=
  Finsupp.single F 1
\end{lstlisting}
Here $\to_0$ denotes the type of finitely supported functions, and \lean{Finsupp.single F 1} constructs the basis vector $\mathbf e_F$ for $F \in \mathcal{F}^\sigma$: the function taking value $1$ at $F$ and $0$ elsewhere.

Next, we define the zero space $\mathcal{Z}^\sigma$.  The auxiliary function
\lean{flagExpansion F}~$\ell$ is the Lean version of
$\sum_{G \in \mathcal{F}^\sigma_\ell} \den{F}{G}\cdot \mathbf e_G$, the expansion of
$F$ at size $\ell$, with the black dot below for the scalar multiplication
written $\cdot$ here.  We encode the generating element of
Equation~\eqref{eq:zero-element} as \lean{zeroElement}, \lean{zeroSet} $\sigma$
collects all elements of this form, and \lean{ZeroSpace} $\sigma$ is the
linear subspace they span, represented in Lean by a \lean{Submodule}:
\begin{lstlisting}
noncomputable def flagExpansion (F : FinFlag (*@$\sigma$@*)) ((*@$\ell$@*) : (*@$\mathbb{N}$@*)) : FlagVector (*@$\sigma$@*) :=
  (*@$\sum$@*) G : Flag (*@$\sigma$@*) (Fin (*@$\ell$@*)), flagDensity(*@$_1$@*) F.2 G (*@$\leansmul$@*) basisVector (*@$\langle$@*)(*@$\ell$@*), G(*@$\rangle$@*)

noncomputable def zeroElement (F : FinFlag (*@$\sigma$@*)) ((*@$\ell$@*) : (*@$\mathbb{N}$@*)) : FlagVector (*@$\sigma$@*) :=
  basisVector F - flagExpansion F (*@$\ell$@*)
\end{lstlisting}
\Needspace{9\baselineskip}
\begin{lstlisting}
noncomputable def zeroSet ((*@$\sigma$@*) : FlagType (Fin k)) : Set (FlagVector (*@$\sigma$@*)) :=
  {z | (*@$\exists$@*) (F : FinFlag (*@$\sigma$@*)) ((*@$\ell$@*) : (*@$\mathbb{N}$@*)), F.1 (*@$\leq$@*) (*@$\ell$@*) (*@$\wedge$@*) z = zeroElement F (*@$\ell$@*)}

noncomputable def ZeroSpace ((*@$\sigma$@*) : FlagType (Fin k)) : Submodule (*@$\mathbb{R}$@*) (FlagVector (*@$\sigma$@*)) :=
  Submodule.span (*@$\mathbb{R}$@*) (zeroSet (*@$\sigma$@*))
\end{lstlisting}
Following the same pattern as the definition of \lean{Flag} from
\lean{LabeledGraph} in \Cref{sec:flagtypes}, we now define the flag algebra
$\mathcal{A}^\sigma$ as a quotient of \lean{FlagVector} $\sigma$, by the
relation \lean{flagVectorEqv} that two flag vectors differ by an element of
\lean{ZeroSpace} $\sigma$.
\begin{lstlisting}
def flagVectorEqv (f g : FlagVector (*@$\sigma$@*)) : Prop :=
  f - g (*@$\in$@*) ZeroSpace (*@$\sigma$@*)

abbrev FlagAlgebra ((*@$\sigma$@*) : FlagType (Fin k)) : Type :=
  Quotient (flagVectorSetoid (*@$\sigma$@*)) -- the setoid of flagVectorEqv
\end{lstlisting}
Because \lean{ZeroSpace} $\sigma$ is spanned by the differences
\lean{zeroElement F}~$\ell$, this quotient imposes a family of expansion
identities: for every $\sigma$-flag \lean{F : FinFlag}~$\sigma$ and every size
$\ell \geq F.1$, the equivalence class of \lean{basisVector F} is identified
with that of its expansion \lean{flagExpansion F}~$\ell$.  As
\lean{FlagAlgebra}~$\sigma$ is again a quotient type, if $f$ is a
\lean{FlagVector}~$\sigma$ then $\llbracket f\rrbracket$ is its equivalence
class in this quotient.

Formalizing multiplication in $\mathcal{A}^\sigma$ also requires several
steps.  The term \lean{flagMulWithSize} defines
the product of two $\sigma$-flags at a chosen output size $\ell$ as in
Equation~\eqref{eq:mul}, and \lean{flagMul} specializes it to the minimal
admissible size $\ell = F.1 + F'.1 - k$.  Then \lean{bilinearExtension}
extends \lean{flagMul} linearly in each argument to give a product on
\lean{FlagVector} $\sigma$, an operation on formal sums for which none of the
algebra laws yet hold, and \lean{Quotient.map$_2$} lifts that product to
\lean{FlagAlgebra} $\sigma$.
\begin{lstlisting}
noncomputable def flagMulWithSize (F F' : FinFlag (*@$\sigma$@*)) ((*@$\ell$@*) : (*@$\mathbb{N}$@*)) : FlagVector (*@$\sigma$@*) :=
  (*@$\sum$@*) G : Flag (*@$\sigma$@*) (Fin (*@$\ell$@*)), flagDensity(*@$_2$@*) F.2 F'.2 G (*@$\leansmul$@*) basisVector (*@$\langle$@*)(*@$\ell$@*), G(*@$\rangle$@*)

noncomputable def flagMul (F F' : FinFlag (*@$\sigma$@*)) : FlagVector (*@$\sigma$@*) :=
  flagMulWithSize F F' (F.1 + F'.1 - k)

noncomputable instance ((*@$\sigma$@*) : FlagType (Fin k)) : Mul (FlagVector (*@$\sigma$@*)) where
  mul := bilinearExtension flagMul

noncomputable instance ((*@$\sigma$@*) : FlagType (Fin k)) : Mul (FlagAlgebra (*@$\sigma$@*)) where
  mul := Quotient.map(*@$_2$@*) ((*@$\cdot$@*) * (*@$\cdot$@*))
    (...) -- proof that the product respects flagVectorEqv in each argument
\end{lstlisting}
Here \lean{Quotient.map$_2$} lifts a binary operation to a quotient type given
a proof that the operation respects the equivalence relation.  The expression
\lean{($\cdot$ * $\cdot$)} is not a new operation: its two arguments are flag
vectors, so Lean resolves \lean{*} using the immediately preceding
\lean{Mul (FlagVector}~$\sigma$\lean{)} instance, whose multiplication is
\lean{bilinearExtension flagMul}.  Its proof obligation is that replacing
either factor by an equivalent flag vector gives an equivalent product, which
settles the representative half of the independence claimed for
Equation~\eqref{eq:mul}.  A separate theorem settles the other half, that any
two admissible output sizes give equivalent flag vectors.
\begin{lstlisting}
theorem flagMulWithSize_indep_on_size {F(*@$_1$@*) F(*@$_2$@*) : FinFlag (*@$\sigma$@*)} {(*@$\ell_1$@*) (*@$\ell_2$@*) : (*@$\mathbb{N}$@*)}
    (h(*@$\ell_1$@*) : F(*@$_1$@*).1 + F(*@$_2$@*).1 (*@$\leq$@*) (*@$\ell_1$@*) + k) (h(*@$\ell_2$@*) : F(*@$_1$@*).1 + F(*@$_2$@*).1 (*@$\leq$@*) (*@$\ell_2$@*) + k)
    : flagVectorEqv (flagMulWithSize F(*@$_1$@*) F(*@$_2$@*) (*@$\ell_1$@*)) (flagMulWithSize F(*@$_1$@*) F(*@$_2$@*) (*@$\ell_2$@*)) := ...
\end{lstlisting}

With these definitions in place, we prove that \lean{FlagAlgebra} $\sigma$ is
a commutative algebra over the real numbers, completing the formalization of
the algebraic structure of $\mathcal{A}^\sigma$.

\subsection{Positive Homomorphisms and Semantic Order}

This subsection describes how we encode the semantic side of flag algebras in
Lean: sequences of finite $\sigma$-flags, their limiting density functions, positive
homomorphisms, and the order induced by testing against all positive
homomorphisms.  The convergence theorem from
Theorem~\ref{thm:convergent-hom} is the main mathematical bridge between finite
$\sigma$-flag sequences and positive homomorphisms, but the formalization first has to
make each of these surrounding notions explicit.

A sequence of $\sigma$-flags indexed by $\mathbb{N}$ is represented by
\lean{FlagSeq}~$\sigma$, and \lean{ConvergesTo s a} says that such a sequence
$s$ is convergent in the sense of \Cref{sec:background-semantic-order}, with
limit $a$.
\begin{lstlisting}
variable {k : (*@$\mathbb{N}$@*)} {(*@$\sigma$@*) : FlagType (Fin k)}

abbrev FlagSeq ((*@$\sigma$@*) : FlagType (Fin k)) :=
  (*@$\mathbb{N}$@*) (*@$\to$@*) FinFlag (*@$\sigma$@*)

def ConvergesTo (s : FlagSeq (*@$\sigma$@*)) (a : FinFlag (*@$\sigma$@*) (*@$\to$@*) (*@$\mathbb{R}$@*)) : Prop :=
  Increases s (*@$\wedge$@*) Tendsto (flagDensitySeq s) atTop ((*@$\mathcal{N}$@*) a)
\end{lstlisting}
Here the candidate limit \(a\) is a real-valued function on finite
\(\sigma\)-flags.  The first conjunct, \lean{Increases s}, asserts that the
underlying sizes of the \(\sigma\)-flags \(s(0),s(1),\ldots\) strictly increase.  The
second conjunct uses Lean's filter notation for convergence:
\lean{atTop} is the limit \(n\to\infty\) on \(\mathbb{N}\), and
\lean{$\mathcal{N}$\,\,a} is the neighborhood filter of the function~\(a\).  The term
\lean{flagDensitySeq s} is the sequence of density functions associated to
\(s\); its \(n\)-th value sends a finite \(\sigma\)-flag \(F\) to
\(\den{F}{s(n)}\).  Thus, the \lean{Tendsto} line is the Lean statement that,
for every finite \(\sigma\)-flag~\(F\),
\[
  \lim_{n\to\infty}\den{F}{s(n)}=a(F).
\]

The Lean definition of a positive homomorphism follows the mathematical one of
\Cref{sec:background-semantic-order} almost literally: a map
\(\phi:\mathcal{A}^\sigma\to\mathbb{R}\) preserving the algebraic structure and
non-negative on every flag basis element.
\begin{lstlisting}
abbrev Hom ((*@$\sigma$@*) : FlagType (Fin k)) :=
  FlagAlgebra (*@$\sigma$@*) (*@$\to_a$@*)[(*@$\mathbb{R}$@*)] (*@$\mathbb{R}$@*)

def PositiveHom ((*@$\sigma$@*) : FlagType (Fin k)) : Type :=
  { (*@$\phi$@*) : Hom (*@$\sigma$@*) // (*@$\forall$@*) F : FinFlag (*@$\sigma$@*), (*@$\phi$@*) (*@$\llbracket$@*)basisVector F(*@$\rrbracket$@*) (*@$\geq$@*) 0 }
\end{lstlisting}
The arrow \lean{$\to_a$[$\mathbb{R}$]} is Lean's notation for the type of
functions preserving addition, multiplication, multiplication by real scalars,
and the unit.  The second definition is Lean's subtype construction, whose
separator \lean{//} precedes the defining predicate; mathematically we read
\lean{PositiveHom}~$\sigma$ as the set
\[
  \{\phi\in \operatorname{Hom}(\mathcal{A}^\sigma,\mathbb{R})
    \mid \phi([F])\ge 0
    \text{ for every }F\in\mathcal{F}^{\sigma}\}.
\]

Using these definitions, we formalize Theorem~\ref{thm:convergent-hom} as two
separate theorems.\footnote{For readability, the displayed statements give the
limit profile the ambient function type
\lean{FinFlag $\sigma$ $\to$ $\mathbb{R}$}.  The Lean implementation instead
packages such profiles in the auxiliary subtype
\lean{FlagDensitySpace}~\(\sigma\), which records that every value lies in
\([0,1]\), and packages the profile induced by a positive homomorphism as
\lean{$\phi$.coe}.  Accordingly, the implemented conclusions are written more
compactly as \lean{$\phi$.coe = a} and
\lean{ConvergesTo s $\phi$.coe}.  The formulations are equivalent: every flag
density lies in \([0,1]\), so the pointwise limit of a convergent flag sequence
does as well.}
\begin{lstlisting}
(*@\textcolor{gray}{\itshape -- Theorem~\ref*{thm:convergent-hom} (a)}@*)
theorem flagSeq_limit_mem_positiveHom (s : FlagSeq (*@$\sigma$@*))
    {a : FinFlag (*@$\sigma$@*) (*@$\to$@*) (*@$\mathbb{R}$@*)} (hs_conv : ConvergesTo s a)
    : (*@$\exists$@*) ((*@$\phi$@*) : PositiveHom (*@$\sigma$@*)), (*@$\forall$@*) (F : FinFlag (*@$\sigma$@*)), (*@$\phi$@*) (*@$\llbracket$@*)basisVector F(*@$\rrbracket$@*) = a F := ...

(*@\textcolor{gray}{\itshape -- Theorem~\ref*{thm:convergent-hom} (b)}@*)
theorem positiveHom_as_flagSeq_limit ((*@$\phi$@*) : PositiveHom (*@$\sigma$@*))
    : (*@$\exists$@*) (s : FlagSeq (*@$\sigma$@*)), ConvergesTo s (fun F => (*@$\phi$@*) (*@$\llbracket$@*)basisVector F(*@$\rrbracket$@*)) := ...
\end{lstlisting}
Part~(a) assumes \lean{ConvergesTo s a} and produces a positive
homomorphism \(\phi\) with
\(\phi(\llbracket\lean{basisVector F}\rrbracket)=a(F)\) for every finite
\(\sigma\)-flag \(F\).  Thus, the limiting density profile of every convergent
\(\sigma\)-flag sequence is induced by a positive homomorphism.  Conversely,
part~(b) starts with a positive homomorphism \(\phi\) and produces a
\(\sigma\)-flag sequence whose limiting density profile is the one \(\phi\)
induces.
Both directions require nontrivial formalization work.  In part (a), one must
show that the limit \(a\) of a \(\sigma\)-flag sequence \(s\) satisfies all algebraic
properties of a positive homomorphism, including details that informal proofs
leave implicit.  Part~(b) goes the other way, and its proof is probabilistic:
a sequence of \(\sigma\)-flags sampled using the values of \(\phi\) as weights
almost surely has \(\phi\) as its limiting density profile.  Formalizing it
means constructing the product measure on flag sequences, proving that the
coordinatewise convergence event has probability one, and selecting a
realization from that event.

We also define the semantic non-negativity cone $\mathcal{C}^\sigma$
and the order it induces on $\mathcal{A}^\sigma$.
Here,
\lean{semanticCone}~$\sigma$ is the set of elements in
\lean{FlagAlgebra}~$\sigma$ that every positive homomorphism
evaluates non-negatively.  The induced order is installed as a Lean
\lean{LE} instance, with $f \leq g$ defined to mean
$g - f \in \mathcal{C}^\sigma$.
\begin{lstlisting}
def semanticCone ((*@$\sigma$@*) : FlagType (Fin k)) : Set (FlagAlgebra (*@$\sigma$@*)) :=
  { f | (*@$\forall$@*) (*@$\phi$@*) : PositiveHom (*@$\sigma$@*), (*@$\phi$@*) f (*@$\geq$@*) 0 }

instance ((*@$\sigma$@*) : FlagType (Fin k)) : LE (FlagAlgebra (*@$\sigma$@*)) where
  le := fun f g => g - f (*@$\in$@*) semanticCone (*@$\sigma$@*)
\end{lstlisting}

\subsection{The Downward Operator}
\label{sec:formal-downward}

The downward operator from \Cref{sec:background-downward} forgets the
labels of a typed flag-algebra expression while keeping the correct
normalizing factor.  On a single $\sigma$-flag $F$, the intended formula is
\[
  \llbracket [F] \rrbracket_\sigma = q_\sigma(F)\cdot [F|_\emptyset].
\]
The displayed formula specifies the downward operator on flag basis elements.
To implement it, we first define its two ingredients, the normalizing factor
$q_\sigma(F)$ and the unlabeled flag $F|_\emptyset$, on concrete
$\sigma$-typed labeled graphs and lift both to flags.  This gives the downward
image of an individual $\sigma$-flag.  We then extend that operation linearly
to flag vectors and lift it once more, to $\mathcal{A}^\sigma$, giving a
well-defined operator
\(\mathcal{A}^\sigma\longrightarrow\mathcal{A}^\emptyset\).

For a concrete labeled graph $G$ on $n$ vertices, the normalizing factor is a
ratio: the numerator \lean{isomorphismCount G} counts the injections
$[k]\hookrightarrow V(G)$ that produce a typed labeled graph isomorphic
to~$G$, and the denominator $n!/(n-k)!$ counts all possible injections.  The
other ingredient, $F|_\emptyset$, is computed by
\lean{unlabeledGraph}, which forgets the type embedding.  Lifting these
two operations to flags takes the \lean{Quotient.lift} pattern of
\Cref{sec:formal-subflag-densities}, so each needs a proof of invariance under
isomorphism: isomorphic representatives admit the same label
placements, and their unlabelings are isomorphic as $\emptyt$-flags.  The two
then combine in \lean{downwardFlag}, the downward image of an individual
$\sigma$-flag.
\Needspace{14\baselineskip}
\begin{lstlisting}
variable {k n : (*@$\mathbb{N}$@*)} {V : Type} {(*@$\sigma$@*) : FlagType (Fin k)}

noncomputable def downwardNormalizingFactor : Flag (*@$\sigma$@*) (Fin n) (*@$\to$@*) (*@$\mathbb{Q}$@*) :=
  Quotient.lift (fun G => isomorphismCount G / (n.factorial / (n - k).factorial)) (...)

noncomputable def unlabel : Flag (*@$\sigma$@*) V (*@$\to$@*) Flag (*@$\emptyt$@*) V :=
  Quotient.lift (fun G => (*@$\llbracket$@*)unlabeledGraph G(*@$\rrbracket$@*)) (...)

noncomputable def downwardFlag (F : Flag (*@$\sigma$@*) (Fin n)) : FlagVector (*@$\emptyt$@*) :=
  downwardNormalizingFactor F (*@$\leansmul$@*) basisVector (*@$\langle$@*)n, unlabel F(*@$\rangle$@*)
\end{lstlisting}

We next lift this construction to the algebra.  The map
\lean{downwardFlagVector} extends \lean{downwardFlag} linearly to
$\sigma$-typed flag vectors, and \lean{downwardFlagVectorQuot} takes the
zero-space quotient of its output, so its values lie in
\lean{FlagAlgebra}~$\emptyt$ while its input is still a flag vector.
Quotienting the input as well requires the respect lemma below: two
$\sigma$-typed flag vectors differing by a zero-space element have downward
images representing the same element of the empty-type algebra.  This is the
compatibility condition needed for the final quotient lift.
\begin{lstlisting}
noncomputable def downwardFlagVector : FlagVector (*@$\sigma$@*) (*@$\to$@*) FlagVector (*@$\emptyt$@*) :=
  linearExtension (fun F : FinFlag (*@$\sigma$@*) => downwardFlag F.2)

noncomputable def downwardFlagVectorQuot : FlagVector (*@$\sigma$@*) (*@$\to$@*) FlagAlgebra (*@$\emptyt$@*) :=
  fun f : FlagVector (*@$\sigma$@*) => (*@$\llbracket$@*)downwardFlagVector f(*@$\rrbracket$@*)

lemma downwardFlagVectorQuot_respect_eqv {f f' : FlagVector (*@$\sigma$@*)}
    (h : f - f' (*@$\in$@*) ZeroSpace (*@$\sigma$@*))
    : downwardFlagVectorQuot f = downwardFlagVectorQuot f' := ...

noncomputable def downward : FlagAlgebra (*@$\sigma$@*) (*@$\to$@*) FlagAlgebra (*@$\emptyt$@*) :=
  Quotient.lift (fun f : FlagVector (*@$\sigma$@*) => downwardFlagVectorQuot f) (...)

notation "(*@\textcolor{orange!80!black}{$\llbracket$}@*)" f "(*@\textcolor{orange!80!black}{$\rrbracket_0$}@*)" => (downward f)
\end{lstlisting}
The last definition is now a genuine map
\(\lean{FlagAlgebra}~\sigma\to\lean{FlagAlgebra}~\emptyt\), independent of all
representative choices.  The final line introduces the Lean notation
\(\llbracket f\rrbracket_0\) for \lean{downward f}; the subscript $0$
emphasizes that the result lies in the empty-type algebra.  In mathematical
prose we use \(\llbracket f\rrbracket_\sigma\), following
\Cref{sec:background-downward}, where the subscript instead records the type
being forgotten.

Having constructed the algebraic map, we turn to its semantic meaning: a
formal counterpart of the random-extension identity of
Equation~\eqref{eq:ensemble},
\[
  \mathbb{E}\bigl[\phi^\sigma(f)\bigr]
  \;=\;
  \frac{\phi_0\bigl(\llbracket f \rrbracket_\sigma\bigr)}
       {\phi_0\bigl(\llbracket \mathbf{1}_\sigma \rrbracket_\sigma\bigr)}.
\]
Stating it in Lean requires a measurable space for the measure to live on,
and the implementation obtains one by identifying a positive homomorphism
with its density profile: the sample space \lean{PositiveHomSpace}~$\sigma$ is
the measurable subspace of \([0,1]^{\mathcal{F}^\sigma}\) consisting of the
profiles induced by positive homomorphisms.  The formal counterpart is then
the following existence statement.
\Needspace{10\baselineskip}
\begin{lstlisting}
theorem (*@\texttt{exists\_probMeasure\_extend\_emptyType\_positiveHom}@*)
    {(*@$\phi_0$@*) : PositiveHom (*@$\emptyt$@*)} (h(*@$\sigma$@*) : (*@$\phi_0$@*) ((*@$\langle\sigma\rangle_0$@*)) > 0)
    : (*@$\exists$@*) (*@$\mathbb{P}$@*) : ProbabilityMeasure (PositiveHomSpace (*@$\sigma$@*)), (*@$\forall$@*) f : FlagAlgebra (*@$\sigma$@*),
        (*@$\int$@*) ((*@$\phi$@*) : PositiveHomSpace (*@$\sigma$@*)), PositiveHomSpace.toPosHom (*@$\phi$@*) f (*@$\partial$@*)(*@$\mathbb{P}$@*)
        = (*@$\phi_0$@*) (*@$\llbracket$@*)f(*@$\rrbracket_0$@*) / (*@$\phi_0$@*) (*@$\llbracket$@*)(1 : FlagAlgebra (*@$\sigma$@*))(*@$\rrbracket_0$@*)
  := ...
\end{lstlisting}
Mathlib writes \lean{$\int$ x, g x $\partial\mu$} for the integral of
$x\mapsto g(x)$ with respect to $\mu$; here $\partial\mu$ specifies the
measure and is not a derivative.  In the displayed theorem, $\phi$ ranges over
\lean{PositiveHomSpace}~$\sigma$, the integrand
\lean{PositiveHomSpace.toPosHom $\phi$ f} evaluates at $f$ the homomorphism
that $\phi$ induces, and the measure is $\mathbb{P}$.  Because
$\mathbb{P}$ is a probability measure, the left-hand side is the expectation
$\mathbb{E}_{\phi\sim \mathbb{P}}[\phi(f)]$.  The theorem therefore says
exactly that $\mathbb{P}$ satisfies the identity above, so $\mathbb{P}$
realizes the random-extension measure $\Ext_\sigma(\phi_0)$
of \Cref{sec:background-downward}.

The Lean counterpart of \Cref{thm:downward-nonneg} now follows from the
random-extension identity:
\begin{lstlisting}
theorem downward_preserve_semanticCone (f : FlagAlgebra (*@$\sigma$@*))
    (hf : f (*@$\in$@*) semanticCone (*@$\sigma$@*)) : (*@$\llbracket$@*)f(*@$\rrbracket_0$@*) (*@$\in$@*) semanticCone (*@$\emptyt$@*) := ...
\end{lstlisting}
This theorem concerns the semantic cone of the ambient algebra.  The
certificate proofs of \Cref{sec:flagmatic} use its counterpart for the
ensemble order of \Cref{sec:forbidden}, \lean{downward\_forbidLEWith\_nonneg},
which follows from the random-extension identity in the same way.

\subsection{Forbidden Graphs and the Ensemble Semantic Order}
\label{sec:forbidden}

In \Cref{sec:background}, a forbidden family \(\mathcal{H}\) is fixed
first, and the flag types, flags, algebras, and positive homomorphisms are all
built inside the \(\mathcal{H}\)-free world.  In that setup, an inequality
\(f\leq_{\mathcal{H},\sigma} g\) is already an inequality in the constrained
\(\sigma\)-typed algebra: the forbidden configurations have been removed
before the semantic order is tested.

Our Lean formalization deliberately separates building the algebra from
imposing the constraint.  The core development constructs the ambient flag
algebra \(\mathcal{A}^\sigma\) of all finite simple graphs, with no forbidden
family as a parameter.  A forbidden graph enters only when we compare two
ambient algebra elements \(f,g\in\mathcal{A}^\sigma\), and even then it defines
no separate flag-algebra universe: the forbidden graph restricts the
empty-type positive homomorphisms from which the random extensions of
\Cref{sec:background-downward} are formed.  The whole construction of the
preceding subsections is therefore reused unchanged across extremal problems.

The order used for this a posteriori condition is the
\emph{ensemble semantic order}.  We state it for one forbidden graph \(H\); a
family would impose the zero-density condition below for each member.

\begin{definition}[Ensemble semantic order]\label{def:ensemble-semantic-order}
Fix a forbidden graph \(H\), a flag type \(\sigma\), and elements
\(f,g\in\mathcal{A}^\sigma\) in the ambient flag algebra.  We write
\(f\leq^{\mathrm{ens}}_{H,\sigma} g\) if, for every empty-type positive
homomorphism \(\phi_0\) such that, viewing graphs as
empty-type flag-algebra elements,
\[
  \phi_0([F])=0
  \quad\text{for every graph \(F\) containing \(H\) as a subgraph}
  \qquad\text{and}\qquad
  \phi_0(\langle\sigma\rangle_0)>0,
\]
the random extension \(\phi\sim\Ext_\sigma(\phi_0)\) satisfies
\[
  \mathbb{P}\bigl[\phi(f)\leq\phi(g)\bigr]=1.
\]
\end{definition}

The first condition restricts \(\phi_0\) to the \(H\)-free unlabeled limits: it
assigns density zero to every graph containing \(H\).  The second condition,
\(\phi_0(\langle\sigma\rangle_0)>0\), says that the underlying unlabeled graph of
\(\sigma\) has positive limiting density along any sequence of finite
graphs converging to \(\phi_0\).  This is precisely the hypothesis under which the random-extension
measure \(\Ext_\sigma(\phi_0)\) is defined.  The final display asks that
\(\phi(f)\leq\phi(g)\) hold with probability one under this measure.

\begin{remark}[Relation to the standard \(H\)-free semantics]\label{rem:builtin-hfree}
The standard approach of \Cref{sec:background} and the ensemble approach agree
at the empty type, where no labels are chosen, so the distinction does not
affect the final unlabeled Tur\'an-density bounds.  For nonempty types,
however, they can differ: the standard approach may fix the labels at an
exceptional vertex configuration, which uniform random sampling makes
invisible to the ensemble comparison.
\Cref{sec:metatheory} gives a precise criterion for when the two approaches
agree and an example showing that they need not.
\end{remark}

The Lean predicate \lean{forbidLE} is the formalization of
\Cref{def:ensemble-semantic-order}:
\begin{lstlisting}
variable {k m : (*@$\mathbb{N}$@*)} {(*@$\sigma$@*) : FlagType (Fin k)}

def forbidLEWith (C : ForbidCondition) (f g : FlagAlgebra (*@$\sigma$@*)) : Prop :=
  (*@$\forall\;$@*)((*@$\phi_0$@*) : PositiveHom (*@$\emptyt$@*)),
    (*@$\phi_0$@*) (*@$\langle\sigma\rangle_0$@*) > 0
    -> (*@$\,$@*) C (*@$\phi_0$@*)
    -> (*@$\,$@*) (*@$\mathbb{P}$@*)[(*@$\phi_0$@*)] {(*@$\phi$@*) : PositiveHomSpace (*@$\sigma$@*) | (*@$\phi$@*) f (*@$\leq\;$@*)(*@$\phi$@*) g} = 1

def forbidLE (H : SimpleGraph (Fin m)) (f g : FlagAlgebra (*@$\sigma$@*)) : Prop :=
  forbidLEWith (forbiddenCondition H) f g
\end{lstlisting}
Read \lean{forbidLE} first.  Its explicit argument \lean{H} is the forbidden
graph, and \lean{f} and \lean{g} are the two ambient \(\sigma\)-typed
flag-algebra elements being compared.  It delegates to \lean{forbidLEWith}, the
condition-generic version parameterized by a \lean{ForbidCondition}, a
predicate selecting which empty-type positive homomorphisms are allowed.  Here
that predicate is \lean{forbiddenCondition H}, the \(H\)-freeness condition of
\Cref{def:ensemble-semantic-order}.

In the body of \lean{forbidLEWith}, the hypothesis \lean{C $\phi_0$} imposes
that condition and the conclusion is the probability-one requirement of the
same definition.  Two points are specific to the formalization: the event is
carved out of \lean{PositiveHomSpace $\sigma$}, the measurable space of density
profiles of \Cref{sec:formal-downward}, each of which determines a positive
homomorphism on \(\mathcal{A}^\sigma\); and the measure
\lean{$\mathbb{P}$[$\phi_0$]} is the random-extension measure
\(\Ext_\sigma(\phi_0)\) on that space.

What remains is to connect this semantic predicate to the corresponding
asymptotic bound on induced densities, the Lean counterpart of
\Cref{lem:semantic-bound} with its hypothesis stated in the ensemble order
rather than the standard order of \Cref{sec:background-semantic-order}.  In the
theorem below, \(H\) is the forbidden graph, \(F\) is the target graph, and
\lean{F.toFlagAlgebra} converts \(F\) into the element of
\lean{FlagAlgebra $\emptyt$} whose evaluation is the induced \(F\)-density.
\begin{lstlisting}
theorem generalizedTuranDensity_le_of_forbidLE
    {n m : (*@$\mathbb{N}$@*)} {H : SimpleGraph (Fin n)}
    {F : SimpleGraph (Fin m)} {c : (*@$\mathbb{R}$@*)} (hc : 0 (*@$\leq$@*) c)
    (h : forbidLE H F.toFlagAlgebra (c (*@$\leansmul$@*) (1 : FlagAlgebra (*@$\emptyt$@*))))
    : generalizedTuranDensity H F (*@$\leq$@*) c := ...
\end{lstlisting}
The hypothesis \lean{h} is the empty-type ensemble inequality
\([F]\leq^{\mathrm{ens}}_{H,\emptyset}c\cdot\mathbf{1}\), and the conclusion is
the Tur\'an-density bound \(\tdensity{F}{H}\leq c\).  Each later application
proves the hypothesis and calls this theorem for the final conversion.

The Tur\'an density \(\tdensity{F}{H}\) is represented in Lean by
\lean{generalizedTuranDensity H F}, whose definition applies Mathlib's
\lean{limUnder} to the normalized extremal numbers
\(\mathrm{ex}(n,\,F;\,H)/\binom{n}{|V(F)|}\) of
\Cref{sec:background-semantic-order}.  That operation returns a real number
even before the sequence has been shown to converge, in which case the value
would have no asymptotic meaning.  We therefore prove
\lean{tendsto\_\allowbreak generalizedTuranDensity}: once \(n\geq |V(F)|\) the
sequence is non-increasing and bounded below by zero, so it converges, and its
limit is \lean{generalizedTuranDensity H F}.

\section{The Reflection Layer}
\label{sec:reflection}

The finite calculations a formal proof depends on can be discharged in two
ways.  One is to proceed as a
mathematician usually would: when a finite claim is needed, write a proof of
that claim directly.  In the flag-algebra setting, this would mean proving, one
by one, that two particular flags are isomorphic, that a certain density is a
given rational number, or that a product of flag-algebra basis elements expands into a specified
linear combination.  This is possible in principle, but it is the wrong scale:
a single nontrivial application contains many such finite checks, and hand-proving each
one would bury the mathematical argument under routine case analysis.

The second approach is computation-oriented.  Instead of proving each finite
calculation from scratch, we prove once that an algorithm correctly performs the
kind of calculation we need, and then we let Lean run that algorithm inside later
proofs.  This is the idea of \emph{reflection}: for each such operation, a
program that carries out the calculation, paired with a theorem saying that the
program's output agrees with the abstract definition.  Later proofs run the
program and invoke the theorem, rather than reconstructing the finite
verification case by case.  Flag algebras are especially well suited to this
style, because the candidate bounds and their supporting numerical data are
found by an SDP solver, so validating the resulting proof already amounts to
checking finite data.  Such data should
neither be taken on trust nor replayed as a hand proof, and reflection is the
middle path: Lean recomputes the entries itself.

\Cref{sec:abstract} gave the specification layer: definitions close to the
mathematics, where flags are identified up to isomorphism and quotient classes
hide arbitrary choices of representatives.  Those definitions are the right ones
for stating theorems, but many are marked \lean{noncomputable} because they are
specifications rather than efficient programs.  The reflection layer of our Lean
formalization is their computable counterpart.

\subsection{Concrete Graph Representations and Finite Search}
\label{sec:reflection-concrete}

The finite coefficients in a flag-algebra proof come from concrete counting
problems.  Computing a density \(\den{F}{G}\) means carrying out the sampling
experiment of \Cref{sec:background-densities}: enumerate the candidate vertex
subsets of the host \(G\), test which of them induce a copy of the pattern
\(F\), and divide by the number of choices.  The downward normalizing factor
\(q_\sigma(F)\) is a similar finite count, but of label placements rather than
subflags.

An executable flag-algebra development therefore needs two basic ingredients: a
way to enumerate the finite objects being counted, such as candidate induced
subflags, and a way to decide yes/no questions, such as whether two flags are
isomorphic.  In Lean these are the \lean{Fintype} and \lean{Decidable}
typeclasses, and together they make a definition that counts the objects
satisfying a predicate executable.

The specification layer of \Cref{sec:abstract}, however, is not designed
around producing such executable instances.  Its starting point is Mathlib's graph type
\lean{SimpleGraph V}, whose adjacency relation returns a proposition:
\begin{lstlisting}
structure SimpleGraph (V : Type*) where
  Adj      : V -> (*@$\,$@*) V -> (*@$\,$@*) Prop
  symm     : Symmetric Adj
  loopless : Irreflexive Adj
\end{lstlisting}
This works well for the mathematical specification, but it is too general to be
the data structure of the evaluator: the vertex type \lean{V} need not come with
an algorithm for enumerating its vertices, and \lean{Adj} need not come with a
Boolean decision procedure.  \lean{LabeledGraph}s and \lean{Flag}s, built on top
of this graph type, inherit the same generality; knowing that only finitely many
flags of each size exist up to isomorphism supplies neither ingredient.

The reflection layer therefore does not use \lean{SimpleGraph} directly as the
evaluator's representation.  Instead, it introduces a family of more rigid
edge-set representations: an \(n\)-vertex graph has vertex type \lean{Fin n},
whose vertices \(0,1,\ldots,n-1\) Lean can enumerate, and its edge set is stored
as \lean{Finset (Sym2 (Fin n))}, an explicit finite set of unordered pairs
against which adjacency is a membership test.  This representation is packaged
as \lean{Sym2Graph} and reused to define the concrete flag type
\lean{Sym2FlagType} and labeled graph \lean{Sym2LabeledGraph}.
\begin{lstlisting}
structure Sym2Graph (n : (*@$\mathbb{N}$@*)) where
  edges       : Finset (Sym2 (Fin n))
  edges_valid : forall e in edges, (*@$\neg$@*) e.IsDiag

abbrev Sym2FlagType (k : (*@$\mathbb{N}$@*)) := Sym2Graph k

structure Sym2LabeledGraph {k : (*@$\mathbb{N}$@*)} ((*@$\sigma$@*) : Sym2FlagType k) (n : (*@$\mathbb{N}$@*)) where
  edges       : Finset (Sym2 (Fin n))
  edges_valid : forall e in edges, (*@$\neg$@*) e.IsDiag
  type_embed  : fromEdgeSet (SetLike.coe (*@$\sigma$@*).edges) (*@$\hookrightarrow$@*)g fromEdgeSet (SetLike.coe edges)
\end{lstlisting}
Mathlib's constructor \lean{fromEdgeSet} builds an abstract \lean{SimpleGraph}
from a set of unordered pairs, and \lean{SetLike.coe} presents the edge
\lean{Finset} as such a set, so \lean{fromEdgeSet (SetLike.coe edges)} is the abstract graph decoded
from the concrete edge data.  Between the two decoded graphs, the field
\lean{type\_embed} carries the same embedding \(\theta_F\) as the
\lean{LabeledGraph} of \Cref{sec:flagtypes}.

These basic choices then scale to the finite searches used by the evaluator.
To compute a density \(\den{F}{G}\) it must range over the placements of the
pattern \(F\) inside the host \(G\), each of which chooses a vertex set of
\(G\) containing all labeled vertices and thereby determines an induced subflag
to be tested against \(F\).  The following \lean{Fintype} instance enumerates
those placements:
\begin{lstlisting}
variable {k m n : (*@$\mathbb{N}$@*)} {(*@$\sigma$@*) : Sym2FlagType k}

structure Sym2InducedLabeledSubgraph (G : Sym2LabeledGraph (*@$\sigma$@*) n) where
  verts : Finset (Fin n)
  verts_subset : G.type_verts (*@$\subseteq$@*) verts

def Sym2InducedLabeledSubgraph.edges {G : Sym2LabeledGraph (*@$\sigma$@*) n}
    (H : Sym2InducedLabeledSubgraph G) : Finset (Sym2 (Fin n)) :=
  G.edges.filter (fun (e : Sym2 (Fin n)) =>
    forall (v : Fin n), v (*@$\in$@*) e -> (*@$\,$@*) v (*@$\in$@*) H.verts)

instance (G : Sym2LabeledGraph (*@$\sigma$@*) n) : Fintype (Sym2InducedLabeledSubgraph G) where
  elems := (Finset.univ : Finset (Finset (Fin n))).filterMap
    (fun (V : Finset (Fin n)) =>
      if hV : G.type_verts (*@$\subseteq$@*) V
      then some ({ verts := V, verts_subset := hV } : Sym2InducedLabeledSubgraph G)
      else none)
    (by ...)
  complete := fun (H : Sym2InducedLabeledSubgraph G) => ...
\end{lstlisting}
The type \lean{Sym2InducedLabeledSubgraph G} represents induced subflags
of \(G\): it stores the chosen vertex set together with a proof that this set
includes \lean{G.type\_verts}, the image of the type embedding.  Its edge set
is then computed by filtering
\(G\)'s explicit edge set to the edges whose endpoints both lie in the
chosen vertex set.  The \lean{Fintype (Sym2InducedLabeledSubgraph G)} instance
provides an enumeration of all induced subflags of \(G\).  Concretely,
the instance contains an \lean{elems} field listing all elements of the type \lean{Sym2InducedLabeledSubgraph~G},
together with a \lean{complete} proof that this enumeration misses none of
them.  Here the
enumeration is deliberately simple: since the vertices of \(G\) are
\lean{Fin n}, Lean enumerates all subsets of \(\{0,\ldots,n-1\}\), keeps
exactly those containing \lean{G.type\_verts}, and stores the retained subset
with the proof of containment.
What remains is the second basic ingredient: a decidable predicate that, given
a candidate \(H\), decides whether its labeled graph is isomorphic to the
pattern \(F\) with the type labels respected.

The definitions below isolate this test at a single pattern.  They are written
for exposition: our Lean development states the same condition for tuples of
patterns, as \lean{predIsoSym2LabeledHl}, and decides it the same way.  As
stated, \lean{candidateMatches F H} is a proposition, not a program; the
reflection step is the instance that shows Lean how to decide it from the
concrete edge data.
\begin{lstlisting}
def candidateMatches (F : Sym2LabeledGraph (*@$\sigma$@*) m) {G : Sym2LabeledGraph (*@$\sigma$@*) n}
    (H : Sym2InducedLabeledSubgraph G) : Prop :=
  Nonempty (H.toLabeledSubgraph.coe (*@$\simeq$@*)f F.toLabeledGraph)

instance candidateMatchesDecidable
    (F : Sym2LabeledGraph (*@$\sigma$@*) m) (G : Sym2LabeledGraph (*@$\sigma$@*) n)
    : DecidablePred (fun H : Sym2InducedLabeledSubgraph G => candidateMatches F H) :=
  fun H => by
    have : Fintype H.toLabeledSubgraph.subgraph.verts := by ...
    have : DecidableRel H.toLabeledSubgraph.coe.graph.Adj := by ...
    have : DecidableRel F.toLabeledGraph.graph.Adj := by ...
    exact if hsize : H.verts.card = m then
      (by
        dsimp [candidateMatches] -- unfold the definition of candidateMatches
        infer_instance)
    else
      isFalse (fun hIso => hsize (verts_card_of_coe_iso H F hIso))

def sym2LabeledGraphDensity(*@$_1$@*) (F : Sym2LabeledGraph (*@$\sigma$@*) m) (G : Sym2LabeledGraph (*@$\sigma$@*) n)
    [DecidablePred (fun H : Sym2InducedLabeledSubgraph G => candidateMatches F H)]
    : (*@$\mathbb{Q}$@*) :=
  ((Finset.univ : Finset (Sym2InducedLabeledSubgraph G)).filter
    (fun H => candidateMatches F H)).card / Nat.choose (n - k) (m - k)
\end{lstlisting}
The body of \lean{candidateMatches} compares two abstract \lean{LabeledGraph}s,
each decoded from its concrete representation.
\lean{F.toLabeledGraph} decodes the pattern as above, and the candidate is
decoded in two steps: \lean{H.toLabeledSubgraph} is the labeled subgraph of
\lean{G.toLabeledGraph} on the vertices \lean{H.verts}, and \lean{.coe} views
that subgraph as a labeled graph in its own right, forgetting that it sits
inside \(G\).  The two are then compared with the isomorphism type
\(\simeq_f\) of \Cref{sec:flagtypes}, and \lean{Nonempty} asserts that such an
isomorphism exists.

The proof of the instance is the computational heart of the example.  After the
local \lean{Fintype} and \lean{DecidableRel} instances have been assembled, the
\lean{hsize} branch reaches \lean{infer\_instance}, which asks typeclass search
to decide that \lean{Nonempty} proposition; the generic instance does so by
enumerating finitely many graph isomorphisms and checking the type-preservation
condition.  The predicate has therefore been reduced to a finite search.  The
\lean{hsize} branch itself is not needed for that, but it reproduces the
implementation's prefilter: a candidate whose vertex set has the wrong
cardinality cannot be isomorphic to \(F\), so it is rejected before the
isomorphism search runs.

Once the matching predicate is decidable, the density is ordinary finite
evaluation.  The function \lean{sym2LabeledGraphDensity}$_1$ filters the candidates by
it,
counts the survivors, and divides by the number of ways to choose the non-type
vertices.  The multi-flag density \(\den{F_1,\ldots,F_t}{G}\) of
\Cref{sec:formal-subflag-densities} follows the same pattern, with the
tuple-level predicate additionally requiring the chosen non-type vertices to be
pairwise disjoint.  In particular, the computable counterparts of
\(\mathtt{flagDensity}_1\) and \(\mathtt{flagDensity}_2\) are
\lean{sym2FlagDensity}$_1$ and \lean{sym2FlagDensity}$_2$, whose arguments have
type \lean{Sym2Flag}, the reflection-layer counterpart of \lean{Flag}
constructed in \Cref{sec:reflection-adequacy}.  Unlike the specification-level
density functions, these definitions do not carry the \lean{noncomputable}
keyword: Lean has an executable algorithm for evaluating them.  But
computability alone does not say that they implement the density functions from
the specification layer.  That correspondence is established by the adequacy
theorems of the next subsection.

\subsection{Adequacy Theorems for Reflection}
\label{sec:reflection-adequacy}

\Cref{sec:flagtypes} defined flags by quotienting labeled graphs by
isomorphism.
The reflection layer repeats the same construction for the concrete graph
representation from \Cref{sec:reflection-concrete}:
\begin{lstlisting}
variable {k m m(*@$_0$@*) m(*@$_1$@*) n : (*@$\mathbb{N}$@*)} {(*@$\sigma$@*) : Sym2FlagType k}

instance sym2LabeledGraphSetoid ((*@$\sigma$@*) : Sym2FlagType k) (n : (*@$\mathbb{N}$@*))
    : Setoid (Sym2LabeledGraph (*@$\sigma$@*) n) where
  r     := sym2LabeledGraphEqv -- related by a label-preserving isomorphism
  iseqv := ... -- reflexivity, symmetry, and transitivity of r

def Sym2Flag ((*@$\sigma$@*) : Sym2FlagType k) (n : (*@$\mathbb{N}$@*)) : Type :=
  Quotient (sym2LabeledGraphSetoid (*@$\sigma$@*) n)
\end{lstlisting}
We now have two representations of flags: \lean{Flag} in the specification
layer and \lean{Sym2Flag} in the reflection layer.  To use a reflection-layer
computation in a specification-layer theorem, we need a map from
\lean{Sym2Flag} to \lean{Flag}.  That map is \lean{Sym2Flag.toFlag}:
\begin{lstlisting}
def Sym2Flag.toFlag (G : Sym2Flag (*@$\sigma$@*) n)
    : Flag (fromEdgeSet (SetLike.coe (*@$\sigma$@*).edges)) (Fin n) :=
  Quotient.lift
    (fun G : Sym2LabeledGraph (*@$\sigma$@*) n => (*@$\llbracket$@*)G.toLabeledGraph(*@$\rrbracket$@*))
    ... -- proof that the lift respects the equivalence relation of the quotient
    G
\end{lstlisting}

With \lean{Sym2Flag.toFlag} in hand, we can carry out the reflection described
in \Cref{sec:reflection}.  The preceding subsection supplied the programs; the
\emph{adequacy theorems} below are the theorems that certify them, saying that
\lean{sym2FlagDensity$_1$} and \lean{sym2FlagDensity$_2$} agree with
\lean{flagDensity$_1$} and \lean{flagDensity$_2$} once the \lean{Sym2Flag}
arguments are decoded:
\begin{lstlisting}
theorem flagDensity(*@$_1$@*)_eq_sym2FlagDensity(*@$_1$@*) (F : Sym2Flag (*@$\sigma$@*) m) (G : Sym2Flag (*@$\sigma$@*) n)
    : flagDensity(*@$_1$@*) F.toFlag G.toFlag = sym2FlagDensity(*@$_1$@*) F G := ...

theorem flagDensity(*@$_2$@*)_eq_sym2FlagDensity(*@$_2$@*)
    (F(*@$_0$@*) : Sym2Flag (*@$\sigma$@*) m(*@$_0$@*)) (F(*@$_1$@*) : Sym2Flag (*@$\sigma$@*) m(*@$_1$@*)) (G : Sym2Flag (*@$\sigma$@*) n)
    : flagDensity(*@$_2$@*) F(*@$_0$@*).toFlag F(*@$_1$@*).toFlag G.toFlag = sym2FlagDensity(*@$_2$@*) F(*@$_0$@*) F(*@$_1$@*) G := ...
\end{lstlisting}
A computable definition with its adequacy theorem is what we call a
\emph{reflected} computation.  A specification-level claim can then be
rewritten into the reflection layer and evaluated there.  For example,
the density \(\den{\fKtwobull,\,\fbarKtwobull}{\fbarPthreebull} = 1/2\) of
\Cref{ex:typed-edge-nonedge-product} is proved as follows:
\Needspace{11\baselineskip}
\begin{lstlisting}
example : flagDensity(*@$_2$@*)
             Sym2Flag_2_1_0_0.toFlag
             Sym2Flag_2_1_0_1.toFlag
             Sym2Flag_3_1_0_1.toFlag = 1 / 2
  := by
  rw [flagDensity(*@$_2$@*)_eq_sym2FlagDensity(*@$_2$@*)]
  decide +kernel
\end{lstlisting}
The three constants are concrete \lean{Sym2Flag}s: \lean{Sym2Flag\_2\_1\_0\_0}
is \(\fbarKtwobull\), \lean{Sym2Flag\_2\_1\_0\_1} is \(\fKtwobull\), and
\lean{Sym2Flag\_3\_1\_0\_1} is \(\fbarPthreebull\).  Applying \lean{toFlag} to
each makes the statement a specification-level one, which \lean{rw} turns by
the adequacy theorem into the closed reflection-layer computation
\begin{lstlisting}
  sym2FlagDensity(*@$_2$@*)  Sym2Flag_2_1_0_0  Sym2Flag_2_1_0_1  Sym2Flag_3_1_0_1  =  1 / 2
\end{lstlisting}
The tactic \lean{decide +kernel} then closes this decidable goal by evaluating
that computation inside the Lean kernel.

The downward operator from \Cref{sec:formal-downward} follows the same
reflection pattern.  The specification-level definition uses the abstract
downward normalizing factor \(q_\sigma(F)\), implemented as
\lean{downwardNormalizingFactor}.  Its executable counterpart is
\lean{downwardNormalizingFactor\_Sym2Flag}.  On a concrete
\lean{Sym2LabeledGraph} over a flag type with \(k\) labels, on \(n\) vertices,
Lean enumerates the finite set of type embeddings that
produce an equivalent flag, takes its cardinality, and divides by the total
number \(n!/(n-k)!\) of injections of the \(k\) labels into the \(n\) vertices.
Because the computation is invariant under flag equivalence, it lifts through
the \lean{Sym2Flag} quotient.
\begin{lstlisting}
def isoEmbeddingCount_sym2LabeledGraph (G : Sym2LabeledGraph (*@$\sigma$@*) n) : (*@$\mathbb{N}$@*) :=
  (isoSym2TypeEmbeddingSetWithSameGraph G).card

def downwardNormalizingFactor_sym2LabeledGraph (G : Sym2LabeledGraph (*@$\sigma$@*) n) : (*@$\mathbb{Q}$@*) :=
  let num_of_all_injections := n.factorial / (n - k).factorial
  isoEmbeddingCount_sym2LabeledGraph G / num_of_all_injections

def downwardNormalizingFactor_Sym2Flag (F : Sym2Flag (*@$\sigma$@*) n) : (*@$\mathbb{Q}$@*) :=
  Quotient.lift
    (fun G => downwardNormalizingFactor_sym2LabeledGraph G)
    ... -- proof that downwardNormalizingFactor_sym2LabeledGraph respects the equivalence relation
    F
\end{lstlisting}
These are executable \lean{def}s rather than specification-level
\lean{noncomputable} ones.  We then prove that the specification-level and
executable implementations of the downward normalizing factor agree.
\begin{lstlisting}
theorem downwardNormalizingFactor_eq (F : Sym2Flag (*@$\sigma$@*) n)
    : downwardNormalizingFactor F.toFlag = downwardNormalizingFactor_Sym2Flag F := ...
\end{lstlisting}
Rewriting with it turns a goal about the specification-level coefficient into a
closed reflection-layer computation, which \lean{decide +kernel} evaluates as
before.

The reflected normalizing factor is what makes the downward operator executable
on concrete $\sigma$-flags.  Unlabeling such a flag yields the corresponding
$\emptyset$-flag by construction, so the coefficient is the only part that has
to be computed, and the adequacy theorem certifies
it.  By linearity, a finite linear combination of flag basis
elements is handled term by term, allowing later downward identities to be
established computationally rather than by separate counting arguments.

\subsection{Generating Flags and Their Identities}
\label{sec:reflection-generation}

The preceding subsection shows the basic unit of reflected computation: rewrite
a specification-level goal with an adequacy theorem, then evaluate the
resulting closed computation.  The adequacy theorems themselves are
uniform: they prove the reflected evaluators correct for every possible input.
Only their specialization to each concrete $\sigma$-flag and density instance
used by the later algebraic proof is repeated.

The generated flags are named by a fixed scheme, already visible in the
constants of \Cref{sec:reflection-adequacy}.  A name such as
\lean{Flag\_n\_k\_m\_i} or \lean{FlagAlgebra\_n\_k\_m\_i} carries four
indices: \(n\) is the size of the flag, \(k\) is the size of its flag type,
\(m\) is the position of the flag type's graph in the enumeration of the
\(k\)-vertex graphs, and \(i\) is the position of the flag itself in the
enumeration of the \(n\)-vertex flags of that flag type.  Both enumerations are
in a fixed canonical order, with positions counted from
zero.  Graphs are ordered first by their number of edges
and then by their canonical edge lists.  The canonical edge list of a graph is
the lexicographically smallest sorted edge list obtained by relabeling its
vertices in all possible ways.  The flags of a fixed type are grouped by their
underlying graph, in that same order, and within a group are ordered
lexicographically by the vertices the labels occupy.  The
reflection-layer constants \lean{Sym2Flag\_n\_k\_m\_i} follow the same
scheme, and the empty flag type is written \(k=m=0\).  Thus,
\lean{Sym2Flag\_3\_1\_0\_1} from \Cref{sec:reflection-adequacy} is flag
number \(1\) among the three-vertex flags of the unique one-vertex type,
namely \(\fbarPthreebull\), and \lean{FlagAlgebra\_3\_1\_0\_1} is its
flag-algebra basis element.

We automate this repetition with custom command elaborators written in Lean's
metaprogramming framework~\cite{moura2021lean}.  During elaboration,
these commands run the finite enumerators and generate proofs of the
specialized facts.  A generation block has the following form:
\begin{lstlisting}[emph={generate_empty_typed_flags,generate_flags,generate_flag_pair_density_theorems_no_forbid,generate_mul_theorems},emphstyle=\color{blue!70!black}\bfseries]
-- Generate empty-typed flags at pattern size 2 and host size 3.
generate_empty_typed_flags 2
generate_empty_typed_flags 3

-- Use the unique flag type on one labeled vertex (k = 1, m = 0),
-- and generate its flags at the sizes 2 and 3.
generate_flags 2 1 0
generate_flags 3 1 0

-- Prove all pair densities, then all product expansions.
generate_flag_pair_density_theorems_no_forbid 2 3 1 0
generate_mul_theorems 2 3 1 0
\end{lstlisting}

The numerical arguments are the \(n\), \(k\), \(m\) of the naming scheme; the
density and multiplication commands take a pattern size \(p\) and a host size
\(h\) before the \(k,m\) of the flag type, and the suffix \lean{no\_forbid}
says that no forbidden-graph filter is being applied.
The flag-generating commands also emit a completeness lemma: every
\(n\)-vertex \(\sigma\)-flag, up to isomorphism, appears among the generated
names.  That lemma specializes a theorem proved once and for all, by induction
on the size: our enumeration algorithm produces a representative of every
isomorphism class.
\Cref{tab:generation-commands} collects the command family, including the
forbid-free variants below.

\begin{table}
  \centering
  \caption{The generation commands.  \(n\) is a flag size and \(k\), \(m\)
  select the flag type, as in the naming scheme; the density and
  multiplication commands take a pattern size \(p\) and a host size \(h\).
  The forbid-free variants take a further argument \lean{H}, a
  \lean{Sym2Graph} term for the forbidden graph.}
  \label{tab:generation-commands}
  \small
  \begin{tabular}{@{}lp{0.32\textwidth}@{}}
    \toprule
    Command & Generates \\
    \midrule
    \lean{generate\_empty\_typed\_flags n} &
      the empty-typed flags on \(n\) vertices, with the completeness lemma \\
    \lean{generate\_flags n k m} &
      the \(n\)-vertex flags of flag type \((k,m)\), with unlabelings, downward
      coefficients, and the completeness lemma \\
    \lean{generate\_flag\_pair\_density\_theorems\_no\_forbid p h k m} &
      the pair-density theorems for size-\(p\) patterns in size-\(h\) hosts \\
    \lean{generate\_mul\_theorems p h k m} &
      the product expansion~\eqref{eq:mul} for pairs of size-\(p\)
      patterns in size-\(h\) hosts \\
    \midrule
    \lean{generate\_forbid\_free\_empty\_typed\_flags n H} &
      the \(H\)-free empty-typed flags, with the filtered completeness lemma \\
    \lean{generate\_forbid\_free\_flags n k m H} &
      the \(H\)-free typed flags, with unlabelings, downward coefficients,
      and the filtered completeness lemma \\
    \lean{generate\_forbid\_free\_flag\_pair\_density\_theorems p h k m H} &
      the pair-density theorems over the \(H\)-free flags \\
    \lean{generate\_forbid\_free\_mul\_theorems p h k m H} &
      the product expansions, modulo the \(H\)-free condition \\
    \lean{generate\_forbid\_free\_flag\_density\_theorems p i h H} &
      the single-flag densities of \lean{Flag\_p\_0\_0\_i} in every
      \(H\)-free size-\(h\) host \\
    \bottomrule
  \end{tabular}
\end{table}

Generating the typed flags also proves a downward identity for each of them,
such as \(q_\sigma(\fbarPthreebull)=2/3\), and the pair-density command's
family includes the identity
\(\den{\fKtwobull,\,\fbarKtwobull}{\fbarPthreebull}=1/2\) proved earlier.  The
multiplication command combines those coefficients with the completeness lemma
for the hosts to prove the product expansion
\([F_1] \cdot [F_2]=\sum_G \den{F_1,F_2}{G}\,[G]\) of~\eqref{eq:mul} for every
pattern pair; completeness is what justifies summing over exactly the named
hosts.

Every bound verified in \Cref{sec:flagmatic} is a bound on \(H\)-free graphs,
and each of the commands above has a forbid-free variant for that setting.  The
\lean{generate\_forbid\_free\_*} commands take one additional argument, a
concrete \lean{Sym2Graph} term for the forbidden graph \(H\); our Mantel
example, for instance, defines \lean{K3} as the complete graph on three
vertices and runs \lean{generate\_forbid\_free\_flags 3 1 0 K3}.  These
variants emit constants and theorems only for flags whose underlying graphs
avoid \(H\) as a subgraph, and their filtered completeness lemmas state that
the named flags exhaust precisely this set.

\paragraph{Optimizations.}
The following optimizations reduce the computational cost of elaboration.
They do not eliminate the inherent combinatorial growth of flag enumeration
and density computation, but make this paper's nontrivial examples feasible.

\begin{itemize}
\item \textbf{Size-inductive enumeration.}
The enumeration of flags is inductive in the size: the underlying
\(n\)-vertex graph representatives are obtained by extending the
\((n-1)\)-vertex representatives by one vertex in all possible ways, and the
flag type's labels are then placed on them in every valid
way.  The working lists therefore stay
close to one representative per isomorphism class, rather than covering all
\(2^{\binom{n}{2}}\) graphs on a fixed vertex set.

\item \textbf{Deduplication by cheap invariants.}
Both steps of that induction produce isomorphic duplicates: different vertex
extensions, and different label placements, can yield the same flag.  The
enumeration therefore deduplicates as it goes, keeping a candidate only when
it is isomorphic to no kept representative.  These isomorphism tests are the
expensive part, so each candidate carries a cheap isomorphism-invariant key,
derived from its edge count and degree sequence, and the expensive test runs
only when two keys collide.

\item \textbf{Pruning under a forbidden graph.}
For a complete forbidden graph, the size induction itself prunes: only
\(H\)-free representatives are extended, only \(H\)-free extensions kept,
and a cheap clique test decides containment, so flags containing \(H\) are
never generated.  Pruning misses no \(H\)-free isomorphism
class, since \(H\)-freeness is isomorphism-invariant and preserved by vertex
deletion, so the filtered completeness lemma again follows by
induction on size.

\end{itemize}

Metaprogramming automates theorem construction, not verification: Lean still
checks every generated proof, and a wrong coefficient makes that proof fail.
Once accepted, these theorems let later arguments simplify downward images,
densities, and products by rewriting rather than reopening the finite searches.
Thus, the story of this section has three steps: concrete representations make
the finite operations executable, adequacy theorems connect those operations
to the abstract definitions, and elaboration-time generation turns their
individual results into tables of proved facts.  \Cref{sec:flagmatic} uses
those tables to check certificates.

\subsection{A Bit-Mask Encoding for Kernel Evaluation}
\label{sec:reflection-bitmask}

The finite identities above are proved by evaluation, and how the evaluation is
carried out determines both what must be trusted and how far it scales.
Evaluation by \lean{decide +kernel} happens inside Lean's kernel and adds
nothing to what must be trusted, but it is slow.  The faster alternative,
\lean{native\_decide}, compiles the computation to native code and therefore
also trusts Lean's compiler.  For the five smaller certificates of
\Cref{sec:results}, the edge-set representation of
\Cref{sec:reflection-concrete} is fast enough for kernel evaluation.  For the
two largest, which need \(10{,}332\) and \(8{,}172\) pair-density identities,
kernel evaluation with this representation did not finish in acceptable time,
so these certificates could at first be checked only with
\lean{native\_decide}.  To check them in the kernel as well, we introduce a
more efficient bit-mask encoding.

The idea is to represent a graph by a single natural number, on which the
kernel computes much faster than on finite sets of edges.  List the pairs
\(a<b\) of vertices in lexicographic order, and let the \(i\)-th bit of the
number record whether the \(i\)-th pair is an edge.  On three vertices, for
instance, the pairs \(\{0,1\}\), \(\{0,2\}\), \(\{1,2\}\) correspond to bits
\(0\), \(1\), \(2\), so the path with edges \(\{0,1\}\) and \(\{0,2\}\) is the
number \(2^0+2^1=3\).  We call such a number a \emph{mask}; flags of a nonempty
type are encoded in the same way, with their labeled vertices listed first.  A
decoding map
turns each mask into an edge-set graph of \Cref{sec:reflection-concrete}, and
every such graph arises from some mask.  Facts proved about masks therefore
carry over to edge-set graphs, and from there to the specification layer, so
the new encoding adds a layer below the existing one rather than replacing it.
Isomorphism is also cheap to check on masks: given a permutation of the
vertices (for flags, one fixing the labeled vertices), one compares the bits of
one mask with the permuted bits of the other.

Masks serve two purposes.  The first is to prove that a list of graphs is
complete, that is, that every graph on \(n\) vertices is isomorphic to one on
the list, which the generated identities of \Cref{sec:reflection-generation}
rely on.  With masks this can be proved by brute force: a
\emph{canonicalization sweep} runs through all \(2^{\binom n2}\) masks on \(n\)
vertices and checks, for each, a permutation carrying it to a listed
representative.  The permutations are found outside Lean by an unverified
script and only checked by the kernel.

The second purpose is to compute densities, which become counting statements
about masks.  The decisive saving concerns the pair densities \(\den{F_1,F_2}{G}\).  Instead of
enumerating the pairs of vertex subsets of a host \(G\) once for every pattern
pair, a single kernel-checked theorem per host records the multiset of
canonical pattern pairs that all admissible subset pairs realize, and each pair
density is then read off as a multiplicity.  With this encoding, the two
largest certificates are checked entirely by the kernel.

\section{Verifying Flag-Algebra Certificates}
\label{sec:flagmatic}

Fix a forbidden graph \(H\), a target graph \(F\), and a proposed bound \(c\).
This section explains how a certificate proves
\[
  [F]\leq_{H}c\cdot\mathbf{1},
\]
the flag-algebra form of the asymptotic bound \(\tdensity{F}{H}\leq c\), and
how Lean checks such a certificate.  Throughout the section, \(\leq_H\) is the
ensemble order \lean{forbidLE H} of \Cref{sec:forbidden}; at the empty type,
\lean{generalizedTuranDensity\_le\_of\_forbidLE} turns the displayed
inequality into the Tur\'an-density bound.  A \emph{flag-algebra certificate}
is a finite witness for this inequality: it specifies flags, rational
matrices, and coefficients that purport to express the gap
\(c\cdot\mathbf{1}-[F]\) as a sum of terms that are non-negative under the
\(H\)-free condition.  Flagmatic~\cite{vaughan2020flagmatic} searches for such
witnesses by semidefinite programming.

Our concern is verification rather than search.  The theorem to be proved is
stated in a Lean source file, and the \lean{flag\_certificate} tactic reads a
certificate produced by Flagmatic and constructs a proof of exactly that
statement.  The certificate serves only as candidate data: every fact the
proof uses is checked by Lean.  Automating this step matters because even
small certificates involve extensive finite bookkeeping that would be tedious
and error-prone to reproduce by hand.

\subsection{How a Certificate Proves a Bound}
\label{sec:flagmatic-background}

At a high level, a flag-algebra certificate proves
$[F] \leq_{H} c\cdot\mathbf{1}$ by constructing two kinds of non-negative
correction term.  One kind consists of terms $Q_t$, each obtained by applying
the downward operator to a non-negative element of some $\mathcal{A}^\sigma$,
which a certificate supplies as a positive semidefinite quadratic form in
flags.  The other is a remainder $r$, a linear combination of $H$-free flags
with non-negative coefficients.  Under the $H$-free condition, these data
yield the chain
\begin{equation}
  \label{eq:cert-chain}
  [F]
  \leq_{H} [F]+\sum_t Q_t
  = c\cdot\mathbf{1}-r
  \leq_{H} c\cdot\mathbf{1}.
\end{equation}
The first inequality uses the non-negativity of the $Q_t$, the equality is the
finite calculation recorded by the certificate, and the final inequality uses
the non-negativity of $r$.  The proof of Mantel's theorem in
\Cref{ex:mantel-fa} is a certificate of exactly this shape, with a single term
$Q=\bigl\llbracket\tfrac12((\fKtwobull)-(\fbarKtwobull))^2\bigr\rrbracket_\sigma$
and remainder $r=\tfrac13\cdot\fbarPthree$, the two non-negative terms added to
$K_2$ in that example.

The same idea scales when the single square is replaced by a positive
semidefinite quadratic form.  Choose $\sigma$-flags $E_1,\ldots,E_r$, let
$e_i:=[E_i]\in\mathcal A^\sigma$ be their basis elements, and collect these
algebra elements into a vector $v=(e_1,\ldots,e_r)^\top$.  A symmetric matrix
$M\in\mathbb{R}^{r\times r}$ is \emph{positive semidefinite} (PSD) if
$x^\top Mx\geq 0$ for every $x\in\mathbb{R}^r$.  Consequently, when
$M=(M_{ij})$ is PSD, the flag-algebra quadratic form
$v^\top Mv=\sum_{i,j}M_{ij} e_i e_j$ is non-negative under every positive
homomorphism: evaluating the algebra elements turns it into the ordinary real
quadratic form $x^\top Mx$.  The downward operator then transfers this
non-negativity from the labeled algebra to the unlabeled one.

Flagmatic finds the matrices by turning the search for them into a
\emph{semidefinite program} (SDP).  It first fixes a host size $N$ and chooses
the flag types and flags that may occur.  It expands the target, the unit, and
the downward quadratic forms in the common basis of $H$-free $N$-vertex
graphs, with the matrix entries as unknowns.  Requiring the coefficient left
over for every host graph to be non-negative gives linear constraints on those
entries, and requiring each matrix to be PSD gives the semidefinite
constraints.  An SDP solver finds numerical matrices satisfying these
conditions, after which Flagmatic converts them into exact rational matrices.

A full certificate uses several flag types at once, so the single matrix $M$
above becomes a family: the index $t$ ranges over the chosen flag types
$\sigma_t$, each contributing its own flag vector $v_t$ and matrix $M_t$.  We
call one such triple an \emph{SDP block}; the name reflects that the matrices
$M_t$ are the diagonal blocks of the semidefinite program's matrix variable.
The downward quadratic form that block~$t$ produces is
\[
  Q_t := \bigl\llbracket
    v_t^\top M_t\, v_t
  \bigr\rrbracket_{\sigma_t}.
\]
The resulting certificate specifies an identity, modulo the $H$-free
condition, of the form
\[
  [F] + \sum_t Q_t
  =
  c\cdot\mathbf{1} - r,
\]
where every $M_t$ is PSD and $r$ is a linear combination of $H$-free flags
with non-negative coefficients, so these data justify the
chain~\eqref{eq:cert-chain}.  The equality is checked after all terms have been
expanded in the common basis at host size~$N$.  The certificate contributes
the bound $c$, the chosen flags, and the rational matrices $M_t$; the density
and product expansions, the positive semidefiniteness of the matrices, and
the final calculation are all established in Lean.

\subsection{Checking a Certificate in Lean}
\label{sec:flagmatic-shape}
\label{sec:flagmatic-pipeline}
\label{sec:flagmatic-gen}\label{sec:flagmatic-psd}\label{sec:flagmatic-expand}%
\label{sec:flagmatic-close}\label{sec:flagmatic-script}%

Turning a certificate into a verified proof raises four obligations, referred
to below as Parts~1--4:

\begin{enumerate}
  \item \textbf{Generating the flags and their identities.}  The proof must name
    the flags and corresponding flag-algebra basis elements it uses,
    restricted to those free of the forbidden graph, and must prove the
    density and product identities it needs.
  \item \textbf{Certifying the matrices.}  For each SDP block, the proof must
    show that the certificate's matrix $M_t$ is positive semidefinite.
  \item \textbf{Expanding the target (when needed).}  If the target graph $F$
    has fewer vertices than the host size $N$, the proof must expand $[F]$ in
    the size-$N$ basis.
  \item \textbf{Assembling the main theorem.}  The proof must assemble the
    chain~\eqref{eq:cert-chain} from the results of the preceding parts.
\end{enumerate}

Reusable proof-producing components handle each part, so a new bound requires
new certificate data rather than a new proof.  Parts~2--4 are carried out by
the \lean{flag\_certificate} tactic.  A source file states the desired bound
and proves it by naming the certificate file and the forbidden graph:
\begin{lstlisting}[emph={flag_certificate},emphstyle=\color{blue!70!black}\bfseries]
theorem Mantel_flagAlgebra
    : FlagAlgebra_2_0_0_1 (*@$\leq$@*)[completeGraph (Fin 3)] (1 / 2 : (*@$\mathbb{R}$@*)) (*@$\leansmul$@*) (1 : FlagAlgebra (*@$\emptyt$@*))
  := by
  flag_certificate "<mantel-cert>.json" K3
\end{lstlisting}
Here \lean{FlagAlgebra\_2\_0\_0\_1} is the basis element $[K_2]$ in the naming
scheme of \Cref{sec:reflection-generation}, and the relation is
\lean{forbidLE} for the forbidden triangle.  Around this theorem, the source
file defines the forbidden graph, runs the generation commands of
\Cref{sec:reflection-generation} that supply Part~1, and, after the theorem,
restates the bound as a Tur\'an-density bound (\Cref{sec:results}).  A script
in the artifact writes such files from a certificate, and a variant of the
tactic prints the synthesized proof in full, so that it can be checked without
the certificate file.  We now describe the four parts.

\paragraph{Generating the flags and their identities (Part~1).}
This is where the reflection layer of \Cref{sec:reflection} enters, making the
finite bookkeeping over flags, densities, and products executable and
checkable.  Part~1 defines the certificate's forbidden graph as a concrete
graph term and invokes the \lean{generate\_forbid\_free\_*} commands of
\Cref{sec:reflection-generation} on it.  The tactic matches each of the
certificate's compact graph strings against the resulting enumeration,
selecting the corresponding \lean{FlagAlgebra\_n\_k\_m\_i} constant, so
Flagmatic's ordering of the flags need not agree with Lean's.  The same
commands prove the identities that Parts~3 and~4 consume: a downward
normalizing factor $q_\sigma(F)$ for each emitted $\sigma$-flag, the product
expansions of the flag pairs the certificate uses, and, when the target must be
expanded in Part~3, the required single-flag densities.  The adequacy theorems
of \Cref{sec:reflection-adequacy} state all of them at the specification level.

\paragraph{Certifying the matrices (Part~2).}
For each SDP block, the certificate stores rational matrices $Q'_t$ and $R_t$
with $M_t=R_tQ'_tR_t^\top$, where $R_t$ may be rectangular and $Q'_t$ is a
smaller symmetric matrix.  The user's guide for
Flagmatic~1.5~\cite{vaughan2012flagmaticguide}, which documents this
certificate format, explains that $Q'_t$ is chosen to be positive definite,
which implies the positive semidefiniteness of $M_t$.

Our verification uses a different, equally elementary witness.  Every PSD
obligation is discharged through one library lemma, based on an exact rational
$LDL^\top$ factorization:
\begin{lstlisting}
theorem posSemidef_real_of_LDLt {n : (*@$\mathbb{N}$@*)}
    {M L : Matrix (Fin n) (Fin n) (*@$\mathbb{Q}$@*)} {d : Fin n -> (*@$\,\mathbb{Q}$@*)}
    (hd : forall (*@$\,$@*)i, 0 <= (*@$\,$@*)d i) (hM : M = L * Matrix.diagonal d * L(*@$^\top$@*))
    : (ratMatrixToReal M).PosSemidef := ...
\end{lstlisting}
The entrywise $\mathbb{Q}$-to-$\mathbb{R}$ cast, written
\lean{ratMatrixToReal}, happens inside the lemma, so the factorization itself
stays in exact rational arithmetic.  To use it, the tactic reconstructs $M_t$
from $Q'_t$ and $R_t$, computes a factorization
$M_t=L_t\operatorname{diag}(d_t)L_t^\top$ with $L_t$ square and unit lower
triangular, and inserts $M_t$, $L_t$, and $d_t$ as literals into the proof.
Both hypotheses of the lemma are then concrete decidable statements about
literals, the non-negativity of every entry of $d_t$ and an equality of finite
matrices over $\mathbb{Q}$, and \lean{decide +kernel} has the Lean kernel
evaluate them; no eigenvalue or floating-point calculation enters.  An $L_t$
or $d_t$ that does not factor $M_t$ is therefore rejected.  The
reconstruction of $M_t$ from the certificate is not itself verified, but it
cannot affect validity: the matrix proved positive semidefinite is the one
used in Part~4, and a wrong matrix makes the final calculation fail.

\paragraph{Expanding the target (Part~3).}
The certificate identity of \Cref{sec:flagmatic-background} is expressed at the
host size $N$, so Part~3 has nothing to do when the target graph $F$ already
has $N$ vertices.  When $F$ is smaller, the tactic proves an auxiliary identity
that expands $[F]$ over the $N$-vertex $H$-free $\emptyset$-flags,
\[
  [F]
  \;=\;
  \sum_{\substack{G \in \mathcal{F}^{\emptyset}_{N} \\ G\ H\text{-free}}}
  \den{F}{G}\,[G]
  \qquad\text{modulo the \(H\)-free condition,}
\]
with coefficients $\den{F}{G}$ given by the single-flag densities generated in
Part~1.  The restriction to the \(H\)-free condition matters because our
development works in the ambient algebra.  In the standard \(H\)-free algebra (\Cref{rem:builtin-hfree}), the
restricted sum would simply be the expansion
identity~\eqref{eq:expansion-identity}, all of whose flags are \(H\)-free
there.  In the ambient algebra, however, the flags containing \(H\) are nonzero
basis elements, so dropping their terms cannot give a literal equality; the
auxiliary identity instead says that every empty-type positive homomorphism
satisfying the \(H\)-freeness condition of \Cref{def:ensemble-semantic-order}
gives both sides the same value.  This is what lets the rest of the proof
expand over only the flags that Part~1 emitted.

\paragraph{Assembling the main theorem (Part~4).}
Part~4 realizes the chain~\eqref{eq:cert-chain}.  The first step is to prove
$[F]\leq_H[F]+\sum_tQ_t$, using the following lemma:
\begin{lstlisting}
theorem forbidLEWith_add_QuadraticForm {k n : (*@$\mathbb{N}$@*)} {(*@$\sigma$@*) : FlagType (Fin k)}
    {C : ForbidCondition} {f g : FlagAlgebra (*@$\emptyt$@*)}
    (M : Matrix (Fin n) (Fin n) (*@$\mathbb{R}$@*)) (hM : M.PosSemidef) (v : FlagAlgebraVec (*@$\sigma$@*) n)
    : forbidLEWith C f g -> (*@$\,$@*) forbidLEWith C f (g + (*@$\llbracket$@*)flagQuadraticForm M v(*@$\rrbracket_0$@*)) := ...
\end{lstlisting}
Here \lean{v} is the flag vector $v_t$ of \Cref{sec:flagmatic-background},
\lean{flagQuadraticForm M v} is its quadratic form $v_t^\top M_t v_t$,
\lean{forbidLEWith} is instantiated at \lean{forbiddenCondition H}, and
\lean{hM} is the positive-semidefiniteness proof from Part~2.  Starting from
the reflexive inequality $[F]\leq_H[F]$ and applying this lemma once for each
SDP block gives $[F]\leq_H[F]+\sum_tQ_t$.

It remains to prove $[F]+\sum_tQ_t\leq_H c\cdot\mathbf{1}$.  Both sides are
rewritten in the basis of $H$-free flags at the host size $N$, using the
Part~3 identity for $[F]$ when necessary.  Expanding each
$Q_t=\bigl\llbracket v_t^\top M_t v_t\bigr\rrbracket_{\sigma_t}$ produces
downward images of products of basis elements, which the
\lean{reduce\_downward\_flagmul} tactic rewrites using the multiplication and
downward identities of Part~1.  Then \lean{expand\_one\_hfree\_at} expands
$\mathbf{1}$ in the same basis, and \lean{flagsum\_ac\_sort\_rhs\_pipeline}
collects like terms.  The difference between the two sides is now exactly the
remainder $r$ of~\eqref{eq:cert-chain}; the tactic \lean{flag\_nonneg} checks
that its coefficients are non-negative and uses the non-negativity of every
basis element to conclude $0\leq_H r$.

\subsection{What Lean Checks}
\label{sec:flagmatic-trust}

Neither the certificate nor the tactic that reads it is trusted for the
validity of the result.  The tactic only proposes a proof, which Lean's kernel
checks: the matrices must be proved positive semidefinite, the calculation
must normalize to non-negative coefficients, and the final theorem must
typecheck.  An erroneous certificate, a mistranscribed target or forbidden
graph, or a bound smaller than the certificate supports therefore makes
elaboration fail; it cannot produce a false theorem.  The custom tactics are
unverified programs, and their dependence on the shape of elaborated
expressions affects robustness, but a tactic failure only leaves the proof
incomplete.

Every finite computation in the seven files of \Cref{sec:results} is evaluated
inside the kernel by \lean{decide +kernel}; the two largest certificates reach
this through the bit-mask encoding of \Cref{sec:reflection-bitmask}.  The
final theorems therefore depend only on the three standard axioms of Lean's
logic, \lean{propext}, \lean{Classical.choice}, and \lean{Quot.sound}.
What remains for a reader to check is that each final statement expresses the
intended mathematical claim; as \Cref{sec:results} explains, these statements
mention only specification-level definitions.

\subsection{The Seven Verified Bounds}
\label{sec:results}\label{sec:flagmatic-scenarios}

We verified seven certificates, spanning host sizes $N\in\{3,4,5\}$, targets
at and below the host size, between one and four SDP blocks, and four
forbidden graphs ($K_3$, $K_4$, $K_5$, and $C_5$).
Table~\ref{tab:flagmatic-scenarios} summarizes them.  All seven files compile
and contain neither \lean{sorry} nor axiom declarations, and no certificate
coefficient appears in the Lean source: the tactic reads the data at
elaboration time.  \Cref{sec:compile-times} reports compilation times and
peak memory.

\begin{table}
  \centering
\caption{Certificates verified in Lean.  \(H\) is the forbidden graph,
  \(n_F\) the number of vertices of the target graph \(F\), \(N\) the host
  size, and \(T\) the number of SDP blocks.}
  \label{tab:flagmatic-scenarios}
  \begin{tabular}{lccccc}
    \toprule
    Target density & $H$ & $n_F$ & $N$ & $T$ & Bound \\
    \midrule
    Edge density (Mantel)               & $K_3$ & 2 & 3 & 1 & $1/2$    \\
    $P_3$ density                       & $K_3$ & 3 & 3 & 1 & $3/4$    \\
    $C_4$ density                       & $K_3$ & 4 & 4 & 2 & $3/8$    \\
    Edge density                        & $K_4$ & 2 & 4 & 2 & $2/3$    \\
    $C_5$ density (Erd\H{o}s pentagon)  & $K_3$ & 5 & 5 & 3 & $24/625$ \\
    Edge density                        & $K_5$ & 2 & 5 & 4 & $3/4$    \\
    Edge density                        & $C_5$ & 2 & 5 & 4 & $1/2$    \\
    \bottomrule
  \end{tabular}
\end{table}

The cases play complementary roles.  Mantel is a hand-checkable baseline: its
single $2\times2$ SDP block is the square of \Cref{ex:mantel-fa}.  The
Erd\H{o}s pentagon certificate has three blocks, of dimensions $8$, $6$, and
$5$, and the $K_5$- and $C_5$-free edge-density certificates each use four
blocks of $8\times8$ matrices.  The $C_5$-free case shows that the forbidden
graph need not be complete.  All seven are Flagmatic's certificates as
produced, with no attempt to minimize the number of blocks or the matrix
dimensions.  Most of these bounds also have classical proofs; the clique-free
edge densities, for example, follow from Tur\'an's theorem.  Here they serve as tests of increasing size, while the
Erd\H{o}s pentagon bound is the one whose published proofs rely on flag
algebras.

Each file ends by restating its flag-algebra bound at the specification level:
a final theorem, obtained through
\lean{generalizedTuranDensity\_le\_of\_forbidLE}, states the Tur\'an-density
bound $\tdensity{F}{H}\le c$ with the target graph spelled out as an explicit
edge list.  This statement mentions none of the generated
\lean{FlagAlgebra\_n\_k\_m\_i} constants and can be read against the
definitions of \Cref{sec:abstract} alone.  The certificate establishes only
the upper bound.  For Mantel's theorem and the Erd\H{o}s pentagon theorem, the
matching lower bounds are proved directly in Lean (\Cref{sec:lowerbounds}),
completing the Tur\'an densities $\tfrac12$ and $\tfrac{24}{625}$ of
\Cref{ex:mantel,ex:pentagon}; the Erd\H{o}s pentagon case is, to our
knowledge, the first proof-assistant formalization of the Tur\'an-density
identity $\tdensity{C_5}{K_3}=\tfrac{24}{625}$.

\paragraph{Scope.}
\label{sec:flagmatic-open}
The seven cases establish coverage, not completeness.  The tactic accepts
Flagmatic certificates for graphs with a single forbidden graph, and even a
certificate in this scope may produce a normalization shape that the custom
tactics do not handle.  Such failures are safe: unsupported or malformed input
is rejected, and a tactic failure leaves the proof incomplete rather than
producing an unchecked theorem.

\section{Exact Tur\'an Densities and Goodman's Inequalities}
\label{sec:lowerbounds}

The certificates of \Cref{sec:flagmatic} prove upper bounds on Tur\'an
densities.  For Mantel's theorem and the Erd\H{o}s pentagon theorem, we prove
matching lower bounds in Lean, which turn these bounds into the exact values
$\tfrac12$ and $\tfrac{24}{625}$.  The formalization also applies beyond
Tur\'an densities: we prove two inequalities of Goodman that hold in all
graphs, deriving them from the Cauchy--Schwarz inequality for the downward
operator.

\paragraph{Mantel's theorem.}
\Cref{ex:mantel} states $\tdensity{K_2}{K_3}=\tfrac12$.  The Lean development
records it as follows:
\begin{lstlisting}
theorem Mantel_flag_bound : K2 (*@$\leq$@*) (1 / 2 : (*@$\mathbb{R}$@*)) (*@$\leansmul$@*) 1 + K3 := ...

theorem Mantel_flag_bound' : (*@$\forall$@*) ((*@$\phi$@*) : PositiveHom (*@$\emptyt$@*)), (*@$\phi$@*) K3 = 0 (*@$\to$@*) (*@$\phi$@*) K2 (*@$\leq$@*) 1 / 2 := ...

lemma extremal_density_K3_ge (n : (*@$\mathbb{N}$@*)) (hn2 : n (*@$\geq$@*) 2)
    : (extremalNumber n (completeGraph (Fin 3)) / n.choose 2 : (*@$\mathbb{R}$@*)) (*@$\geq$@*) 1 / 2 := ...

theorem Mantel_Turan : turanDensity (completeGraph (Fin 3)) = 1 / 2 := ...
\end{lstlisting}
Here \lean{K2} and \lean{K3} denote the basis elements $[K_2]$ and $[K_3]$,
and \lean{extremalNumber} and \lean{turanDensity} are Mathlib's definitions;
the former agrees with $\mathrm{ex}(n,\,K_2;\,K_3)$ because the induced copies
of $K_2$ are precisely the edges.  The certificate of \Cref{sec:results} also
yields the upper bound, but this proof uses neither the certificate tactic nor
the \lean{forbidLE} interface of \Cref{sec:forbidden}.

The upper-bound argument takes place in the ambient algebra and keeps the
triangle contribution explicitly:
\[
  K_2\leq\tfrac12\cdot\mathbf{1}+K_3.
\]
This inequality, \lean{Mantel\_flag\_bound}, is proved along the lines of
\Cref{ex:mantel-fa}.  If a positive homomorphism $\phi$ represents a
triangle-free limit, then $\phi([K_3])=0$, and evaluating the displayed
inequality gives $\phi([K_2])\leq\tfrac12$; this is
\lean{Mantel\_flag\_bound'}.

The matching lower bound, \lean{extremal\_density\_K3\_ge}, comes from the
balanced complete bipartite graph $K_{k,k}$, whose $k^2$ edges make up at
least half of the $\binom{2k}{2}$ vertex pairs; odd orders follow because the
normalized extremal numbers are non-increasing.  Finally,
\lean{Mantel\_Turan} combines the two bounds through
\Cref{lem:semantic-bound}.

\paragraph{The Erd\H{o}s pentagon theorem.}
The full result of \Cref{ex:pentagon},
\begin{lstlisting}
theorem ErdosPentagon_Turan : generalizedTuranDensity K3 C5 = 24/625 := ...
\end{lstlisting}
is the conjunction of the upper bound of \Cref{sec:results} and the lower
bound
\begin{lstlisting}
theorem ErdosPentagon_Turan_lowerBound : generalizedTuranDensity K3 C5 (*@$\geq$@*) 24/625 := ...
\end{lstlisting}
proved by an explicit construction in Lean.

By the convergence theorem \lean{tendsto\_generalizedTuranDensity} of
\Cref{sec:forbidden}, the normalized extremal numbers
$\mathrm{ex}(n,\,C_5;\,K_3)/\binom{n}{5}$ converge to $\tdensity{C_5}{K_3}$,
so it suffices to bound them from below for infinitely many $n$, here the
multiples of five.  The witness is the \emph{$n$-fold blow-up} of $C_5$:
replace each vertex of $C_5$ by a class of $n$ vertices, and join two vertices
of distinct classes exactly when the corresponding vertices of $C_5$ are
adjacent; vertices of a common class remain non-adjacent.  The Lean
development is:
\begin{lstlisting}
def blowUp {V : Type} [Fintype V] (G : SimpleGraph V) (n : (*@$\mathbb{N}$@*))
    : SimpleGraph (V (*@$\times$@*) Fin n) :=
  { Adj v w := G.Adj v.1 w.1
    symm v w := by apply G.symm }

theorem blowUp_K3_free {m : (*@$\mathbb{N}$@*)} {G : SimpleGraph (Fin m)}
    (n : (*@$\mathbb{N}$@*)) (hfree : K3.Free G) : K3.Free (blowUp G n) := ...

lemma subgraphCount_blowUp_C5_ge (n : (*@$\mathbb{N}$@*))
    : subgraphCount C5 (blowUp C5 n) (*@$\geq$@*) n ^ 5 := ...

theorem generalizedExtremalNumber_K3_C5_ge (n : (*@$\mathbb{N}$@*))
    : generalizedExtremalNumber (5 * n) K3 C5 (*@$\geq$@*) n ^ 5 := ...
\end{lstlisting}
Here \lean{blowUp G n} realizes the classes as the product type
\lean{V $\times$ Fin n}, with adjacency read off the first coordinates alone.
The blow-up of a triangle-free graph is triangle-free, since adjacent vertices
lie in distinct classes and the classes of a hypothetical triangle would form a
triangle in the base graph; this is \lean{blowUp\_K3\_free}.  Picking one
vertex from each of the five classes of the blow-up of $C_5$ yields five
vertices that inherit exactly the adjacencies of $C_5$, so each of the $n^5$
choice functions induces a copy of $C_5$.  \lean{subgraphCount\_blowUp\_C5\_ge}
records this count, where \lean{subgraphCount C5 G} counts the induced
subgraphs of \lean{G} isomorphic to $C_5$, and
\lean{generalizedExtremalNumber\_K3\_C5\_ge} records the resulting bound
$\mathrm{ex}(5n,\,C_5;\,K_3)\geq n^5$.  Normalizing gives, for every $n\geq1$,
\[
  \frac{\mathrm{ex}(5n,\,C_5;\,K_3)}{\binom{5n}{5}}
  \;\geq\;
  \frac{n^5}{\binom{5n}{5}}
  \;\geq\;
  \frac{n^5}{(5n)^5 / 5!}
  \;=\;
  \frac{24}{625}.
\]
One representation step remains.  The natural vertex set of
\lean{blowUp C5 n} is the product \lean{Fin 5 $\times$ Fin n}, while
\lean{generalizedExtremalNumber} maximizes over graphs on \lean{Fin (5 * n)}.
The graph \lean{blowUp\_fin} reindexes the blow-up on this vertex set; it is
proved isomorphic to \lean{blowUp} and is used as the witness.  Passing to the
limit along the sizes $5n$ then gives \lean{ErdosPentagon\_Turan\_lowerBound}.

\paragraph{Goodman's inequalities.}
Finally, we use the formalized theory of \Cref{sec:abstract} directly to prove
two inequalities outside Tur\'an-density problems.  Both are due to Goodman~\cite{goodman1959}.  The first
result is
\[
  K_3 \;\geq\; K_2\cdot(2K_2-\mathbf{1}),
\]
which bounds the triangle density from below in terms of the edge density:
under a positive homomorphism with edge density $e$, the triangle density is
at least $e(2e-1)$, a quantity that becomes positive as soon as $e$ exceeds
the Mantel threshold $\tfrac12$.  The proof combines the size-three expansion
of $K_2$ with the Cauchy--Schwarz inequality
$\llbracket f\,g\rrbracket_\sigma^{\,2}\leq
\llbracket f^2\rrbracket_\sigma\cdot\llbracket g^2\rrbracket_\sigma$ for the
downward operator, applied to the edge flag labeled at one endpoint.  That
inequality is a separate theorem in Razborov's
paper~\cite[Theorem~3.14]{razborov2007flag}, proved in the development in this general form
as \lean{Cauchy\_Schwarz\_inequality}.  The resulting inequality is formalized
as \lean{Goodman\_bound\_on\_triangle\_density}.

The second result is
\[
  \overline{K_3} + K_3 \;\geq\; \tfrac14\cdot\mathbf{1},
\]
where $\overline{K_3}$ is the empty graph on three vertices: asymptotically,
at least a quarter of all vertex triples induce a triangle or an independent
set.  Reading a graph and its complement as the two color classes of a
$2$-edge-colored complete graph turns this into Goodman's classical bound on
the number of monochromatic triangles; the proof reuses the downward square of
\Cref{ex:mantel-fa}.  It is formalized as
\lean{Goodman\_theorem\_on\_Ramsey\_multiplicity}.

\section{Engineering Obstacles}
\label{sec:obstacles}

The specification layer of \Cref{sec:abstract} turns several steps that
ordinary combinatorial notation leaves implicit into explicit proof
obligations.  These obligations do not change the underlying mathematics, but
discharging them repeatedly requires substantial proof-engineering effort.
This section examines three recurring sources of that effort.  First, because
flags are represented as isomorphism classes, every construction defined on
concrete graphs must be shown to be independent of the chosen representatives.
Second, some graph data, including flag sizes and the vertex types of tuple
entries, is encoded in Lean's types.  Reordering a tuple therefore changes not
only its entries but also the type-level order of their vertex types, and Lean
need not recognize the result as having the same type as a tuple constructed
directly in the new order.  Third, typeclass inference may construct different
terms representing the same finite enumeration, causing otherwise
straightforward rewrites to fail unless the relevant instances are controlled
explicitly.  The first three subsections examine these obstacles in turn.  The
fourth describes a separate proof pattern that became common in the
specification layer, proving finite counting identities by constructing
explicit bijections, and the last asks which of these obstacles are specific
to Lean.

\subsection{Computing with Graphs up to Isomorphism}
\label{sec:iso-quotient}

As described in \Cref{sec:flagtypes}, a flag is an isomorphism class of
concrete labeled graphs.  This representation ensures that later statements
do not depend on arbitrary vertex names, but it also imposes a recurring
well-definedness obligation.  A construction defined on concrete graphs gives
a construction on flags only after it has been shown to respect isomorphism.
Thus, for a rule $G\mapsto A(G)$, one must prove
that $G \simeq_f G'$ implies $A(G)=A(G')$.
In Lean, \lean{Quotient.lift} combines the rule with this compatibility proof
to produce a function on flags.  The quotient lift itself is brief; proving
the required compatibility is often the substantial part.

This obligation arises whenever a density or graph operation is first defined
on concrete labeled graphs and then lifted to flags.  It appears again at the
reflection boundary: \lean{Sym2Flag} is a quotient of computational graph
representations, so defining the decoding map
\lean{Sym2Flag.toFlag} of \Cref{sec:reflection-adequacy} requires proving that
equivalent concrete representations decode to the same abstract flag.  Once
this compatibility has been established, \lean{Quotient.lift} yields a
well-defined decoding map that the subsequent adequacy theorems can reuse
without exposing concrete representatives.

\subsection{Comparing Reordered Flag Tuples}
\label{sec:reordered-flag-tuples}

Part of a flag's data lives in its type, so objects that mathematics compares
freely can inhabit different types and fall outside ordinary equality.  What
the remedy costs depends on what the mismatch means.  A flag's size is such an
index: the expansion identity~\eqref{eq:expansion-identity} equates a size-$m$
flag with a combination of size-$n$ flags, and the dependent sum of
\Cref{sec:flagtypes} suffices, since it puts every size in
\lean{FinFlag}~$\sigma$ and the two sides are then distinct elements that the
quotient is there to identify, their difference being a generator of the zero
space.  A reordered tuple is the harder case.  A tuple's type records the
vertex types of its entries, so permuting the entries permutes the type too,
and Lean does not recognize the type so reached as the one it assigns to the
reordered tuple written out directly.  Here the two are meant to be equal
outright, so there is no quotient to appeal to: the types must be proved
equal first and the objects carried across that equality.

The multi-flag density is invariant under reordering its pattern flags.  One
instance is the symmetry
\(
  \den{F_0,F_1}{G}=\den{F_1,F_0}{G}
\),
which is what makes flag-algebra multiplication commutative.  On paper, the
invariance is proved once, in full generality: for any number of pattern
flags, reordering the tuple does not change its density in a host.  The
two-flag symmetry is then read off by taking the permutation that swaps the
two positions, an instantiation a mathematician would not count as a step of
the proof.  The formalization takes the same route, and the friction appears
exactly there: what costs nothing on paper is where the formal proof does its
work.

Recall from \Cref{sec:formal-subflag-densities} that a tuple of pattern
flags is a dependent function: an element of
\lean{FlagList}~$\sigma$~\lean{t Vl} assigns
a $\sigma$-flag on the vertex type \lean{Vl i} to each index \lean{i : Fin t}, and \lean{flagListDensity}
computes the multi-flag density of such a tuple in a host flag.  Both the
entries and their types vary with the index, and the vertex-type family
\lean{Vl} is itself an argument of the tuple's type.  Reordering is defined
by precomposition:
\begin{lstlisting}
abbrev Perm (t : (*@$\mathbb{N}$@*)) := Fin t (*@$\simeq$@*) Fin t

def listTypePermute {t : (*@$\mathbb{N}$@*)} (Vl : Fin t (*@$\to$@*) Type) ((*@$\pi$@*) : Perm t) : Fin t (*@$\to$@*) Type :=
  fun i => Vl ((*@$\pi$@*) i)

def FlagList.permute {T : Type} {(*@$\sigma$@*) : FlagType T} {t : (*@$\mathbb{N}$@*)} {Vl : Fin t (*@$\to$@*) Type}
    (Fl : FlagList (*@$\sigma$@*) t Vl) ((*@$\pi$@*) : Perm t) : FlagList (*@$\sigma$@*) t (listTypePermute Vl (*@$\pi$@*)) :=
  fun i => Fl ((*@$\pi$@*) i)
\end{lstlisting}
The permuted tuple places \lean{Fl ($\pi$ i)} at position \lean{i}, so the
entry types are reordered too: the vertex-type family of the result is
\lean{listTypePermute Vl $\pi$}, which appears in the type of the result.
The invariance is proved once, for
every length and every permutation; the symmetry to be derived from it is the
pair case:
\Needspace{12\baselineskip}
\begin{lstlisting}
variable {T W U(*@$_0$@*) U(*@$_1$@*) : Type} [Fintype T] [Fintype W] [Fintype U(*@$_0$@*)] [Fintype U(*@$_1$@*)]
         [DecidableEq W] [DecidableEq U(*@$_0$@*)] [DecidableEq U(*@$_1$@*)] {(*@$\sigma$@*) : FlagType T} 
         {t : (*@$\mathbb{N}$@*)} {Vl : Fin t (*@$\to$@*) Type} [FintypeList Vl] [DecidableEqList Vl]

theorem flagDensity_permute (Fl : FlagList (*@$\sigma$@*) t Vl) (G : Flag (*@$\sigma$@*) W) ((*@$\pi$@*) : Perm t)
    : flagListDensity Fl G = flagListDensity (Fl.permute (*@$\pi$@*)) G := ...

theorem flagPairDensity_comm (F(*@$_0$@*) : Flag (*@$\sigma$@*) U(*@$_0$@*)) (F(*@$_1$@*) : Flag (*@$\sigma$@*) U(*@$_1$@*)) (G : Flag (*@$\sigma$@*) W)
    : flagDensity(*@$_2$@*) F(*@$_0$@*) F(*@$_1$@*) G = flagDensity(*@$_2$@*) F(*@$_1$@*) F(*@$_0$@*) G := ...
\end{lstlisting}
To derive the second statement from the first, suppose the $\sigma$-flags
$F_0$ and $F_1$ have vertex types $U_0$ and $U_1$, and let \lean{Vl} be the
vertex-type family with \lean{Vl 0 = U}$_0$ and \lean{Vl 1 = U}$_1$, so that
\lean{[F}$_0$\lean{, F}$_1$\lean{]}$^{\mathtt{f}}$ is an element of
\lean{FlagList}~$\sigma$~\lean{2 Vl}.  Take $\pi$ to swap the two indices.
Then \lean{flagDensity\_permute} equates the density of
\lean{[F}$_0$\lean{, F}$_1$\lean{]}$^{\mathtt{f}}$ with that of its permuted
tuple, whose vertex-type family is \lean{listTypePermute Vl $\pi$}, whereas
\lean{flagPairDensity\_comm} speaks of the directly constructed pair
\lean{[F}$_1$\lean{, F}$_0$\lean{]}$^{\mathtt{f}}$, an element of
\lean{FlagList}~$\sigma$~\lean{2 Vl'} for the family \lean{Vl'} with
\lean{Vl' 0 = U}$_1$ and \lean{Vl' 1 = U}$_0$.  The two families are pointwise
equal, since both send $0$ to $U_1$ and $1$ to $U_0$, but Lean need not reduce
them to the same term, so the two tuples inhabit types that are not
definitionally the same.

The bridging lemma \lean{flagListDensity\_HEq\_eq} closes the gap:
\begin{lstlisting}
theorem flagListDensity_HEq_eq {Vl(*@$_0$@*) Vl(*@$_1$@*) : Fin t (*@$\to$@*) Type}
    [FintypeList Vl(*@$_0$@*)] [DecidableEqList Vl(*@$_0$@*)] [FintypeList Vl(*@$_1$@*)] [DecidableEqList Vl(*@$_1$@*)]
    {Fl(*@$_0$@*) : FlagList (*@$\sigma$@*) t Vl(*@$_0$@*)} {Fl(*@$_1$@*) : FlagList (*@$\sigma$@*) t Vl(*@$_1$@*)}
    (h_Vl_eq : Vl(*@$_1$@*) = Vl(*@$_0$@*)) (h_HEq : HEq Fl(*@$_0$@*) Fl(*@$_1$@*)) (G : Flag (*@$\sigma$@*) W)
    : flagListDensity Fl(*@$_0$@*) G = flagListDensity Fl(*@$_1$@*) G := ...
\end{lstlisting}
Its hypothesis \lean{h\_Vl\_eq} states \lean{Vl' = listTypePermute Vl $\pi$},
proved by checking the two indices of \lean{Fin 2}, while \lean{h\_HEq} uses
Lean's \emph{heterogeneous} equality, which relates terms whose types are not
yet known to be equal; here it says that both tuples place $F_1$ at position
$0$ and $F_0$ at position $1$.  Once the types are identified, that
heterogeneous equality can be used as an ordinary equality, and the lemma
concludes that the densities agree.  The
formal argument thus recovers what is immediate on paper: permuting
\lean{[F}$_0$\lean{, F}$_1$\lean{]}$^{\mathtt{f}}$ produces the same pair as
constructing \lean{[F}$_1$\lean{, F}$_0$\lean{]}$^{\mathtt{f}}$ directly.

\subsection{Controlling Finite Enumerations}
\label{sec:fintype}

Lean distinguishes knowing that a type is finite from having explicit finite
data listing all its elements.  An instance of \lean{Finite T} provides only
the former, whereas an instance of \lean{Fintype T} provides the latter:
\begin{lstlisting}
class Fintype (T : Type*) where
  elems : Finset T
  complete : (*@$\forall$@*) x : T, x (*@$\in$@*) elems
\end{lstlisting}
The field \lean{elems} contains every element of \lean{T}, as guaranteed by
\lean{complete}.  Our specification-layer definitions use this data when
taking cardinalities.  In Lean, the collections they count are typically
subtypes defined by predicates rather than explicit enumerations.  Examples
include tuples of pairwise-disjoint sets of unlabeled vertices in a host flag
that, together with the labeled vertices, induce the flags in a given flag
list, and type embeddings whose images produce a given flag.  Lean can
generally establish \lean{Finite} for these subtypes, and
\lean{Fintype.ofFinite} can then choose a \lean{Fintype} instance
noncomputably, using the axiom of choice.  Such an instance exists as a term
that proofs can mention, and every lemma about \lean{Fintype} applies to it,
but it carries no executable code.  This is sufficient here because the proofs
use the chosen instance only to state and compare cardinalities, not to
evaluate the enumeration; computing a cardinality or listing the elements
requires an executable instance, and the reflection layer of
\Cref{sec:reflection} builds its enumerations in that form.

\lean{Fintype.ofFinite} is not Lean's only way to construct the required
instance.  At later use sites, operations such as \lean{Fintype.card} leave
their \lean{Fintype} arguments implicit, and typeclass inference may construct
them by another available route.  For example, Lean can construct a
\lean{Fintype} for a predicate-defined subtype by filtering the existing
\lean{Fintype} enumeration of its ambient type.  We call the resulting instance
\emph{structurally derived} because it is assembled from the subtype structure
and the existing instance on the ambient type.  Consequently, a count
definition may contain the instance chosen locally by
\lean{Fintype.ofFinite}, while a cardinality lemma receives this structurally
derived instance.  Both enumerate the same elements and are provably equal,
but if their terms are not definitionally equal, a direct rewrite can fail.

The formalization responds in two ways.  The first is preventive: the moment a
finite set is named, the intended instance is chosen and bound with
\lean{let}, so that every later occurrence elaborates against the same
binding.  The fragment below proves that two finite sets
related by a bijection have equal \lean{toFinset} counts:
\begin{lstlisting}
-- S0 and S1 are the sets that are finite. They are being counted.
-- h_iso_S0_S1 a bijection between them.
let hS0 : Fintype S0 := Fintype.ofFinite S0
let hS1 : Fintype S1 := Fintype.ofFinite S1
have card_eq : @Fintype.card S0 hS0 = @Fintype.card S1 hS1 :=
  @Fintype.card_congr S0 S1 hS0 hS1 h_iso_S0_S1
have count_eq : S0.toFinset.card = S1.toFinset.card := by
  rw [Set.toFinset_card, Set.toFinset_card]
  exact card_eq
\end{lstlisting}
The proof names the \lean{Fintype} instance for each set and uses the bijection
to derive an equality of their cardinalities.  The \lean{@} notation makes the
instance arguments to \lean{Fintype.card} and \lean{Fintype.card\_congr}
explicit.  By contrast, the corresponding arguments of the two
\lean{toFinset} expressions remain implicit in the displayed code; typeclass
inference supplies \lean{hS0} and \lean{hS1} from the preceding bindings.  The
final two lines confirm that these choices agree: \lean{Set.toFinset\_card}
rewrites each count using the instance carried by its \lean{toFinset}
expression, and \lean{exact card\_eq} succeeds because those inferred instances
are definitionally the same as the explicit instances in \lean{card\_eq}.
Without the bindings, separate occurrences may receive different instances,
causing the final step to fail.

The proof of \lean{flagDensity\allowbreak\_permute}, the invariance theorem of
\Cref{sec:reordered-flag-tuples}, uses this combination of an explicit
bijection and fixed \lean{Fintype} instances.  For each concrete representative
of the host flag, \lean{S0} and \lean{S1} consist of tuples of subflags obtained
from pairwise-disjoint choices of unlabeled host vertices.  The tuples in
\lean{S0} realize \lean{Fl}, while those in \lean{S1} realize
\lean{Fl.permute $\pi$}.  Reindexing the tuple entries gives a bijection between
\lean{S0} and \lean{S1}, and the fixed instances allow the resulting
cardinality equality to rewrite the numerator of the density.  A related
bijective count proves that \lean{isomorphismCount}, used in the downward
operator of \Cref{sec:formal-downward}, is invariant under replacing a labeled
graph by an isomorphic representative: the isomorphism transports label
placements and gives a bijection between the corresponding sets.

The second approach is corrective.  After the underlying mathematical objects
have been identified, Lean may still present an equality whose two sides look
identical but contain different hidden instance arguments.  The tactic
\lean{congr!} decomposes such an equality into congruence subgoals and proves
that the instances agree.  This is possible because \lean{Fintype T} is a
subsingleton, so any two \lean{Fintype} instances of the same type are provably
equal; the same holds for \lean{DecidableEq T}.  This is how
\lean{flagListDensity\allowbreak\_HEq\allowbreak\_eq} from
\Cref{sec:reordered-flag-tuples} concludes.  After the equality of the
vertex-type families and the heterogeneous equality of the flag tuples have
been used to identify the inputs, the remaining difference lies only in the
\lean{FintypeList} and \lean{DecidableEqList} instances supplied for their
component types, and \lean{congr!} resolves it.

Explicitly chosen \lean{Fintype} instances are therefore needed only when a
collection specified by a predicate is enumerated as a \lean{Finset}, which
stores an explicit finite enumeration, or when its cardinality is compared
with that of another collection.  Once the required counting identity has
been proved, the subsequent algebraic argument uses that identity without
referring to the chosen instances.

\subsection{Proving Counting Identities by Bijections}
\label{sec:proof-style}

Explicit bijections became an unexpectedly prevalent proof pattern in the
specification layer.  We expected the counting equalities the layer requires
to follow by rewriting with existing counting lemmas or by symbolic
simplification, but case after case, unfolding the definitions exposed the
same structure: the two sides are cardinalities of finite sets of
configurations presented in different ways.  Our most reliable method is to
identify the two sets and construct an explicit bijection between them, which
Mathlib's cardinality lemmas then convert into the required equality.

Proving the chain rule (\Cref{lem:chain-rule}) made the pattern
clear.  Its two sampling procedures are visibly two descriptions of one
experiment on paper; in Lean, after unfolding the definitions, the proof came
down to a bijection between the finite sets of choices the two procedures
make, which yielded the equality of counts and hence the rule.

The same method proves permutation invariance, invariance under inserting an
empty flag, the downward-operator counting identities, and the adequacy of
the computational subgraph representation.  In each, the bijection carries the
mathematical argument, while the instance discipline of \Cref{sec:fintype}
handles the resulting cardinalities.  We do not claim bijections are the only
approach, but their effectiveness across these obligations made them our
default method.

\subsection{Which Obstacles Are Specific to Lean?}
\label{sec:lean-specific}

The obligations above depend on Lean to different degrees.  The
well-definedness obligation of \Cref{sec:iso-quotient} for constructions on
isomorphism classes would arise in any proof assistant that represents flags
as quotients.  The remaining ones are tied more closely to Lean.  Binding a
particular \lean{Fintype} instance is needed because Lean's elaborator may
synthesize different instance terms for different occurrences of the same
type, and the use of heterogeneous equality reflects Lean's dependent type
theory, in which pointwise-equal type families need not be definitionally
equal.  Other proof assistants handle finite enumeration, implicit arguments,
and dependent equality differently, so we do not know which of these
techniques a port to Rocq, Agda, or Isabelle/HOL would need.  The combinatorial identities and
the well-definedness obligations created by quotients would remain, but their
formal expression and cost could differ.

\section{A Meta-Theory of Ensemble Semantics}
\label{sec:metatheory}

Defining \lean{forbidLE} as we did in \Cref{sec:forbidden}, by constraining
the interpretations rather than the algebra, exposed a mathematical problem in
its own right: when does this \emph{ensemble semantics} validate exactly the
same typed inequalities as the standard semantics, in which the constraint is
built into the algebra from the beginning?  This section gives the definitions
and main results needed to state the answer.  The proofs are only sketched;
a companion paper, in preparation, will present the full theory and its
applications, and the results reported here have been autoformalized in Lean.

\paragraph{The constrained positive-homomorphism space.}
Let \(\mathcal G\) be the class of all finite simple graphs, let
\(\mathcal K\subseteq\mathcal G\) be a hereditary class, and fix a flag type
\(\sigma\).  Here \emph{hereditary} means closed under taking induced
subgraphs: if~\(G\in\mathcal K\) and \(S\subseteq V(G)\), then
\(G[S]\in\mathcal K\).  In particular, the \(H\)-free classes used earlier are
hereditary, because taking an induced subgraph cannot create a new subgraph
isomorphic to \(H\).

We write \(\FlagAlg{\sigma}[\mathcal G]\) for the ambient \(\sigma\)-typed flag
algebra, and we define the constrained algebra \(\FlagAlg{\sigma}[\mathcal K]\)
as its quotient by the ideal\footnote{An ideal \(I\) here is a linear subspace
of the flag algebra that is also closed under multiplication by arbitrary
algebra elements.  That closure is what makes the quotient by \(I\) again an
algebra, one in which every element of \(I\) has become zero.}
\(I_{\mathcal K}\) generated by the \(\sigma\)-flags whose underlying graphs
lie outside \(\mathcal K\), with quotient map
\[
  q_\sigma:
  \FlagAlg{\sigma}[\mathcal G]
  \twoheadrightarrow
  \FlagAlg{\sigma}[\mathcal K],
\]
which sends each of these flags to zero.  Because \(\mathcal K\) is
hereditary, \(I_{\mathcal K}\) is in fact the linear span of these flags, and
\(\FlagAlg{\sigma}[\mathcal K]\) is canonically isomorphic to the algebra that
\Cref{sec:background} builds from the \(\sigma\)-flags in \(\mathcal K\)
alone.

Let
\[
  X_\sigma
  :=
  \mathrm{Hom}^{+}\bigl(\FlagAlg{\sigma}[\mathcal G],\R\bigr)
\]
be the compact space of positive homomorphisms on the ambient algebra, with the
topology of pointwise convergence on \(\sigma\)-flags.  Pullback along
\(q_\sigma\) embeds the positive homomorphisms of the constrained algebra into
\(X_\sigma\).  Its image is
\begin{equation}\label{eq:Q-sigma}
  Q_\sigma
  :=
  \left\{
    \psi\circ q_\sigma
    \;\middle|\;
    \psi\in
    \mathrm{Hom}^{+}\bigl(\FlagAlg{\sigma}[\mathcal K],\R\bigr)
  \right\}
  \subseteq X_\sigma.
\end{equation}
Thus, \(Q_\sigma\) is the space of all \(\sigma\)-typed positive homomorphisms
admitted by the constrained algebra.  There is also an intrinsic description:
\begin{equation}\label{eq:Q-sigma-intrinsic}
  Q_\sigma
  =
  \left\{
    \chi\in X_\sigma
    \;\middle|\;
    \chi([F])=0
    \text{ for every \(\sigma\)-flag \(F\) whose underlying graph lies outside
    \(\mathcal K\)}
  \right\}.
\end{equation}
The inclusion of \eqref{eq:Q-sigma} in \eqref{eq:Q-sigma-intrinsic} is
immediate, since \(q_\sigma\) kills every \(\sigma\)-flag outside
\(\mathcal K\).  Conversely, a homomorphism vanishing on these flags vanishes
on the ideal \(I_{\mathcal K}\) they generate, because its kernel is an ideal,
so it factors through \(q_\sigma\).  The description also exhibits \(Q_\sigma\)
as an intersection of zero sets of continuous evaluation maps, hence as a
closed subset of \(X_\sigma\).  We write \(Q_0:=Q_{\emptyset}\) for the
corresponding space of constrained empty-type positive homomorphisms.

\paragraph{Random extensions and the root-planting space.}
By Razborov's random-extension theorem recalled in
\Cref{sec:background-downward}, every \(\phi_0\in Q_0\) with
\(\phi_0(\langle\sigma\rangle_0)>0\) determines a probability measure
\(\Ext_\sigma(\phi_0)\) on \(X_\sigma\).  We call \(\sigma\)
\emph{non-degenerate} for \(\mathcal K\) if such a \(\phi_0\) exists.  For a
probability measure~\(\mu\) on~\(X_\sigma\), its support
\(\operatorname{supp}(\mu)\) consists of those points \(\chi\in X_\sigma\) for
which every open neighborhood of \(\chi\) has positive \(\mu\)-measure.
Define
\begin{equation}\label{eq:S-sigma}
  S_\sigma
  :=
  \overline{
    \bigcup_{\substack{\phi_0\in Q_0\\
                        \phi_0(\langle\sigma\rangle_0)>0}}
    \operatorname{supp}\bigl(\Ext_\sigma(\phi_0)\bigr)
  }
  \subseteq X_\sigma,
\end{equation}
the closure of the union of the supports of all admissible random extensions;
it is empty when \(\sigma\) is degenerate.

The support of each random-extension measure is contained in the constrained
space:
\[
  \operatorname{supp}\bigl(\Ext_\sigma(\phi_0)\bigr)
  \subseteq Q_\sigma
  \qquad
  \bigl(\phi_0\in Q_0,\
        \phi_0(\langle\sigma\rangle_0)>0\bigr).
\]
Indeed, for a \(\sigma\)-flag \(F\) outside \(\mathcal K\), the
downward image of \([F]\) is a non-negative multiple of \([F|_\emptyset]\),
which every \(\phi_0\in Q_0\) evaluates to zero, so the random-extension
identity~\eqref{eq:ensemble} forces \(\mathbb{E}[\phi([F])]=0\) and hence
\(\phi([F])=0\) almost surely.  Ranging over the countably many such flags puts
the sample in \(Q_\sigma\) with probability one, and \(Q_\sigma\) is closed, so
Equation~\eqref{eq:S-sigma} yields the fundamental inclusion
\begin{equation}\label{eq:S-subset-Q}
  S_\sigma\subseteq Q_\sigma.
\end{equation}
We say that \((\mathcal G,\mathcal K,\sigma)\) is \emph{root-plantable} when
\(S_\sigma=Q_\sigma\).  The name records what the equality says: the labeled
vertices of \(\sigma\) are its \emph{roots}, and to \emph{plant} a
\(\sigma\)-flag's density profile in a larger graph of the class, without
choosing its labels in advance, is to make a fraction bounded away from zero of
all embeddings of \(\sigma\) reproduce that profile.  Root-plantability asks
that every constrained \(\sigma\)-typed positive homomorphism be approximable
by finite constructions of this kind.

\begin{definition}[Quotient and ensemble orders]\label{def:two-constrained-orders}
Let \(\mathcal K\subseteq\mathcal G\) be hereditary, let \(\sigma\) be a flag
type, and let \(f,g\in\FlagAlg{\sigma}[\mathcal G]\).  Define
\[
  f\leq^{\mathrm{quot}}_{\mathcal K,\sigma}g
  \quad\Longleftrightarrow\quad
  \chi(f)\leq\chi(g)
  \quad\text{for every }\chi\in Q_\sigma,
\]
and
\[
\begin{aligned}
  f\leq^{\mathrm{ens}}_{\mathcal K,\sigma}g
  \quad\Longleftrightarrow\quad
  \Ext_\sigma(\phi_0)
  \bigl(\{\chi\in X_\sigma\mid \chi(f)\leq\chi(g)\}\bigr)=1
  \quad\text{for every }\phi_0\in Q_0\text{ with }
  \phi_0(\langle\sigma\rangle_0)>0.
\end{aligned}
\]
\end{definition}

The first is the standard order of \Cref{rem:builtin-hfree}, which we call the
\emph{quotient order} because it is the ordinary semantic order of the
constrained quotient algebra: equivalently, every positive homomorphism
\(\psi:\FlagAlg{\sigma}[\mathcal K]\to\R\) satisfies
\(\psi(q_\sigma(f))\leq\psi(q_\sigma(g))\).  The second requires, for each
admissible \(\phi_0\), that the random positive homomorphism
\(\chi\sim\Ext_\sigma(\phi_0)\) satisfy \(\chi(f)\leq\chi(g)\) with probability
one.  For the class of \(H\)-free graphs it is the order
\(\leq^{\mathrm{ens}}_{H,\sigma}\) of \Cref{def:ensemble-semantic-order}.

\begin{theorem}[Support-closure criterion]\label{thm:support-closure-summary}
Let \(\mathcal K\) be hereditary and let \(\sigma\) be a flag type.  For every
\(f,g\in\FlagAlg{\sigma}[\mathcal G]\),
\[
  f\leq^{\mathrm{quot}}_{\mathcal K,\sigma}g
  \quad\Longrightarrow\quad
  f\leq^{\mathrm{ens}}_{\mathcal K,\sigma}g.
\]
Moreover, the two orders agree for every pair \((f,g)\) if and only if the
class is root-plantable at \(\sigma\):
\[
  \left(
    \forall f,g\in\FlagAlg{\sigma}[\mathcal G],\quad
    f\leq^{\mathrm{quot}}_{\mathcal K,\sigma}g
    \ \Longleftrightarrow\
    f\leq^{\mathrm{ens}}_{\mathcal K,\sigma}g
  \right)
  \quad\Longleftrightarrow\quad
  S_\sigma=Q_\sigma.
\]
\end{theorem}

Both statements of the theorem come from one observation: each order asks
that the same continuous function \(\chi\mapsto\chi(g-f)\) be non-negative, and
they differ only in the set of \(\chi\) on which it is tested.  The quotient
order tests \(Q_\sigma\) by definition.  The ensemble order asks for
almost-sure non-negativity under each \(\Ext_\sigma(\phi_0)\), which is
non-negativity on that measure's support; because the function is continuous,
testing it on those supports is the same as testing it on their closure
\(S_\sigma\).  The first statement is therefore exactly the inclusion
\(S_\sigma\subseteq Q_\sigma\) of Equation~\eqref{eq:S-subset-Q}, and
\(S_\sigma=Q_\sigma\) makes the two tests identical.  When the inclusion is
strict, a point of \(Q_\sigma\setminus S_\sigma\) can be separated from the
closed set \(S_\sigma\) by a continuous function, and the Stone--Weierstrass
theorem approximates that function by the evaluation of a flag-algebra
element.  This element is non-negative in the ensemble order but not in the
quotient order.  The same reading covers a degenerate type \(\sigma\): there
\(S_\sigma\) is empty, so the ensemble order relates every pair, and the two
orders agree exactly when \(Q_\sigma\) is empty.

\paragraph{When root-plantability holds.}
The criterion in \Cref{thm:support-closure-summary} reduces the agreement of
the two orders to a structural problem about the graph class.  Our main positive
result gives a broad answer.  Given \(G\in\mathcal K\), a vertex~\(v\in V(G)\),
and a graph \(J\), write \(G[v\mapsto J]\) for the graph obtained by replacing
\(v\) by \(J\), keeping the edges inside \(J\), and joining every vertex of
\(J\) to every former neighbor of \(v\).  Call \(\mathcal K\)
\emph{blow-up-closed} if, for every \(G\in\mathcal K\), \(v\in V(G)\), and
\(N \in \mathbb{N}\), some graph \(J\) on exactly \(N\) vertices satisfies
\(G[v\mapsto J]\in\mathcal K\).  For a hereditary class, this is equivalent to
asking only for \(\lvert J\rvert\geq N\): choose an induced \(N\)-vertex
subgraph \(J'\) of \(J\), and observe that \(G[v\mapsto J']\) is an induced
subgraph of \(G[v\mapsto J]\).  The interior graph \(J\) may depend on \(G\),
\(v\), and \(N\); the definition requires the existence of a suitable interior,
not that every interior works.

\begin{theorem}[Blow-up closure implies root-plantability]
\label{thm:blowup-root-plantability-summary}
Every blow-up-closed hereditary graph class is root-plantable at every nonempty
flag type.  The empty flag type is root-plantable for every hereditary graph
class, independently of blow-up closure.  Consequently, for a blow-up-closed
class the quotient and ensemble orders agree at every flag type.
\end{theorem}

We sketch the first claim.  If \(\sigma\) is degenerate, its underlying graph
lies outside \(\mathcal K\), since blowing it up inside \(\mathcal K\) would
give it positive density; then \(Q_\sigma\) and \(S_\sigma\) are both empty.
It therefore suffices to treat a non-degenerate \(\sigma\).  Since
\(S_\sigma\subseteq Q_\sigma\) always holds, root-plantability amounts to
placing every \(\chi\in Q_\sigma\) in \(S_\sigma\), and the gap to close is
between a chosen labeling and a random one: \(\chi\) is a limit of
\(\sigma\)-flags in \(\mathcal K\), whose labels may be placed to suit it,
while \(S_\sigma\) asks that a \emph{uniformly chosen} labeling reproduce its
density profile with probability bounded away from zero.  Blow-ups close that
gap: blow-up closure lets us replace every vertex of a \(\sigma\)-flag in
\(\mathcal K\) by a large interior while staying inside \(\mathcal K\).  A
uniformly chosen embedding of \(\sigma\) into the blow-up then places each root
in the interior of the corresponding labeled vertex with probability bounded
below independently of the size of the blow-up, while any fixed number of
sampled vertices land in distinct interiors with probability tending to one,
so the blow-up leaves the sampled densities unchanged.  Given a basic open
neighborhood \(U\) of \(\chi\), choose a \(\sigma\)-flag far enough along a
sequence representing \(\chi\) and blow it up with growing interiors; a limit
\(\phi_0\) of the resulting unlabeled graphs satisfies
\(\Ext_\sigma(\phi_0)(U)>0\).  Every neighborhood of
\(\chi\) therefore meets the support of some admissible random extension, so
\(\chi\in S_\sigma\).

Independent blow-ups, in which every replacement graph \(J_v\) is an
independent set, cover all \(K_r\)-free classes for \(r\geq3\), including the
triangle-free graphs.  Complete blow-ups, in which every \(J_v\) is a clique,
cover, for example, the graphs whose components are cliques.  Both are
instances of \Cref{thm:blowup-root-plantability-summary}.

\paragraph{Why the criterion is nontrivial.}
Heredity alone does not imply root-plantability.  Let \(\mathcal K\) be the
class of \(C_4\)-free graphs and let \(\sigma_1\) be the one-vertex type.
Every large \(C_4\)-free graph has \(o(n^2)\) edges, so a vertex chosen
uniformly at random has normalized degree tending to zero.  Let \(d:=\fKtwobull\in\FlagAlg{\sigma_1}[\mathcal G]\), so that
the density of \(d\) in a finite \(\sigma_1\)-flag is the normalized degree of
its labeled vertex.  Its downward image is \(K_2\), so the random-extension
identity gives \(\mathbb{E}[\chi(d)]=\phi_0(K_2)=0\) for
\(\chi\sim\Ext_{\sigma_1}(\phi_0)\), and since \(\chi(d)\geq0\), in fact
\(\chi(d)=0\) almost surely.  Every \(\chi\in S_{\sigma_1}\) therefore
satisfies \(\chi(d)=0\), and \(-d\) is non-negative in the ensemble order.  But
arbitrarily large stars are \(C_4\)-free: the star \(K_{1,t}\) consists of one
vertex, its \emph{center}, adjacent to \(t\) pairwise non-adjacent leaves.
Label the center to obtain a \(\sigma_1\)-flag, and let
\(\chi_*\in Q_{\sigma_1}\) be the limit of that sequence.  Every leaf is adjacent to the center, so \(\chi_*(d)=1\) and
\(\chi_*\notin S_{\sigma_1}\), hence \(S_{\sigma_1}\subsetneq Q_{\sigma_1}\).
Moreover, \(\chi_*(-d)=-1<0\), so \(-d\) is not non-negative in the quotient
order.  The example exhibits a genuine gap between the two typed orders: the
quotient order permits a specially chosen vertex to carry the label, whereas
the ensemble order labels a uniformly sampled vertex.

The gap does not affect the density theorems of this paper.  The certificate
proofs of \Cref{sec:flagmatic} use typed inequalities only through the downward
operator.  Each positive semidefinite term is non-negative under every positive
homomorphism, hence in the ensemble order, and the random-extension
identity~\eqref{eq:ensemble} carries ensemble non-negativity to the empty
type, as in the proof of \Cref{thm:downward-nonneg}.  At the empty type the
two orders agree, since the empty type is root-plantable for every hereditary
class by \Cref{thm:blowup-root-plantability-summary}, so
\Cref{lem:semantic-bound} turns the conclusions into statements about
Tur\'an densities.  Root-plantability is therefore not a hypothesis of any
density theorem in this paper.  It concerns only whether the intermediate
ensemble calculus proves exactly the typed inequalities of the constrained
algebra.

This section has isolated the core consequences of the ensemble semantics:
inequalities valid in the quotient order are always valid in the ensemble
order, root-plantability characterizes when the converse holds, and blow-up
closure is a broad sufficient condition for it.  The companion paper will go
further, with a finite criterion implying root-plantability, further classes on
either side of the boundary, and a relative ensemble semantics that admits
non-hereditary constraints such as a fixed edge density.

\paragraph{Formalization and scope.}
The Lean formalization of this section is part of the released development,
under \lean{LeanFlagAlgebras.MetaTheory}.
There \(Q_\sigma\), \(S_\sigma\),
and the condition \(S_\sigma=Q_\sigma\) are \lean{Q$\sigma$}, \lean{S$\sigma$},
and \lean{RootPlantable}; the two orders of \Cref{def:two-constrained-orders}
appear as the non-negativity predicates \lean{QuotientNonneg} and
\lean{EnsembleNonneg}, with \(f\leq g\) expressed by applying them to \(g-f\).
The theorems are \lean{quotient\_implies\_ensemble} and
\lean{support\_criterion} for the two halves of
\Cref{thm:support-closure-summary}, \lean{blowupClosed\_root\_plantable} and
\lean{heredClass\_emptyType\_rootPlantable} for the nonempty- and empty-type
claims of \Cref{thm:blowup-root-plantability-summary}, and
\lean{c4free\_not\_rootPlantable} with \lean{S$\sigma$\_subset\_Q$\sigma$} for
the strict inclusion \(S_{\sigma_1}\subsetneq Q_{\sigma_1}\).  All of them
depend only on the standard axioms of Lean's logic.  The density theorems of
this paper do not depend on this part of the development: their certificate
proofs reach the empty type through \lean{downward\_forbidLEWith\_nonneg},
which \lean{forbidLEWith\_add\_QuadraticForm} of \Cref{sec:flagmatic-shape}
packages, and conclude with \lean{generalizedTuranDensity\_le\_of\_forbidLE}
of \Cref{sec:forbidden}.

These files were generated by Claude Code from an already-developed
mathematical account and its theorem targets, and Lean checks the resulting
proof terms.  The autoformalization was feasible because it extended rather
than rebuilt the manually written development of the preceding sections, which
had already fixed the representations and semantic interfaces and supplied the
ambient algebra, the positive-homomorphism spaces, the random extensions, and
the downward operator; only the constrained layer, the forbidden ideal and its
quotient, had to be built.

\section{Open Design Questions}
\label{sec:unsure}

The formalization and certificate proofs presented above are fully
checked in Lean.  This section instead concerns design questions that the
present development cannot settle on its own.  For each question, we state
what our development establishes and what remains unknown because answering
it would require either building a competing development or running
substantially larger examples.

\paragraph{When should a forbidden-graph constraint enter the formalization?}
In our ensemble semantics (\Cref{sec:forbidden}), we defer the constraint until
elements are compared.  Consequently, the seven certificates of
\Cref{sec:results}, involving four different forbidden graphs, share a single
ambient algebra and a common library of algebraic lemmas.  This choice does not
affect the validity of the density bounds: the ensemble semantics yields the
asymptotic bounds directly.  How it compares with Razborov's standard semantics
is a separate mathematical question, studied in \Cref{sec:metatheory}, and the
bounds do not depend on its answer.

This does not determine which design is more economical overall.  Building the
constraint in from the outset may streamline arguments for a fixed extremal
problem, but we have not formalized comparable certificate proofs using that
design.  Even when root-plantability ensures that the two semantics validate
the same inequalities, their formalizations may differ in effort and
computational cost.

\paragraph{Should specification and computation use separate representations?}
Using separate representations for specification and computation kept the
abstract development close to standard mathematics while making finite
calculations executable.  In all seven
cases, the density facts were obtained by evaluation and transferred to the
specification layer by adequacy theorems, and the multiplication facts were
derived from them; for the two largest cases, the bit-mask encoding of
\Cref{sec:reflection-bitmask} adds a third representation joined to the
second by further proved bridges.  Maintaining these bridges required
substantial work.  An alternative would be to formulate the theory directly
using a concrete, canonical representation.  This could reduce the bridging
work, but choices about representation and canonicalization would then appear
in the statements and proofs of the mathematical theory.  Determining which
approach is less costly would require a second development of comparable
scope.  Our separation follows the refinement
methodology of Cohen et al.~\cite{cohen2013refinements}, and it appears
independently in the concurrent formalization of local flag
algebras~\cite{davey2026local}, which likewise separates a noncomputable
specification from a computable representation.

\paragraph{How far does the computational layer scale?}
The seven certificates of \Cref{sec:results} use hosts on at most five
vertices.  At that size, the enumeration of \Cref{sec:reflection-generation}
and the bit-mask encoding of \Cref{sec:reflection-bitmask} keep the finite
identities manageable: the Erd\H{o}s pentagon certificate needs 2,832
pair-density identities, and the $K_5$-free and $C_5$-free edge-density
certificates need 10,332 and 8,172, all evaluated in the kernel.  The cost
grows quickly with the host size.  The number of flags, the generated
identities, and the expressions that the final normalization must handle all
increase, and so do the time and memory of kernel evaluation
(\Cref{sec:compile-times}).  The bit-mask sweeps are designed to push this
frontier; in the released development, the canonicalization sweep of
\Cref{sec:reflection-bitmask} already covers all graphs on six vertices, one
more than the seven certificates need.  Certificate generation faces the same combinatorial
growth: the user's guide for Flagmatic~1.5 documents support for bases of up
to eight vertices for ordinary graphs and seven for
3-graphs~\cite{vaughan2012flagmaticguide}.  Certificates of these sizes are
the natural target for verification, whereas this paper reports certificates
with hosts of at most five vertices; extending the evaluation to larger hosts
is ongoing work.

\section{Related Work}
\label{sec:related}

\paragraph{Formalizations of extremal combinatorics.}
Mathlib~\cite{mathlib2020} contains Tur\'an's theorem, together with the
definitions of extremal numbers and Tur\'an densities that our statement of
Mantel's theorem uses.  Tur\'an's theorem implies the finite form of Mantel's
theorem and the $K_4$-free and $K_5$-free edge-density bounds of
\Cref{sec:results}; for these cases our certificates test the flag-algebra
machinery rather than add new formal results.  Other central results of
extremal combinatorics have also been formalized in Lean: Dillies and Mehta
formalized Szemer\'edi's regularity lemma and the qualitative form of Roth's
theorem~\cite{dillies2022szemeredi}, and Mehta the Kruskal--Katona
theorem~\cite{mehta2022kruskal}.  Those formalizations are each organized
around a headline theorem, contributing reusable infrastructure along the way;
ours is aimed at a method, formalizing the flag algebra, its semantics, and the
finite machinery its certificates need.

\paragraph{Concurrent flag-algebra work in Lean.}
Two Lean developments run concurrently with, and independently of, ours.
Davey, Hurley, de Joannis de Verclos, Kang, and Volec formalized \emph{local
flag algebras}, the calculus introduced in their accompanying
paper~\cite{davey2026local} and applied in its companion~\cite{davey2026strong}.
It normalizes densities by the maximum degree $\Delta$ rather than by the
number of vertices, so that the method still yields nontrivial bounds for
sparse graphs of small maximum degree, and under this normalization no quotient
is taken: their algebra is defined directly on formal combinations of
isomorphism classes.  Their formalization machine-checks the new bounds of both
papers, while recording the steps that bridge to SDP certificates and to cited
literature, required by some bounds, as named assumptions documented result by
result.  It also independently arrives at the reflection pattern of
\Cref{sec:reflection}: a noncomputable specification mirrored by a computable
representation (Boolean adjacency in their case) with adequacy lemmas.  Our
development instead formalizes the classical theory (the quotient algebra,
positive homomorphisms, random extensions, and the ensemble semantics) as a
reusable library, with certificate checking built on top.  The other is
Spiegel's \emph{computational flag algebras} in
Lean~\cite{spiegel2026computational}, ongoing independent work aimed at
carrying out flag-algebra inequality proofs automatically inside Lean.

\paragraph{Computer-assisted flag-algebra arguments.}
Flagmatic~\cite{vaughan2020flagmatic} automates the computation of exact and
approximate flag-algebra bounds but does not produce proof-assistant proofs;
Razborov's survey~\cite{razborov2013flag} organizes many early applications of
the method.  The seven certificates verified here are Flagmatic's output as
produced.

\paragraph{Formal verification of computer-found evidence.}
Many formalizations of computer-assisted mathematics follow a common pattern:
an unverified search produces evidence that the proof assistant then checks.
Prominent examples are the four-colour theorem~\cite{gonthier2008fourcolour},
the Kepler conjecture~\cite{hales2017kepler}, and the empty-hexagon
number~\cite{heule2024hexagon,subercaseaux2024hexagon}.  Closest to our setting
are sums-of-squares certificates for polynomial inequalities, which Harrison
checks in HOL Light~\cite{harrison2007verifying}.  Morrison's \lean{sos}
tactic in Lean~4~\cite{morrison2026sos} uses an external semidefinite solver
to find a floating-point Gram matrix, rounds it to rational data, and verifies
the resulting sums-of-squares decomposition through a rational $LDL^\top$
factorization, the same check that we apply to the matrices of a flag-algebra
certificate.  Proof by reflection, which we use for the finite combinatorial
identities, has a long history, from Boutin~\cite{boutin1997reflection} to the
treatment by Chlipala~\cite{chlipala2013cpdt}, and Cohen et
al.~\cite{cohen2013refinements} develop a refinement methodology relating
representations suited to mathematical reasoning to more efficient
representations suited to computation.

\paragraph{Graph limits.}
Graphon theory provides an alternative analytic framework for the asymptotic
graph-density phenomena studied by flag algebras; see Lov\'asz and Szegedy~\cite{lovasz2006limits} and
Lov\'asz~\cite{lovasz2012large}.  Freer has formalized a
substantial graphon theory in Lean~4, covering cut distance, weak regularity,
counting lemmas, and compactness~\cite{freer2026graphon}.  Our development
formalizes the flag-algebra side of this relationship rather than graphons.
Formally connecting our positive-homomorphism semantics to Freer's graphon
formalization remains future work.

\section{Conclusion}
\label{sec:conclusion}

We have formalized Razborov's flag algebra method for finite simple graphs in
Lean~4: the quotient algebra of density expressions, positive homomorphisms
and their correspondence with limits of graph sequences, random extensions,
and the downward operator.  A reflection layer makes the finite parts of the
theory computable, joined to the abstract definitions by adequacy theorems,
and a tactic turns Flagmatic certificates into Lean proofs while checking every
fact they use.  With this development we verified seven Tur\'an-type upper
bounds and proved the asymptotic forms of Mantel's theorem and of the
Erd\H{o}s pentagon theorem, together with two inequalities of Goodman; all of
these depend only on Lean's standard axioms.

One design choice shaped the development throughout.  In the standard
semantics, the forbidden-subgraph constraint is built into a separate algebra
for each forbidden family.  We instead kept one ambient algebra and imposed the
constraint only when its elements are compared.  This ensemble semantics keeps
the core definitions independent of the extremal problem being solved, and it
yields the intended asymptotic bounds on its own terms.  Comparing the two
semantics led to mathematics of its own: validity in the standard semantics
always implies validity in the ensemble semantics, and the converse holds for
all inequalities at a flag type exactly when the class is root-plantable there.

Several directions remain for future work.  The most immediate is to verify
larger certificates, up to hosts on eight vertices, the largest that
Flagmatic~1.5 supported for graphs.  Another is to extend the development to
further parts of Razborov's calculus, especially the differential structure
developed in Section~4.3 of his original paper~\cite{razborov2007flag}, which
would exercise the quotient-level semantics well beyond the certificate
arguments considered here.  A third is to adapt the formalization to richer
combinatorial structures such as hypergraphs and directed graphs.  A fourth is
to build a formal bridge between our positive-homomorphism semantics and
graphon theory, connecting the development to Lov\'asz's graph-limit
theory~\cite{lovasz2012large} and the existing Lean formalization of
graphons~\cite{freer2026graphon}.

\subsection*{Use of AI Assistance}
The parts of this development used AI assistance to different degrees.  The
specification and reflection layers of \Cref{sec:abstract,sec:reflection},
including the quotient constructions, their interface theorems, and the
adequacy proofs, were formalized manually by the authors.  Most of the
metaprogramming, including the generation commands of
\Cref{sec:reflection-generation}, the certificate-checking tactic of
\Cref{sec:flagmatic}, and the script that writes the certificate files, was
developed with the help of Claude Code.  Finally, the mathematical results of
the meta-theory summarized in \Cref{sec:metatheory} were obtained with the help
of GPT~Pro, and their Lean development was autoformalized entirely by Claude
Code.  We also used Codex and Claude Code to assist with writing and revising
this paper.

\section*{Acknowledgments}
We would like to thank Jineon Baek, Taeyoung Kim, Hyunwoo Lee, Joonkyung Lee,
Christoph Spiegel, and Jan Volec for helpful discussions on Razborov's flag
algebras and their Lean formalization.  We are also grateful to Ross Kang and
Sidharth Hariharan for encouraging us to remove \lean{native\_decide} from our
Lean formalization.
This work was supported by the National Research Foundation of Korea (NRF)
grant funded by the Korean Government (MSIT) (No.~RS-2023-00279680), the
Institute for Basic Science (IBS-R029-C1), and
a KIAS Individual Grant (CG109601) at the Korea Institute for Advanced Study.

\bibliographystyle{plainurl}
\bibliography{refs}

\appendix
\section{Compilation Times}
\label{sec:compile-times}

Table~\ref{tab:compile-times} reports compilation time and peak memory for the
seven certificate files of \Cref{sec:results}, compiled with Lean~4 (v4.27.0)
on an Intel Core i7-14700K (20 cores) with 64\,GB of RAM under Windows~11.
Each measurement covers the complete source file: flag enumeration,
generation of the pair-density and multiplication lemmas, the PSD checks of
the certificate, the objective expansion, and the main theorem.  Every
generated lemma is proved by \lean{decide +kernel}; the $K_5$- and $C_5$-free
files use the bit-mask routes of \Cref{sec:reflection-bitmask}.

Each file was compiled five times on its own against the prebuilt library
(\lean{lake env lean}).  The seven files and their certificates are
byte-for-byte those of the artifact (\Cref{sec:artifact}).
Time is wall-clock time, and memory is the peak working set of the Lean
process, taken as the maximum over the five runs.  The table reports the median time and the range of the five runs.  The range is
within $4.1\%$ of the median in every case (at most $0.6$\,s for the two
smallest files) and within $1\%$ for the three largest.

\begin{table}[ht]
  \centering
  \caption{Lean compilation time (median and range of five runs) and peak
    memory (maximum over five runs) for the seven certificate files, all
    checked by \lean{decide +kernel}.  ``Lemmas'' is the number of generated
    \lean{flagDensity}$_2$ lemmas; the shared subset-pair pass of the
    $K_5$- and $C_5$-free files additionally generates $287$ and $227$
    per-host key theorems.}
  \label{tab:compile-times}
  \small
  \begin{tabular}{lccr rr r}
    \toprule
    Case & Forbid & $N$ & Lemmas & Time (s) & Range (s) & RAM (GB) \\
    \midrule
    Mantel (edge density)              & $K_3$ & 3 &    15 & 16.0  & 15.9--16.5   &  2.4 \\
    $P_3$ density                      & $K_3$ & 3 &    15 & 15.9  & 15.8--16.0   &  2.4 \\
    $C_4$ density                      & $K_3$ & 4 &   210 & 75.9  & 75.1--76.8   &  8.1 \\
    Edge density                       & $K_4$ & 4 &   390 & 133.4 & 132.8--135.1 & 11.2 \\
    Erd\H{o}s pentagon ($C_5$ density) & $K_3$ & 5 &  2832 & 2659  & 2655--2680   & 42.3 \\
    Edge density                       & $K_5$ & 5 & 10332 & 1583  & 1579--1585   & 29.5 \\
    Edge density                       & $C_5$ & 5 &  8172 & 1062  & 1058--1067   & 26.7 \\
    \bottomrule
  \end{tabular}
\end{table}

\FloatBarrier
\section{The Artifact}
\label{sec:artifact}

\paragraph{Availability.}
The artifact is the tag \texttt{v1.0} of
\url{https://github.com/taeyool/lean-flag-algebras-release}, at commit
\texttt{439d3e64e02ff044a5e8b2b01937f8d7f36c84b6}, distributed under the Apache
License~2.0.  Software Heritage archived it on 3 October 2026; its identifiers
are \href{https://archive.softwareheritage.org/swh:1:rev:439d3e64e02ff044a5e8b2b01937f8d7f36c84b6;origin=https://github.com/taeyool/lean-flag-algebras-release}{\nolinkurl{swh:1:rev:439d3e64e02ff044a5e8b2b01937f8d7f36c84b6}} for the commit
and \href{https://archive.softwareheritage.org/swh:1:rel:599f924360e59565847c957d5f7562d88ef949d1;origin=https://github.com/taeyool/lean-flag-algebras-release}{\nolinkurl{swh:1:rel:599f924360e59565847c957d5f7562d88ef949d1}} for the tag.

\paragraph{Building.}
The artifact pins Lean~4 v4.27.0 in \texttt{lean-toolchain} and Mathlib at
commit \texttt{a3a10db0e9d6} (tag \texttt{v4.27.0}) in
\texttt{lake-manifest.json}.  After \lean{lake exe cache get} fetches the
compiled Mathlib, \lean{lake build} checks the whole development.  Because
every finite identity is evaluated in the kernel, some files need tens of
gigabytes of memory, and Lake cannot limit how many files it compiles at once.
The six heaviest library modules (26--37\,GB each) therefore import one
another in sequence, so that Lake builds them one at a time, and on a machine
with 64\,GB or less the three largest certificate files should be built one at
a time before the rest:
\begin{lstlisting}
lake build LeanFlagAlgebras.Flagmatic.ErdosPentagon
lake build LeanFlagAlgebras.Flagmatic.K5freeEdge
lake build LeanFlagAlgebras.Flagmatic.C5freeEdge
lake build
\end{lstlisting}
From a clean state with Mathlib prebuilt, this sequence took 157 minutes on
the machine of \Cref{sec:compile-times} (8,083 build jobs), with a peak of
43.3\,GB for a single file.  Individual results can be checked by building
their modules, e.g., \lean{lake build LeanFlagAlgebras.Flagmatic.Mantel}.  The
tactic \lean{flag\_certificate} reads the certificate file at elaboration
time, and Lake does not track this file; after editing a certificate, run
\lean{lake clean} or modify the Lean file that reads it before rebuilding.

\paragraph{What is checked.}
\Cref{tab:statement-declaration} lists the statements of this paper and their
corresponding Lean declarations.
The library contains no \lean{sorry} and declares no axioms.  At the pinned
release, axiom checks on the density theorems and main meta-theory theorems
reported here list only \lean{propext}, \lean{Classical.choice}, and
\lean{Quot.sound}.
The certificates are in \texttt{LeanFlagAlgebras/Flagmatic/Certificates/};
each has the same basename as its Lean module and the suffix
\texttt{\_cert.json}.
\Cref{tab:size} reports the size of the development.

{\small\setlength{\LTcapwidth}{\linewidth}
\begin{longtable}{@{}p{0.15\linewidth}p{0.25\linewidth}>{\raggedright\arraybackslash}p{0.28\linewidth}>{\raggedright\arraybackslash}p{0.24\linewidth}@{}}
\caption{Statements of the paper and the Lean declarations that formalize them, at \texttt{v1.0} (commit \texttt{439d3e64e02f}). Locations are relative to \texttt{LeanFlagAlgebras/} and link to the tagged source.}\label{tab:statement-declaration}\\
\toprule
Paper & Statement & Lean declaration & Location \\
\midrule
\endfirsthead
\toprule
Paper & Statement & Lean declaration & Location \\
\midrule
\endhead
\bottomrule
\endfoot
    \Cref{lem:chain-rule} & Chain rule & \lean{flagDensity\_eq\_sum\_density\_prods} & \href{https://github.com/taeyool/lean-flag-algebras-release/blob/439d3e64e02ff044a5e8b2b01937f8d7f36c84b6/LeanFlagAlgebras/FlagAlgebra/SubflagListDensity.lean\#L2999}{\nolinkurl{FlagAlgebra/SubflagListDensity.lean:2999}} \\
    \Cref{sec:background-densities} & Product independent of the auxiliary size & \lean{flagMulWithSize\_indep\_on\_size} & \href{https://github.com/taeyool/lean-flag-algebras-release/blob/439d3e64e02ff044a5e8b2b01937f8d7f36c84b6/LeanFlagAlgebras/FlagAlgebra/FlagAlgebra.lean\#L542}{\nolinkurl{FlagAlgebra/FlagAlgebra.lean:542}} \\
    \Cref{thm:convergent-hom}(a) & Limits of convergent sequences are positive homomorphisms & \lean{flagSeq\_limit\_mem\_positiveHom} & \href{https://github.com/taeyool/lean-flag-algebras-release/blob/439d3e64e02ff044a5e8b2b01937f8d7f36c84b6/LeanFlagAlgebras/FlagAlgebra/FlagSequence.lean\#L656}{\nolinkurl{FlagAlgebra/FlagSequence.lean:656}} \\
    \Cref{thm:convergent-hom}(b) & Every positive homomorphism is such a limit & \lean{positiveHom\_as\_flagSeq\_limit} & \href{https://github.com/taeyool/lean-flag-algebras-release/blob/439d3e64e02ff044a5e8b2b01937f8d7f36c84b6/LeanFlagAlgebras/FlagAlgebra/FlagSequence.lean\#L1019}{\nolinkurl{FlagAlgebra/FlagSequence.lean:1019}} \\
    \Cref{sec:background-downward} & Random extension of an empty-type homomorphism & \lean{exists\_probMeasure\_extend\_emptyType\_positiveHom} & \href{https://github.com/taeyool/lean-flag-algebras-release/blob/439d3e64e02ff044a5e8b2b01937f8d7f36c84b6/LeanFlagAlgebras/FlagAlgebra/RandomHom.lean\#L1096}{\nolinkurl{FlagAlgebra/RandomHom.lean:1096}} \\
    \Cref{thm:downward-nonneg} & The downward operator preserves non-negativity & \lean{downward\_preserve\_semanticCone} & \href{https://github.com/taeyool/lean-flag-algebras-release/blob/439d3e64e02ff044a5e8b2b01937f8d7f36c84b6/LeanFlagAlgebras/FlagAlgebra/RandomHom.lean\#L1229}{\nolinkurl{FlagAlgebra/RandomHom.lean:1229}} \\
    \Cref{thm:downward-nonneg} & The same in the $H$-free order used by certificates & \lean{downward\_forbidLEWith\_nonneg} & \href{https://github.com/taeyool/lean-flag-algebras-release/blob/439d3e64e02ff044a5e8b2b01937f8d7f36c84b6/LeanFlagAlgebras/Forbid/Basic.lean\#L1505}{\nolinkurl{Forbid/Basic.lean:1505}} \\
    \Cref{sec:lowerbounds} & Cauchy--Schwarz for the downward operator & \lean{Cauchy\_Schwarz\_inequality} \newline \lean{Cauchy\_Schwarz\_inequality\_unit} & \href{https://github.com/taeyool/lean-flag-algebras-release/blob/439d3e64e02ff044a5e8b2b01937f8d7f36c84b6/LeanFlagAlgebras/FlagAlgebra/RandomHom.lean\#L1419}{\nolinkurl{FlagAlgebra/RandomHom.lean:1419}} \newline \href{https://github.com/taeyool/lean-flag-algebras-release/blob/439d3e64e02ff044a5e8b2b01937f8d7f36c84b6/LeanFlagAlgebras/FlagAlgebra/RandomHom.lean\#L1430}{\nolinkurl{FlagAlgebra/RandomHom.lean:1430}} \\
    \Cref{sec:forbidden} & The Tur\'an density is a limit & \lean{tendsto\_generalizedTuranDensity} & \href{https://github.com/taeyool/lean-flag-algebras-release/blob/439d3e64e02ff044a5e8b2b01937f8d7f36c84b6/LeanFlagAlgebras/Turan/GeneralizedTuran.lean\#L253}{\nolinkurl{Turan/GeneralizedTuran.lean:253}} \\
    \Cref{lem:semantic-bound} & A semantic inequality bounds the Tur\'an density & \lean{generalizedTuranDensity\_le\_of\_forbidLE} & \href{https://github.com/taeyool/lean-flag-algebras-release/blob/439d3e64e02ff044a5e8b2b01937f8d7f36c84b6/LeanFlagAlgebras/Forbid/TuranDensity.lean\#L250}{\nolinkurl{Forbid/TuranDensity.lean:250}} \\
    \Cref{def:ensemble-semantic-order} & Ensemble semantic order & \lean{forbidLE} & \href{https://github.com/taeyool/lean-flag-algebras-release/blob/439d3e64e02ff044a5e8b2b01937f8d7f36c84b6/LeanFlagAlgebras/Forbid/Basic.lean\#L101}{\nolinkurl{Forbid/Basic.lean:101}} \\
    \Cref{sec:reflection-adequacy} & Adequacy of the computable densities & \lean{flagDensity$_1$\_eq\_sym2FlagDensity$_1$} \newline \lean{flagDensity$_2$\_eq\_sym2FlagDensity$_2$} & \href{https://github.com/taeyool/lean-flag-algebras-release/blob/439d3e64e02ff044a5e8b2b01937f8d7f36c84b6/LeanFlagAlgebras/FlagAlgebra/Compute/FlagDensity.lean\#L1140}{\nolinkurl{FlagAlgebra/Compute/FlagDensity.lean:1140}} \newline \href{https://github.com/taeyool/lean-flag-algebras-release/blob/439d3e64e02ff044a5e8b2b01937f8d7f36c84b6/LeanFlagAlgebras/FlagAlgebra/Compute/FlagDensity.lean\#L1213}{\nolinkurl{FlagAlgebra/Compute/FlagDensity.lean:1213}} \\
    \Cref{sec:reflection-adequacy} & Adequacy of the downward factors & \lean{downwardNormalizingFactor\_eq} & \href{https://github.com/taeyool/lean-flag-algebras-release/blob/439d3e64e02ff044a5e8b2b01937f8d7f36c84b6/LeanFlagAlgebras/FlagAlgebra/Compute/Downward.lean\#L170}{\nolinkurl{FlagAlgebra/Compute/Downward.lean:170}} \\
    \Cref{sec:reflection-adequacy} & Densities are invariant under relabeling & \lean{flagDensity\_permute} & \href{https://github.com/taeyool/lean-flag-algebras-release/blob/439d3e64e02ff044a5e8b2b01937f8d7f36c84b6/LeanFlagAlgebras/FlagAlgebra/SubflagListDensity.lean\#L649}{\nolinkurl{FlagAlgebra/SubflagListDensity.lean:649}} \\
    \Cref{sec:flagmatic-background} & PSD from a rational $LDL^{\top}$ factorization & \lean{posSemidef\_real\_of\_LDLt} & \href{https://github.com/taeyool/lean-flag-algebras-release/blob/439d3e64e02ff044a5e8b2b01937f8d7f36c84b6/LeanFlagAlgebras/Automation/Matrix/PosSemiDef.lean\#L73}{\nolinkurl{Automation/Matrix/PosSemiDef.lean:73}} \\
    \Cref{sec:flagmatic-background} & Adding a downward quadratic form & \lean{forbidLEWith\_add\_QuadraticForm} & \href{https://github.com/taeyool/lean-flag-algebras-release/blob/439d3e64e02ff044a5e8b2b01937f8d7f36c84b6/LeanFlagAlgebras/Automation/Basic.lean\#L128}{\nolinkurl{Automation/Basic.lean:128}} \\
    \Cref{tab:flagmatic-scenarios} & Edge density, $K_3$-free & \lean{Mantel\_turanDensity} & \href{https://github.com/taeyool/lean-flag-algebras-release/blob/439d3e64e02ff044a5e8b2b01937f8d7f36c84b6/LeanFlagAlgebras/Flagmatic/Mantel.lean\#L72}{\nolinkurl{Flagmatic/Mantel.lean:72}} \\
    \Cref{tab:flagmatic-scenarios} & $P_3$ density, $K_3$-free & \lean{K3freeP3\_turanDensity} & \href{https://github.com/taeyool/lean-flag-algebras-release/blob/439d3e64e02ff044a5e8b2b01937f8d7f36c84b6/LeanFlagAlgebras/Flagmatic/K3freeP3.lean\#L71}{\nolinkurl{Flagmatic/K3freeP3.lean:71}} \\
    \Cref{tab:flagmatic-scenarios} & $C_4$ density, $K_3$-free & \lean{K3freeC4\_turanDensity} & \href{https://github.com/taeyool/lean-flag-algebras-release/blob/439d3e64e02ff044a5e8b2b01937f8d7f36c84b6/LeanFlagAlgebras/Flagmatic/K3freeC4.lean\#L75}{\nolinkurl{Flagmatic/K3freeC4.lean:75}} \\
    \Cref{tab:flagmatic-scenarios} & Edge density, $K_4$-free & \lean{K4freeEdge\_turanDensity} & \href{https://github.com/taeyool/lean-flag-algebras-release/blob/439d3e64e02ff044a5e8b2b01937f8d7f36c84b6/LeanFlagAlgebras/Flagmatic/K4freeEdge.lean\#L77}{\nolinkurl{Flagmatic/K4freeEdge.lean:77}} \\
    \Cref{tab:flagmatic-scenarios} & $C_5$ density, $K_3$-free & \lean{ErdosPentagon\_turanDensity} & \href{https://github.com/taeyool/lean-flag-algebras-release/blob/439d3e64e02ff044a5e8b2b01937f8d7f36c84b6/LeanFlagAlgebras/Flagmatic/ErdosPentagon.lean\#L79}{\nolinkurl{Flagmatic/ErdosPentagon.lean:79}} \\
    \Cref{tab:flagmatic-scenarios} & Edge density, $K_5$-free & \lean{K5freeEdge\_turanDensity} & \href{https://github.com/taeyool/lean-flag-algebras-release/blob/439d3e64e02ff044a5e8b2b01937f8d7f36c84b6/LeanFlagAlgebras/Flagmatic/K5freeEdge.lean\#L90}{\nolinkurl{Flagmatic/K5freeEdge.lean:90}} \\
    \Cref{tab:flagmatic-scenarios} & Edge density, $C_5$-free & \lean{C5freeEdge\_turanDensity} & \href{https://github.com/taeyool/lean-flag-algebras-release/blob/439d3e64e02ff044a5e8b2b01937f8d7f36c84b6/LeanFlagAlgebras/Flagmatic/C5freeEdge.lean\#L97}{\nolinkurl{Flagmatic/C5freeEdge.lean:97}} \\
    \Cref{ex:mantel} & Asymptotic Mantel theorem & \lean{Mantel\_Turan} & \href{https://github.com/taeyool/lean-flag-algebras-release/blob/439d3e64e02ff044a5e8b2b01937f8d7f36c84b6/LeanFlagAlgebras/MantelTheorem/MantelTheorem.lean\#L191}{\nolinkurl{MantelTheorem/MantelTheorem.lean:191}} \\
    \Cref{ex:pentagon} & Asymptotic Erd\H{o}s pentagon theorem & \lean{ErdosPentagon\_Turan} & \href{https://github.com/taeyool/lean-flag-algebras-release/blob/439d3e64e02ff044a5e8b2b01937f8d7f36c84b6/LeanFlagAlgebras/ErdosPentagon/ErdosPentagon.lean\#L313}{\nolinkurl{ErdosPentagon/ErdosPentagon.lean:313}} \\
    \Cref{sec:lowerbounds} & Pentagon lower bound & \lean{ErdosPentagon\_Turan\_lowerBound} & \href{https://github.com/taeyool/lean-flag-algebras-release/blob/439d3e64e02ff044a5e8b2b01937f8d7f36c84b6/LeanFlagAlgebras/ErdosPentagon/ErdosPentagon.lean\#L291}{\nolinkurl{ErdosPentagon/ErdosPentagon.lean:291}} \\
    \Cref{sec:lowerbounds} & $C_5$ and the certificate's target agree & \lean{C5\_toFlagAlgebra\_eq\_certificate} & \href{https://github.com/taeyool/lean-flag-algebras-release/blob/439d3e64e02ff044a5e8b2b01937f8d7f36c84b6/LeanFlagAlgebras/ErdosPentagon/ErdosPentagon.lean\#L26}{\nolinkurl{ErdosPentagon/ErdosPentagon.lean:26}} \\
    \Cref{sec:lowerbounds} & Goodman's triangle bound & \lean{Goodman\_bound\_on\_triangle\_density} & \href{https://github.com/taeyool/lean-flag-algebras-release/blob/439d3e64e02ff044a5e8b2b01937f8d7f36c84b6/LeanFlagAlgebras/MantelTheorem/GoodmanBound.lean\#L15}{\nolinkurl{MantelTheorem/GoodmanBound.lean:15}} \\
    \Cref{sec:lowerbounds} & Goodman's Ramsey multiplicity bound & \lean{Goodman\_theorem\_on\_Ramsey\_multiplicity} & \href{https://github.com/taeyool/lean-flag-algebras-release/blob/439d3e64e02ff044a5e8b2b01937f8d7f36c84b6/LeanFlagAlgebras/MantelTheorem/GoodmanRamsey.lean\#L16}{\nolinkurl{MantelTheorem/GoodmanRamsey.lean:16}} \\
    \Cref{def:two-constrained-orders} & Quotient and ensemble orders & \lean{QuotientNonneg} \newline \lean{EnsembleNonneg} & \href{https://github.com/taeyool/lean-flag-algebras-release/blob/439d3e64e02ff044a5e8b2b01937f8d7f36c84b6/LeanFlagAlgebras/MetaTheory/SupportClosure.lean\#L120}{\nolinkurl{MetaTheory/SupportClosure.lean:120}} \newline \href{https://github.com/taeyool/lean-flag-algebras-release/blob/439d3e64e02ff044a5e8b2b01937f8d7f36c84b6/LeanFlagAlgebras/MetaTheory/SupportClosure.lean\#L125}{\nolinkurl{MetaTheory/SupportClosure.lean:125}} \\
    \Cref{thm:support-closure-summary} & Support-closure criterion & \lean{quotient\_implies\_ensemble} \newline \lean{support\_criterion} & \href{https://github.com/taeyool/lean-flag-algebras-release/blob/439d3e64e02ff044a5e8b2b01937f8d7f36c84b6/LeanFlagAlgebras/MetaTheory/SupportClosure.lean\#L130}{\nolinkurl{MetaTheory/SupportClosure.lean:130}} \newline \href{https://github.com/taeyool/lean-flag-algebras-release/blob/439d3e64e02ff044a5e8b2b01937f8d7f36c84b6/LeanFlagAlgebras/MetaTheory/SupportClosure.lean\#L139}{\nolinkurl{MetaTheory/SupportClosure.lean:139}} \\
    \Cref{thm:blowup-root-plantability-summary} & Blow-up closure implies root-plantability & \lean{blowupClosed\_root\_plantable} \newline \lean{heredClass\_emptyType\_rootPlantable} & \href{https://github.com/taeyool/lean-flag-algebras-release/blob/439d3e64e02ff044a5e8b2b01937f8d7f36c84b6/LeanFlagAlgebras/MetaTheory/BlowupClosed.lean\#L530}{\nolinkurl{MetaTheory/BlowupClosed.lean:530}} \newline \href{https://github.com/taeyool/lean-flag-algebras-release/blob/439d3e64e02ff044a5e8b2b01937f8d7f36c84b6/LeanFlagAlgebras/MetaTheory/EmptyTypeCollapse.lean\#L180}{\nolinkurl{MetaTheory/EmptyTypeCollapse.lean:180}} \\
    \Cref{sec:metatheory} & The $C_4$-free class is not root-plantable & \lean{c4free\_not\_rootPlantable} \newline \lean{S$\sigma$\_subset\_Q$\sigma$} & \href{https://github.com/taeyool/lean-flag-algebras-release/blob/439d3e64e02ff044a5e8b2b01937f8d7f36c84b6/LeanFlagAlgebras/MetaTheory/C4Free.lean\#L455}{\nolinkurl{MetaTheory/C4Free.lean:455}} \newline \href{https://github.com/taeyool/lean-flag-algebras-release/blob/439d3e64e02ff044a5e8b2b01937f8d7f36c84b6/LeanFlagAlgebras/MetaTheory/SupportClosure.lean\#L114}{\nolinkurl{MetaTheory/SupportClosure.lean:114}} \\
\end{longtable}
}

\begin{table}[ht]
\centering
\small
\caption{Size of the development at \texttt{v1.0}, by component: Lean files, lines, and source declarations (theorems and lemmas; definitions, abbreviations, instances, structures and classes). Declarations emitted by the generation commands are not counted here.}
\label{tab:size}
\begin{tabular}{@{}lrrrr@{}}
\toprule
Component & Files & Lines & Theorems & Definitions \\
\midrule
    Specification (flags, densities, flag algebras) & 19 & 17,147 & 534 & 329 \\
    Reflection (computable densities, generators, bit masks) & 51 & 14,950 & 528 & 425 \\
    Certificate checking and the seven certificates & 18 & 3,935 & 52 & 134 \\
    Forbidden subgraphs and Tur\'an densities & 4 & 2,532 & 134 & 24 \\
    Mantel, pentagon, and Goodman theorems & 9 & 1,494 & 66 & 7 \\
    Meta-theory & 96 & 35,682 & 1,219 & 316 \\
\midrule
    Total & 197 & 75,740 & 2,533 & 1,235 \\
\bottomrule
\end{tabular}
\end{table}

\end{document}